\documentclass[a4paper,twocolumn,11pt,accepted=2021-05-31]{quantumarticle}
\pdfoutput=1
\usepackage{env}

\begin{document}

\title{Witnessing Wigner Negativity}

\author{Ulysse Chabaud\texorpdfstring{$^*$}{*}}
\email{uchabaud@caltech.edu}
\orcid{0000-0003-0135-9819}
\affiliation{Institute 
for Quantum Information and Matter, Caltech}
\affiliation{Universit\'e de Paris, IRIF, CNRS, France}

\author{Pierre-Emmanuel Emeriau\texorpdfstring{$^*$}{*}}
\email{pierre-emmanuel.emeriau@lip6.fr}
\orcid{0000-0001-5155-1783}
\affiliation{Sorbonne Université, CNRS, LIP6, F-75005 Paris, France}
\author{Frédéric Grosshans}
\email{frederic.grosshans@lip6.fr}
\affiliation{Sorbonne Université, CNRS, LIP6, F-75005 Paris, France}
\orcid{0000-0001-8170-9668}

\newcommand{\pe}[1]{{\color{red}{#1}}}

\maketitle

\begin{abstract}
\blfootnote{\hspace{-17pt}\vspace{-11pt}$^*$ These authors contributed equally.}
Negativity of the Wigner function is arguably one of the most striking non-classical features of quantum states. Beyond its fundamental relevance, it is also a necessary resource for quantum speedup with continuous variables. As quantum technologies emerge, the need to identify and characterize the resources which provide an advantage over existing classical technologies becomes more pressing. 
Here we derive witnesses for Wigner negativity of single-mode and multimode quantum states, based on fidelities with Fock states, which can be reliably measured using standard detection setups. They possess a threshold expectation value indicating whether the measured state has a negative Wigner function. Moreover, the amount of violation provides an operational quantification of Wigner negativity. 
We phrase the problem of finding the threshold values for our witnesses as an infinite-dimensional linear optimisation problem. By relaxing and restricting the corresponding linear programs, we derive two converging hierarchies of semidefinite programs, which provide numerical sequences of increasingly tighter upper and lower bounds for the threshold values. We further show that our witnesses form a complete family---each Wigner negative state is detected by at least one witness---thus providing a reliable method for experimentally witnessing Wigner negativity of quantum states from few measurements. From a foundational perspective, our findings provide insights on the set of positive Wigner functions which still lacks a proper characterisation.

\end{abstract}

\section{Introduction}
\label{sec:intro}

Quantum information with continuous variables~\cite{lloyd1999quantum}---where information is encoded in continuous degrees of freedom of quantum systems---is one of the promising directions for the future of quantum technologies. For example, continuous-variable quantum optics enables the deterministic experimental preparation of entangled states over millions of modes~\cite{yokoyama2013ultra} and also offers reliable and efficient detection methods, such as homodyne or heterodyne detection~\cite{Leonhardt-essential}. From a theoretical point of view, quantum information with continuous variables provides different perspectives from quantum information with discrete variables and is described via the formalism of infinite-dimensional Hilbert spaces.

To handily manipulate states in those infinite spaces, mathematical tools initially inspired by physics have been developed such as phase-space formalism~\cite{moyal1949quantum}. In this framework, quantum states are represented by a quasiprobability distribution over phase space, like the Wigner function~\cite{wigner1997quantum}. These representations provide a geometric intuition of quantum states~\cite{lee1991measure}: quantum states are separated into two categories, Gaussian and non-Gaussian, depending on whether their Wigner function is a Gaussian function or not. 

Non-Gaussian quantum states are essential to a variety of quantum information processing tasks such as quantum state distillation~\cite{giedke2002characterization,eisert2002distilling,fiuravsek2002gaussian}, quantum error-correction~\cite{niset2009no}, universal quantum computing~\cite{lloyd1999quantum,ghose2007non} or quantum computational speedup~\cite{Bartlett2002,chabaud2020classical}.
Within those, an important subclass of non-Gaussian states are the states which display negativity in the Wigner function. These two classes of states coincide for pure states---namely, non-Gaussian pure states have a negative Wigner function---as pure states with a positive Wigner function are necessarily Gaussian states by Hudson theorem~\cite{hudson1974wigner,soto1983wigner}. However, this is not the case for mixed states and the (convex) set of states with a positive Wigner function becomes much harder to characterise~\cite{mandilara2009extending,filip2011detecting}.

The negativity of other phase-space quasiprobability distributions has been used to define different notions of quantumness~\cite{tan2020negativity}. In particular, the negativity of the Glauber--Sudarshan $P$ quasiprobability distribution of a quantum state is known as its non-classicality~\cite{titulaer1965correlation}. The Wigner function can be obtained from the more singular $P$ function by a Gaussian convolution, thus positivity of the latter implies positivity of the former. In particular, negativity of the Wigner function implies non-classicality, although there are non-classical states with positive Wigner function, such as squeezed Gaussian states.

In addition to its fundamental relevance as a non-classical property of physical systems~\cite{kenfack2004negativity}, Wigner negativity is also essential for quantum computing, since continuous-variable quantum computations described by positive Wigner functions can be simulated efficiently classically~\cite{Mari2012}. Wigner negativity is thus a necessary resource, though not sufficient~\cite{garcia2020efficient}, for quantum computational speedup with continuous variables.

With the rapid development of quantum technologies~\cite{preskill2018quantum}, finding efficient methods for assessing the correct functioning of quantum devices is of timely importance~\cite{eisert2020quantum}. Detecting key properties of a quantum state, such as entanglement or Wigner negativity, can be done by a full tomographic reconstruction~\cite{d2003quantum}. However, such reconstructions are very costly in terms of the number of measurements needed, and require performing a tomographically complete set of measurements---thus usually involving multiple measurement settings. In the continuous-variable setting the task is even more daunting, since the Hilbert space of quantum states is infinite-dimensional~\cite{lvovsky2009continuous,chabaud2020building}. 

Instead, one may introduce witnesses for specific properties of quantum states~\cite{terhal2001family,lewenstein2000optimization,mari2011directly,kiesel2012universal,chabaud2020certification} that are easier to measure experimentally.
These witnesses should be observables that possess a threshold expectation value indicating whether the measured state exhibits the desired property or not.  
Intuitively, a witness for a given property can be thought of as a separating hyperplane in the set of quantum states, such that any state on one side of this hyperplane has this property (see~\refig{witness}). In particular, some states with the sought property may remain unnoticed by the witness. In that regard, one may use a complete set of witnesses such that for each state exhibiting the desired property, there exists at least one witness in the set that captures it.
\begin{figure}[ht!]
	\begin{center}
		\includegraphics[width=0.9\columnwidth]{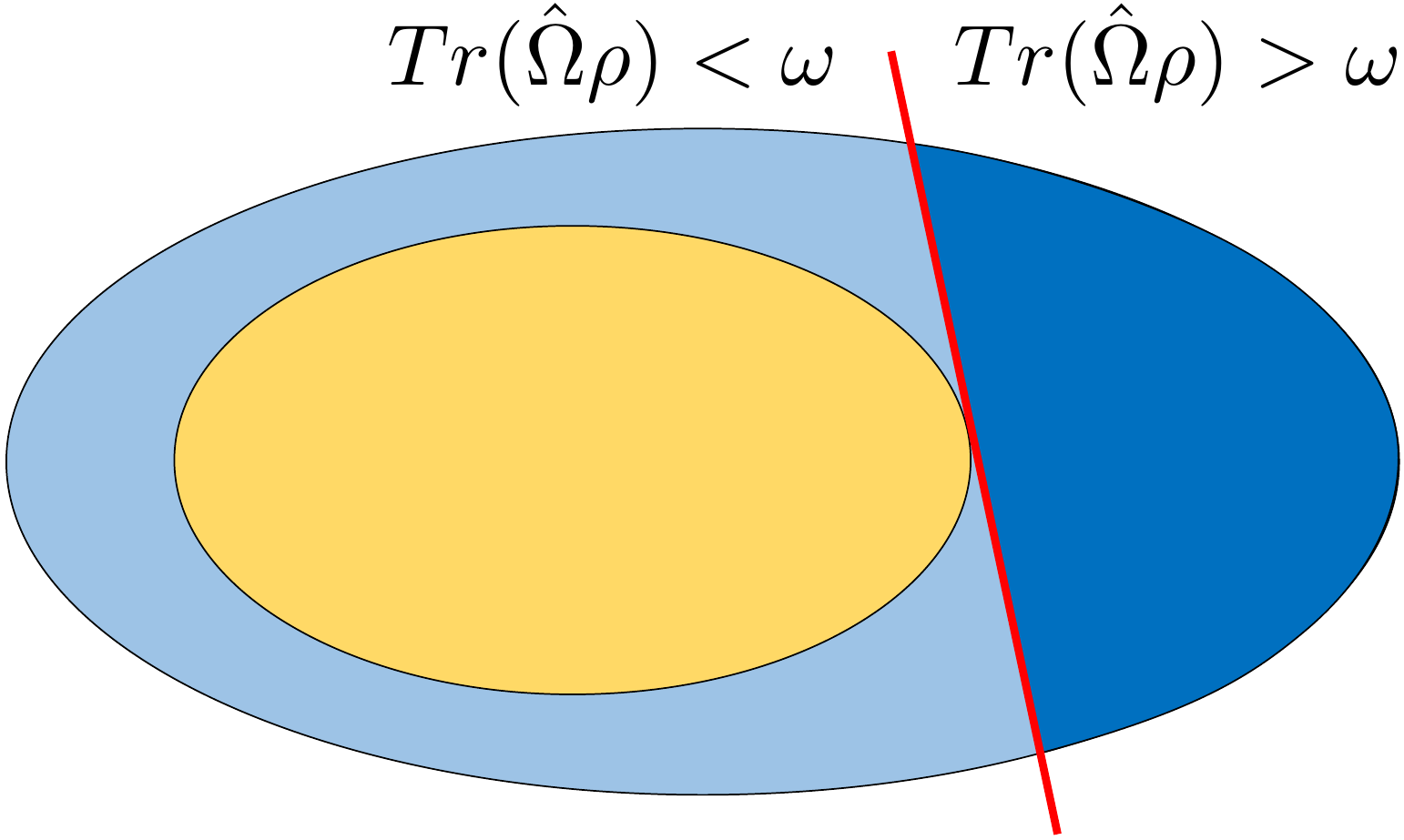}
		\caption{Pictorial representation of a witness $\hat\Omega$ with threshold value $\omega$ for a given property. In yellow: states without the property. In blue: states with the property. In red: witness threshold value. In light blue: states with the property undetected by the witness.}
		\label{fig:witness}
	\end{center}
\end{figure}

Here, we introduce and study a complete family of witnesses for Wigner negativity of single-mode quantum states, expressed using Fock states projectors. The expectation values of our witnesses are linear functions of the state which may be efficiently estimated experimentally using standard homodyne or heterodyne detection, thus providing a reliable method for detecting Wigner negativity with certifiable bounds. 
Additionally, we show that the amount by which the measured expectation value exceeds the threshold value of the witness provides an operational measure of Wigner negativity: it directly lower bounds the distance between the measured state and the set of states with positive Wigner function.

We cast the computation of the threshold values of the witnesses as infinite-dimensional linear programs, which can be either relaxed or restricted. Upper and lower bounds for the threshold values of our witnesses are then given by two converging hierarchies of finite-dimensional semidefinite programs, similar in spirit to the Lasserre--Parrilo hierarchy~\cite{lasserre2001global,parrilo2000structured} and the subsequent Lasserre hierarchy~\cite{lasserre2011new}. 

Finally, we discuss the generalisation to multimode quantum states and show that most of our results are also applicable in this case.

Our work is thought to be relevant for physicists interested in characterising Wigner negativity of quantum states---either theoretically or experimentally---and mathematicians interested in infinite-dimensional convex optimisation theory.
To that end, the rest of the paper is structured as follows: we give some notations and background in the next section~\ref{sec:background} before a detailed exposition of our witnesses in section~\ref{sec:results}. Section~\ref{sec:protocol} describes the experimental procedure for witnessing Wigner negativity of a quantum state using these witnesses, together with use-case examples. The following section~\ref{sec:opti}---which deals with infinite-dimensional optimisation techniques of independent interest---is devoted to estimating the threshold values of our witnesses: after some technical background in section~\ref{sec:functionspaces}, section~\ref{sec:LP} reformulates the problem of finding the threshold value of a witness as an infinite-dimensional linear optimisation, while section~\ref{sec:SDP} derives two hierarchies of semidefinite relaxations and restrictions for this linear program, yielding numerical upper and lower bounds for the threshold value. Section~\ref{sec:CVproof} establishes the proof of convergence of these hierarchies of upper and lower bounds to the threshold values in their respective optimisation spaces. 
We introduce the generalisation to the multimode case in section~\ref{sec:multi} and conclude with a few open questions in section~\ref{sec:conclusion}.

\section{Notations and background}
\label{sec:background}

\subsection{Preliminary notations}

For all $m\in\mathbb N^*$, Sym$_m$ denotes the space of $m\times m$ real symmetric matrices. An exponent $T$ denotes the transpose while an exponent $\dag$ denotes the conjugate transpose.

$\mathcal H$ denotes a separable infinite-dimensional Hilbert space equipped with a countable orthonormal single-mode Fock basis $\{\ket n\}_{n\in\mathbb N}$. We write $\mathcal D(\mathcal H)$ the set of quantum states (positive semidefinite operators with unit trace) over $\mathcal H$. A single-mode quantum state $\rho$ can be expanded in Fock basis as $\rho=\sum_{k,l=0}^{+\infty}\rho_{kl}\ket k\!\bra l$.

The fidelity between two quantum states $\rho$ and $\sigma$ is denoted $F(\rho,\sigma)=\Tr(\sqrt{\sqrt\rho\sigma\sqrt\rho})^2$. When one of the states is pure, it reduces to $F(\rho,\sigma)=\Tr(\rho\sigma)$. The trace distance between two quantum states $\rho$ and $\sigma$ is denoted $D(\rho,\sigma)=\frac12\Tr(\sqrt{(\rho-\sigma)^2})$. The trace distance can be related to the maximum probability of distinguishing between two quantum states. The fidelity and trace distance are related by $1-F\le D\le\sqrt{1-F}$~\cite{nielsen2002quantum}.

We denote by $\hat a$ and $\hat a^\dag$ the single-mode annihilation and creation operators, respectively, defined by their action on the Fock basis:
\begin{equation}
    \begin{aligned}
        \,&\hat a\ket n=\sqrt n\ket{n-1},\quad\quad\quad\text{for }n\in\mathbb N^*,\\
        \,&\hat a\ket 0=0,\\
        \,&\hat a^\dag\ket n=\sqrt{n+1}\ket{n+1},\;\quad\text{for }n\in\mathbb N.
    \end{aligned}
\end{equation}
These operators satisfy the canonical commutation relation $[\hat a,\hat a^\dag]=\mymathbb 1$. For all $\alpha\in\mathbb C$, we write
\begin{equation}
    \hat D(\alpha)=e^{\alpha\hat a^\dag-\alpha^*\hat a}
\end{equation}
the displacement operator of amplitude $\alpha$~\cite{weedbrook2012gaussian}. The coefficients of the displacement operator in Fock basis are given by~\cite{wunsche1998laguerre}:
\begin{equation}\label{eq:coefD}
    \begin{aligned}
      &\braket{k|\hat D(\alpha)|l}=e^{-\frac12|\alpha|^2}\\
      &\quad\times\sum_{p=0}^{\min k,l}\frac{\sqrt{k!l!}(-1)^{l-p}}{p!(k-p)!(l-p)!}\alpha^{k-p}\alpha^{*l-p},
    \end{aligned}
\end{equation}
for all $k,l\in\N$ and all $\alpha\in\C$.

\subsection{Wigner function}
\label{sec:Wigner}

The Wigner function of a single-mode quantum state $\rho$ is an equivalent representation of the state in phase space which can be expressed as~\cite{royer1977wigner,banaszek1999direct}:
\begin{equation}
    \label{eq:defWigner}
    W_\rho(\alpha)=\frac2\pi\Tr\left[\hat D(\alpha)\hat\Pi\hat D^\dag(\alpha)\rho\right],
\end{equation}
for all $\alpha\in\mathbb C$, where $\hat D$ is defined above and $\hat\Pi=(-1)^{\hat a^\dag\hat a}=\sum_{n\ge0}{(-1)^n\ket n\!\bra n}$ is the parity operator. In particular, the Wigner function of a quantum state is related to the expectation value of displaced parity operators. 

The Wigner function is a real-valued quasi-probability distribution~\cite{cahill1969density}, i.e.\@, a normalised distribution which can take negative values. Hence, it cannot be sampled directly experimentally. However, its marginals are probability distributions which can be sampled using homodyne detection~\cite{lvovsky2009continuous}. Alternatively, heterodyne detection (also called double homodyne detection) allows for sampling from a smoothed version of the Wigner function~\cite{husimi1940some,richter1998determination}. In both cases, applying a displacement before the detection is equivalent to measuring directly with homodyne or heterodyne detection and applying a classical post-processing procedure---namely, a translation of the classical outcome according to the displacement amplitude~\cite{chabaud2020efficient,chabaud2020certification}.

As mentioned in the introduction, continuous-variable quantum states are classified in two categories, Gaussian and non-Gaussian, depending on the shape of their Wigner function. The set of Gaussian states is well-understood~\cite{ferraro2005gaussian} but has a limited power, while characterising the set of non-Gaussian states is an active research topic~\cite{albarelli2018resource,takagi2018convex,zhuang2018resource,chabaud2020stellar}.

The negativity of the Wigner function can only decrease under Gaussian operations~\cite{albarelli2018resource}, i.e., operations that map Gaussian states to Gaussian states. In particular, it is invariant under displacements. It is also a robust property, since two almost indistinguishable quantum states have similar Wigner functions. 
An operational measure of Wigner negativity for a quantum state $\rho\in\mathcal D(\mathcal H)$ is given by its distance with the set of states having a positive Wigner function~\cite{mari2011directly}:
\begin{equation}\label{eq:eta}
    \eta_\rho=\inf_{\substack{\sigma\in\mathcal D(\mathcal H)\\W_\sigma \geq 0}}D(\rho,\sigma),
\end{equation}
where $D$ denotes the trace distance, thus quantifying the operational distinguishability between the state $\rho$ and any state having a positive Wigner function~\cite{nielsen2002quantum}.

A natural choice for a witness of Wigner negativity is the fidelity with a pure state having a Wigner function with negative values, since it is a quantity that can be accessed experimentally by direct fidelity estimation~\cite{d2003quantum,chabaud2020building}. Building on this intuition, and given that all Fock states---with the exception of the (Gaussian) vacuum state $\ket0$---have a negative Wigner function, we introduce in the following section a broad family of  Wigner negativity witnesses for single-mode continuous-variable quantum states based on fidelities with Fock states.

\section{Wigner negativity witnesses}
\label{sec:results}

We introduce the following Wigner negativity witnesses:
\begin{equation}\label{eq:witnessOmega}
    \hat\Omega_{\bm a,\alpha}:=\sum_{k=1}^na_k\hat D(\alpha)\ket k\!\bra k\hat D^\dag(\alpha),
\end{equation}
for $n\in\mathbb N^*$, $\bm a=(a_1,\dots,a_n)\in[0,1]^n$, 
with $\max_ka_k=1$, and $\alpha\in\mathbb C$. These operators are weighted sums of displaced Fock states projectors. 
They can be thought of as Positive Operator-Valued Measure (POVM) elements, and their expectation value for a quantum state $\rho$ is given by
\begin{equation}
    \Tr\!\left(\hat\Omega_{\bm a,\alpha}\,\rho\right)=\sum_{k=1}^na_kF\!\left(\hat D^\dag(\alpha)\rho\hat D(\alpha),\ket k\right),
\end{equation}
where $F$ is the fidelity. This quantity can be directly estimated from homodyne or heterodyne detection of multiple copies of the state $\rho$ by 
translating the samples obtained by the amplitude $\alpha$ in the classical postprocessing and performing fidelity estimation with the Fock states $\ket1,\dots,\ket n$~\cite{d2003quantum,chabaud2020building}.

For $n\in\mathbb N^*$, each choice of $(\bm a,\alpha)\in[0,1]^n\times\mathbb C$ yields a different Wigner negativity witness. In particular, when $\alpha=0$ and one entry of the vector $\bm a$ is equal to $1$ and all the other entries are $0$, the expectation value of the witness is given by the fidelity with a single Fock state.

To each witness $\hat\Omega_{\bm a,\alpha}$ is associated a threshold value defined as:
\begin{equation}\label{eq:threshold}
    \omega_{\bm a}:=\sup_{\substack{\rho\in\mathcal D(\mathcal H)\\W_\rho \geq 0}}\Tr\!\left(\hat\Omega_{\bm a,\alpha}\,\rho\right).
\end{equation}
Since negativity of the Wigner function is invariant under displacements, the threshold values do not depend on the value of the displacement amplitude $\alpha$ and we thus write $\omega_{\bm a}$ (rather than $\omega_{\bm a,\alpha}$) for the threshold value associated to the witness $\hat\Omega_{\bm a,\alpha}$. This is sensible, given that the threshold value asks for non-negativity anywhere in phase space, so a displacement in phase space should not change its value. 
Combining \begin{enumerate*}[label=(\roman*)]\item that the threshold value associated to a witness does not depend on the displacement parameter $\alpha$ and \item that we can always take into account displacement via classical post-processing if one uses homodyne or heterodyne detection associated to $\hat \Omega_{\bm a,0}$~\cite{chabaud2020efficient,chabaud2020certification}, we can restrict the analysis to witnesses of the form $\hat \Omega_{\bm a,0}$ that will generate the family $\{ \hat \Omega_{\bm a,\alpha} \}_{\alpha \in \C}$.\end{enumerate*}
Note however that the choice of displacement amplitude can play an important role for certifying negativity of certain quantum states. 

If the measured expectation value for an experimental state is higher than the threshold value given by \refeq{threshold}, this implies by definition that its Wigner function takes negative values.
Moreover, the following result shows that the amount by which the expectation value exceeds the threshold value directly provides an operational quantification of Wigner negativity for that state:

\begin{lemma}\label{lem:operational}
Let $\rho\in\mathcal D(\mathcal H)$ Wigner negative, and fix a witness $\hat\Omega_{\bm a,\alpha}$ defined in~\refeq{witnessOmega}, for $n\in\N^*$, $\bm a=(a_1,\dots,a_n)\in[0,1]^n$, and $\alpha\in\C$, with threshold value $\omega_{\bm a}$. Let us further assume that that it violates the threshold value of the witness \ie $\Tr(\hat\Omega_{\bm a,\alpha}\,\rho) > \omega_{\bm a}$ and denote the amount of violation as
\begin{equation}
    \delta_{\bm a,\alpha}(\rho):=\Tr\!\left(\hat\Omega_{\bm a,\alpha}\,\rho\right)-\omega_{\bm a}.
\end{equation}
Then,
\begin{equation}
    \eta_\rho\ge\delta_{\bm a,\alpha}(\rho),
\end{equation}
where $\eta_\rho$ is the distance between $\rho$ and the set of states having a positive Wigner function, defined in~\refeq{eta}.
\end{lemma}

\begin{proof}
We use the notations of the lemma. Let us consider the binary POVM $\{\hat\Omega_{\bm a,\alpha},\mymathbb 1-\hat\Omega_{\bm a,\alpha}\}$. For all $\sigma\in\mathcal D(\mathcal H)$, we write $P^\sigma_{\bm a,\alpha}$ the associated probability distribution: $P^\sigma_{\bm a,\alpha}(0)=1-P^\sigma_{\bm a,\alpha}(1)=\Tr(\hat\Omega_{\bm a,\alpha}\,\sigma)$. 

Let $\sigma$ be a state with a positive Wigner function, so that $\Tr(\hat\Omega_{\bm a,\alpha}\,\sigma)\le\omega_{\bm a}$, by definition of the threshold value. We have:
\begin{equation}
    \begin{aligned}
        \delta_{\bm a,\alpha}(\rho)&=\Tr(\hat\Omega_{\bm a,\alpha}\rho)-\omega_{\bm a}\\
        &\le|\Tr(\hat\Omega_{\bm a,\alpha}\rho)-\Tr(\hat\Omega_{\bm a,\alpha}\sigma)|\\
        &=|P^\rho_{\bm a,\alpha}(0)-P^\sigma_{\bm a,\alpha}(0)|\\
        &=\|P^\rho_{\bm a,\alpha}-P^\sigma_{\bm a,\alpha}\|_{tvd}\\
        &\le D(\rho,\sigma),
    \end{aligned}
\end{equation}
where we used $\delta_{\bm a,\alpha}(\rho)\ge0$ in the second line, $\|P-Q\|_{tvd}=\frac12\sum_x|P(x)-Q(x)|$ denotes the total variation distance, and we used the operational property of the trace distance in the last line~\cite{nielsen2002quantum}. With~\refeq{eta}, taking the infimum over $\sigma$ concludes the proof.
\end{proof}

\noindent This result directly extends to the case where only an upper bound of the threshold value is known: the amount by which the expectation value exceeds this upper bound is also a lower bound of the distance to the set of states having a positive Wigner function.

Importantly, the family of Wigner negativity\index{Wigner negativity} witnesses $\{\hat\Omega_{\bm a,\alpha}\}$ is \textit{complete}, i.e., for any quantum state with negative Wigner function there exists a choice of witness $(\bm a,\alpha)$ such that the expectation value of $\hat\Omega_{\bm a,\alpha}$ for this state is higher than the threshold value. Indeed, by taking $\bm a=(1,0,1,0,1,\dots)$, this family includes as a subclass the complete family of witnesses from~\cite{chabaud2020certification}. 
Indeed, by expanding the definition of the Wigner function in Eq.~\eqref{eq:defWigner} with the expression of the parity operator and using the completeness relation
$\sum_{n \in \N} \ket n \bra n = \Id$, the Wigner function of any density operator $\rho \in \mathcal D(\mathcal H)$ reads:
\begin{equation*}
    W_{\rho}(\alpha) = \frac 2 \pi \left(1 - 2 \Tr\left( \hat \Omega_{(1,0,1,0,\dots),\alpha} \rho \right) \right) \Mdot
\end{equation*}
Hence, for any state with a negative Wigner function, there exists a choice of $\alpha\in\C$ such that the witness $\hat \Omega_{(1,0,1,0,\dots),\alpha}$ with threshold value $\frac 12$ can detect its negativity.

The threshold value in \refeq{threshold} is given by an optimisation problem over quantum states having a positive Wigner function, which is a convex subset of an infinite-dimensional space that does not possess a well-characterised structure. While solving this optimisation problem thus seems unfeasible in general, it turns out that we can obtain increasingly good numerical upper and lower bounds for the threshold value using semidefinite programming. 

Semidefinite programming is a particular case of conic programming---a subfield of convex optimisation---where one optimises linear functions within the convex cone of positive semidefinite matrices~\cite{vandenberghe1996semidefinite}. 
This is a powerful optimisation technique as semidefinite programs (SDP) can be solved efficiently using interior point methods.

The relevant programs are derived in section \ref{sec:SDPupper} and ~\ref{sec:SDPlower} where their convergence is proven. Since these proofs introduce several intermediate forms of the programs,  we explicitly give them below to avoid confusion on which programs to implement numerically.
For  $n\in\mathbb N^*$, $\bm a=(a_1,\dots,a_n)\in[0,1]^n$, and $m\ge n$,
the hierarchies of semidefinite programs that respectively provide lower bounds and upper bounds for the threshold value $\omega_{\bm a}$ associated to the witnesses $\{ \hat \Omega_{\bm a,\alpha} \}_{\alpha \in  \C}$ are:
\begin{fleqn}
\begin{equation*}
    \label{prog:lowerSDP}
    \tag*{$(\text{SDP}^{m,\leq}_{\bm a})$}
        \begin{aligned}
            & \quad \text{Find } Q \in \text{Sym}_{m+1} \text{ and } \bm{F} \in \R^{m+1} \\
            & \quad \text{maximising } \textstyle \sum_{k=1}^na_kF_k \\
            & \quad \text{subject to}  \\
            & \begin{dcases}
            \sum_{k=0}^{m} F_k = 1  \\
            \forall k \in \llbracket 0,m \rrbracket, \; F_k \geq 0\\
            \forall l \in \llbracket 1,m \rrbracket, \;
            \sum_{i+j=2l-1} Q_{ij} = 0 \\
            \forall l \in \llbracket 0,m \rrbracket, \; \sum_{i+j=2l} Q_{ij} = \sum_{k=l}^{m} \frac{(-1)^{k+l}}{l!}  \binom kl F_k \hspace{-3cm} \\
            Q \succeq 0,
            \end{dcases}
        \end{aligned}
\end{equation*}
\end{fleqn}
and
\begin{fleqn}
\begin{equation*}
    \label{prog:upperSDP}
    \tag*{$(\text{SDP}^{m,\geq}_{\bm a})$}
        \begin{aligned}
            & \quad \text{Find } A \in \text{Sym}_{m+1} \text{ and } \bm{F} \in \R^{m+1} \\
            & \quad \text{maximising } \textstyle \sum_{k=1}^na_kF_k \\
            & \quad \text{subject to}  \\
            & \begin{dcases}
            \sum_{k=0}^{m} F_k = 1  \\
            \forall k \in \llbracket 0,m \rrbracket, \; F_k \geq 0\\
            \forall l \in \llbracket 1,m \rrbracket, \forall i+j=2l-1, \;
            A_{ij} = 0 \hspace{-3cm}  \\
            \forall l \leq m, \forall i+j = 2l, \;  A_{ij} = \sum_{k=0}^{l} F_k  \binom lk l! \hspace{-3cm} \\
            A \succeq 0,
            \end{dcases}
        \end{aligned}
\end{equation*}
\end{fleqn}
Let $\omega_{\bm a}^{m,\geq}$ be the optimal value of \refprog{upperSDP}. We show in section~\ref{sec:opti} that the sequence $\{\omega_{\bm a}^{m,\geq}\}_{m\ge n}$ is a decreasing sequence of upper bounds of $\omega_{\bm a}$, which converges to $\omega_{\bm a}$. Similarly, let $\omega_{\bm a}^{m,\leq}$ be the optimal value of \refprog{lowerSDP}. We show that the sequence $\{\omega_{\bm a}^{m,\leq}\}_{m\ge n}$ is an increasing sequence of lower bounds of $\omega_{\bm a}$, which converges to $\omega_{\bm a}^\mathcal S$---a modified threshold value computed with Schwartz functions rather than square-integrable functions; we have $\omega_{\bm a}^\mathcal S\le\omega_{\bm a}$, and the equality between the two values is open.

In particular, the numerical upper bounds $\omega_{\bm a}^{m,\geq}$ can be used instead of the threshold value $\omega_{\bm a}$ to witness Wigner negativity, while the numerical lower bounds $\omega_{\bm a}^{m,\leq}$ may be used to control how much the upper bounds differ from the threshold value. We give a detailed procedure in the following section, together with use-case examples and details on the numerical implementation.

Note that although both problems are maximisation problems, the constraints of \refprog{upperSDP} get tighter as $m$ increases, and thus the corresponding sequence of optimal values is decreasing, while the constraints of \refprog{lowerSDP} get looser as m increases, and thus the corresponding sequence of optimal values is increasing.

\paragraph{State-of-the-art}

Our Wigner negativity witnesses outperform existing ones~\cite{mari2011directly,fiuravsek2013witnessing} in terms of generality and practicality, since they form a complete family and provide much more flexibility with the choice of $n\in\N^*$, $\bm a\in[0,1]^n$ and $\alpha\in\mathbb C$. They are accessible with optical homodyne or heterodyne measurement, and do not require making any assumption on the measured state, unlike other existing methods~\cite{PhysRevLett.119.183601}. Moreover, our witnesses generalise those of~\cite{chabaud2020certification}, and may provide simpler alternatives to detect Wigner negativity. We also provide two converging hierarchies to approximate the threshold values associated to these witnesses. 
Proving convergence is 
important 
and was not considered in the other approaches mentioned above. Finally, our approach also generalises to the multimode setting, as detailed in section~\ref{sec:multi}.

\section{Witnessing Wigner negativity}
\label{sec:protocol}

\subsection{Procedure}
\label{sec:procedure}

In this section, mainly devoted to experimentalists, we describe a procedure to check whether a continuous-variable quantum state exhibits Wigner negativity using our witnesses.
\begin{figure}[ht!]
	\begin{mybox}{Witnessing Wigner negativity:}\label{box:protocol}
    \justifying{
        \begin{enumerate}
            \item Choose a fidelity-based witness $\hat\Omega_{\bm a,\alpha}$ defined in \refeq{witnessOmega} by picking $n\in\mathbb N^*$, $\bm a\in[0,1]^n$ and $\alpha\in\mathbb C$.
            \item Run the upper bound semidefinite program \refprog{upperSDP} for $m\ge n$, get a numerical estimate $\omega^{m,\geq}_{\bm a}$. These values are already computed for $\bm a=(0,\dots,0,1)$ and $n\le10$ in Table~\ref{tab:Fockbounds}.
             \item Run the lower bound semidefinite program \refprog{lowerSDP} for $m \ge n$, get a numerical estimate $\omega^{m,\leq}_{\bm a}$. These values are already computed for $\bm a=(0,\dots,0,1)$ and $n\le10$ in Table~\ref{tab:Fockbounds}.
             \item Estimate the expectation value for that witness of the experimental state from samples of homodyne or heterodyne detection by translating the samples by $\alpha$ and performing fidelity estimation with the corresponding Fock states. This yields an experimental witness value denoted $\omega_{\text{exp}}$.
            \item Compare the value obtained experimentally with the numerical bounds:
            if it is greater than the numerical upper bound, i.e.,
            $\omega_{\text{exp}}\ge\omega^{m,\geq}_{\bm a}$ then the state displays Wigner negativity, and its distance to the set of Wigner positive states is lower bounded by $\omega_{\text{exp}}-\omega^{m,\geq}_{\bm a}$. 
            Otherwise, if it is lower than the numerical lower bound, i.e., $\omega_{\text{exp}}\leq \omega^{m,\leq}_{\bm a}$ then the witness cannot detect Wigner negativity for this state. 
        \end{enumerate}
    }
	\end{mybox}
\end{figure}

The main subroutine of this procedure is to estimate fidelities with displaced Fock states using classical samples from homodyne or heterodyne detection\footnote{Actually, using a fidelity witness rather than a fidelity estimate is sufficient for our purpose.}, in order to compute the experimental value for a Wigner negativity witness. Displacement can be achieved with classical post-processing by translating the classical samples according to the displacement amplitude, and performing direct fidelity estimation with Fock states (see section~\ref{sec:Wigner} and~\cite{lvovsky2009continuous,d2003quantum,chabaud2020building}).
Upper and lower bounds on the threshold value of the witness are then obtained using semidefinite programming, and comparing the experimental witness value to these bounds gives insight about the Wigner negativity of the measured quantum state.

We give a detailed procedure for using our witnesses for detecting Wigner negativity in the framed box. This procedure starts by the choice of a specific witness, and we explain hereafter a heuristic method for picking a good witness.

If the experimental state is anticipated to have negativity at $\alpha$, then one may use the witness with parameters $(n,\bm a,\alpha)$ with $\bm a=(1,0,1,0,\dots)$, which will detect negativity for $n$ large enough~\cite{chabaud2020certification}. However, this may imply having to estimate fidelities with Fock states having a high photon number with homodyne or heterodyne detection, which requires a lot of samples, while simpler witnesses can suffice for the task and be more efficient, as we show in the next section. Moreover, there are cases where the state to be characterised is fully unknown. 

Instead, a simple heuristic for picking a good witness for Wigner negativity is the following: 
\begin{itemize}
\item From samples of homodyne or heterodyne detection of multiple copies of an experimental state, estimate the expected values of witnesses in \refeq{witnessOmega} for a small value of $n$ and a large set of values $\bm a$ and $\alpha$, using the same samples for all witnesses. 
\item Based on these values, pick the simplest witness possible---with the smallest value of $n$---that is able to witness Wigner negativity with a reasonable violation. This is done by comparing the estimated expected values with the upper bounds on the corresponding threshold values. These bounds depend only the choice of the witness parameters $n,\bm a$ and can be precomputed using~\refprog{upperSDP}. To facilitate the use of our methods, we have collected such bounds for $\bm a=(0,\dots,0,1)$ and $n\le10$ in Table~\ref{tab:Fockbounds}. We also precomputed these bounds for $n=3$ and a large number of values of $\bm a$ in Appendix~\ref{sec:weightedwitness_n3} and~\cite{codes}. 
\item Then, estimate the expected value for that witness using a new collection of samples---thus obtaining proper error bars and avoiding the accumulation of statistical errors.
\end{itemize}

\noindent In what follows, we give a few theoretical examples for using our witnesses to detect negativity of the Wigner function of single-mode quantum states.

\subsection{Examples}
\label{sec:examples}

We identify three levels of generality within our family of witnesses in \refeq{witnessOmega}: \begin{enumerate*}[label=(\roman*)]\item fidelities with single Fock states,\item linear combinations of fidelities with Fock states, and \item displaced linear combinations of fidelities with Fock states.\end{enumerate*} 

Fidelities with Fock states are the most practical of our witnesses, since they require the estimation of only one diagonal element of the density matrix of the measured state.
The corresponding values in Table~\ref{tab:Fockbounds} can be used directly by experimentalists: if an estimate of $\braket{n | \rho | n}$ for appropriate $n$ is above one of these numerical upper bounds then it ensures that $\rho$ has a Wigner function with negative values. Moreover, by Lemma~\ref{lem:operational}, the amount by which the estimate of $\braket{n | \rho | n}$ exceeds the numerical upper bound directly provides a lower bound on the distance between $\rho$ and the set of states having a positive Wigner function.

For instance, if we focus on $n=3$ in Table~\ref{tab:Fockbounds}, the threshold value $\omega_3$ satisfies $0.378 \le \omega_3 \le 0.427$. Having a state $\rho$ such that $\braket{3|\rho|3} > 0.427$ guarantees that $\rho$ has Wigner negativity. If $\braket{3|\rho|3} < 0.378$ then we conclude that the witness cannot detect negativity for this state.
\begin{table}[htp]
    \centering
        \begin{tabular}{@{}l|cc@{}}
            \toprule
            $n$ & Lower bound & Upper bound \\ \midrule
            1 & 0.5 & 0.5   \\
            2 & 0.5 & 0.5   \\
            3 & 0.378 & 0.427 \\
            4 & 0.375 & 0.441 \\
            5 & 0.314 & 0.385 \\
            6 & 0.314 & 0.378 \\
            7 & 0.277 & 0.344 \\
            8 & 0.280 & 0.348 \\
            9 & 0.256 & 0.341 \\
            10 & 0.262 & 0.334 \\
            \bottomrule
        \end{tabular}
    \caption{Table of numerical upper and lower bounds for the threshold value $\omega_n$ of various Wigner negativity witnesses obtained using our hierarchies of semidefinite programs at rank up to around $30$. The witnesses considered here are Fock states projectors (photon-number states $\ket n$) from $1$ to $10$. Note that the gap between the lower and upper bounds never exceeds $0.1$. Additionally, the bounds in the first two lines are analytical (see section~\ref{sec:LP}) and for the upper bounds, the corresponding numerical values are $0.528$ and $0.551$, respectively. See section \ref{sec:numerical} and \cite{codes} for the numerical implementation.}
    \label{tab:Fockbounds}
\end{table}
When the experimental state is close to a Fock state (different from the vacuum), a natural choice for the witness thus is the fidelity with the corresponding Fock state. 
For instance, consider a photon-subtracted squeezed vacuum state~\cite{ourjoumtsev2006generating}
\begin{equation}
    \ket{\text{p-ssvs}(r)}=\frac1{\sinh r}\hat a\hat S(r)\ket 0,
\end{equation}
where $\hat S(r)=e^{\frac r2(\hat a^2-\hat a^{\dag2})}$ is a squeezing operator with parameter $r\in\mathbb R$. Its fidelity with the single-photon Fock state $\ket1$ is given by:
\begin{equation}\label{eq:fidepssvs}
    \frac1{(\sinh r)^2}\left|\braket{1|\hat a\hat S(r)|0}\right|^2=\frac1{(\cosh r)^3}.
\end{equation}
When the squeezing parameter is small, this state is close to a single-photon Fock state. In particular, for $0<r<0.70$, the fidelity $F(\text{p-ssvs}(r),1)$ in~\refeq{fidepssvs} is greater than $\omega_1^\ge=\frac12$ and our witness can be used to detect Wigner negativity of this state (see~\refig{examples} (a)).

Another example is given by superpositions of coherent states states: we consider the cat state~\cite{sanders1992entangled}
\begin{equation}
    \ket{\text{cat}_2(\alpha)}=\frac{\ket\alpha+\ket{-\alpha}}{\sqrt{2(1+e^{-2|\alpha|^2})}},
\end{equation}
and the compass state~\cite{zurek2001sub}
\begin{equation}
 \ket{\text{cat}_4(\alpha)}=\frac{\ket\alpha+\ket{-\alpha}+\ket{i\alpha}+\ket{-i\alpha}}{2\sqrt{1+e^{-|\alpha|^2}(2\cos(|\alpha|^2)+1)}},   
\end{equation}
where $\ket\alpha=e^{-\frac12|\alpha|^2}\sum_{k\ge0}\frac{\alpha^k}{\sqrt{k!}}\ket k$ is the coherent state of amplitude $\alpha\in\mathbb C$. 

\begin{figure}[htp]
	\begin{center}
		\subfloat{
		    \includegraphics[width=.9\columnwidth]{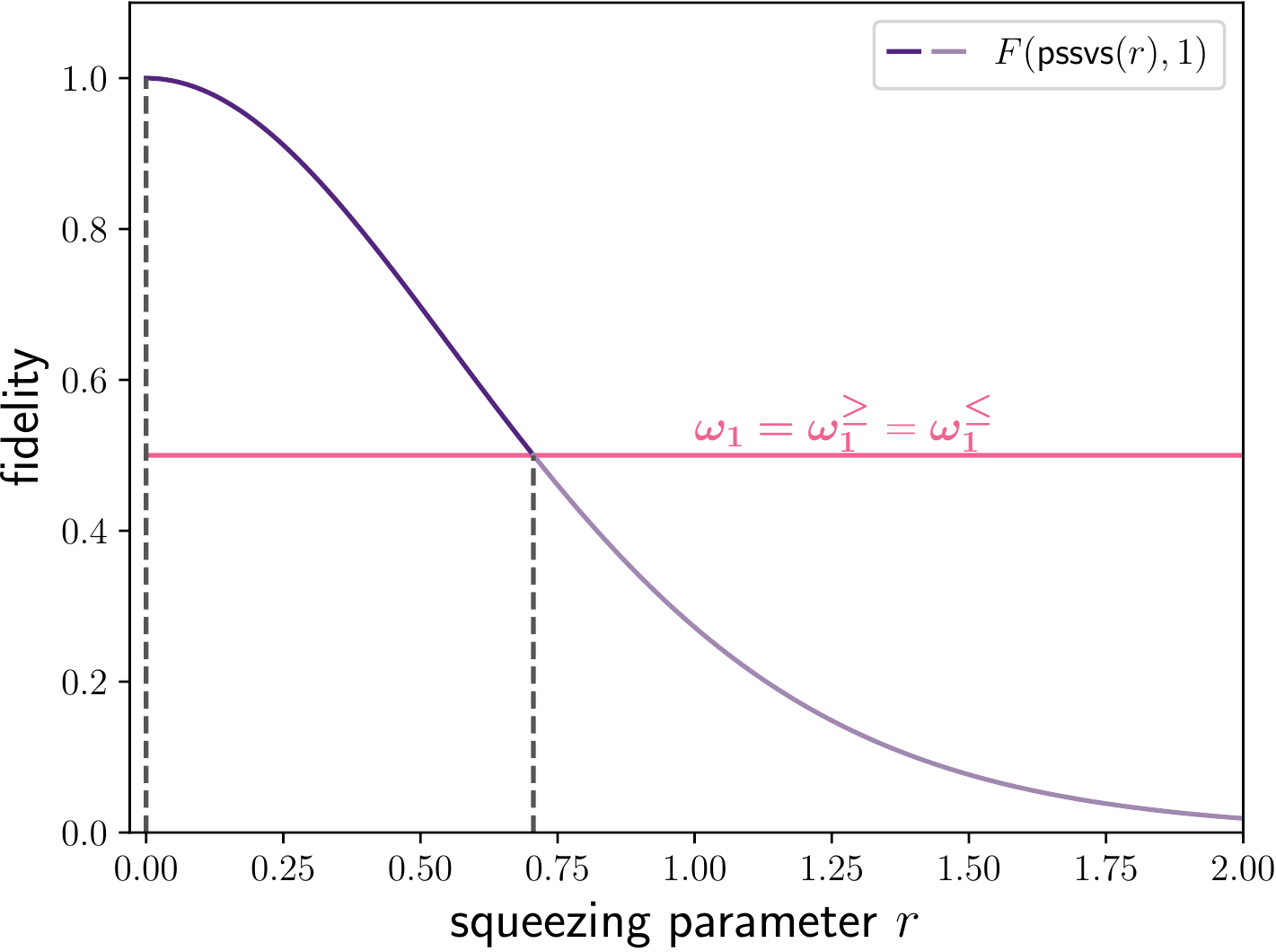}}
		
		\subfloat{
		    \includegraphics[width=.9\columnwidth]{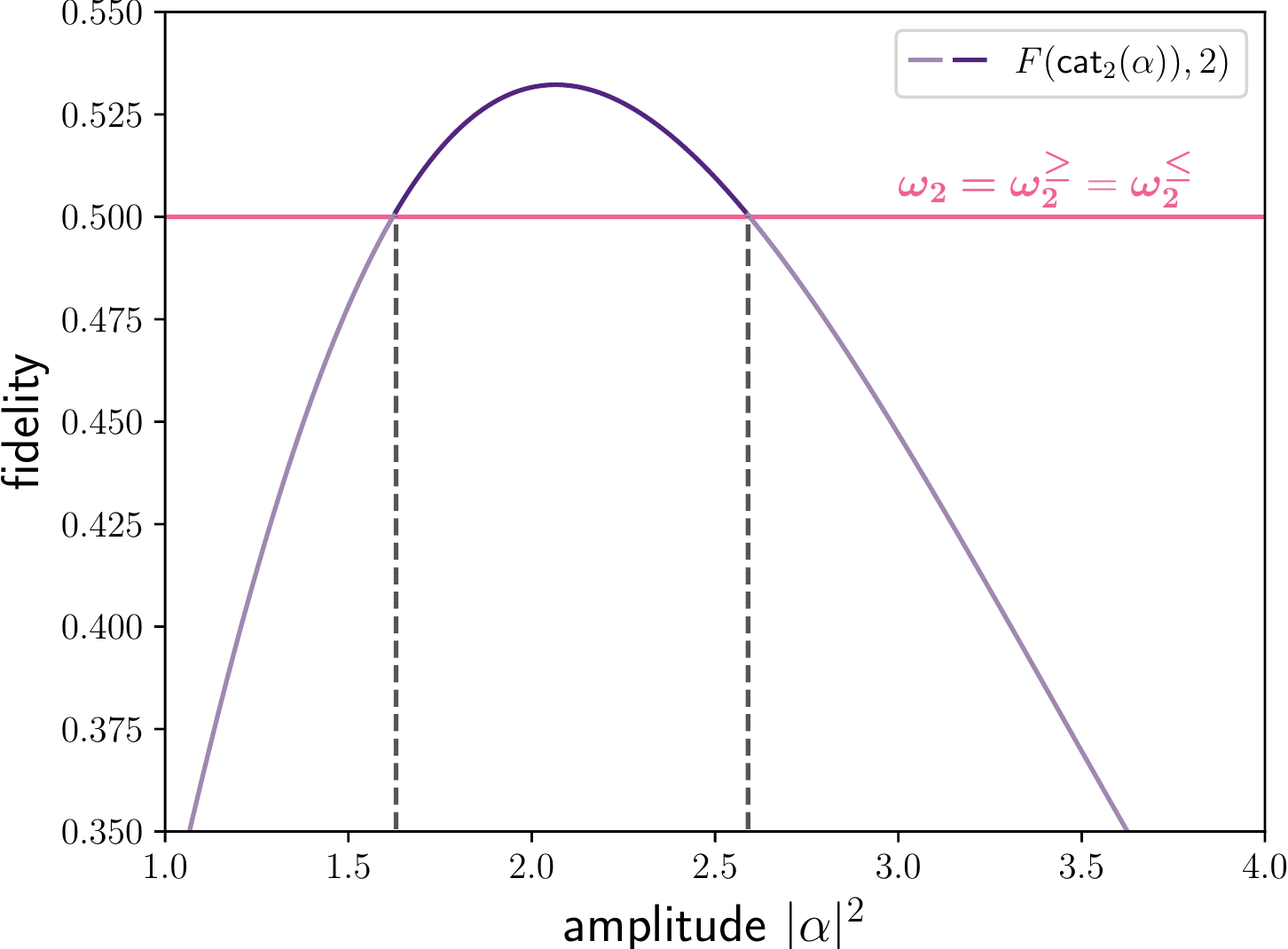}}
		    
		\subfloat{
		    \includegraphics[width=.9\columnwidth]{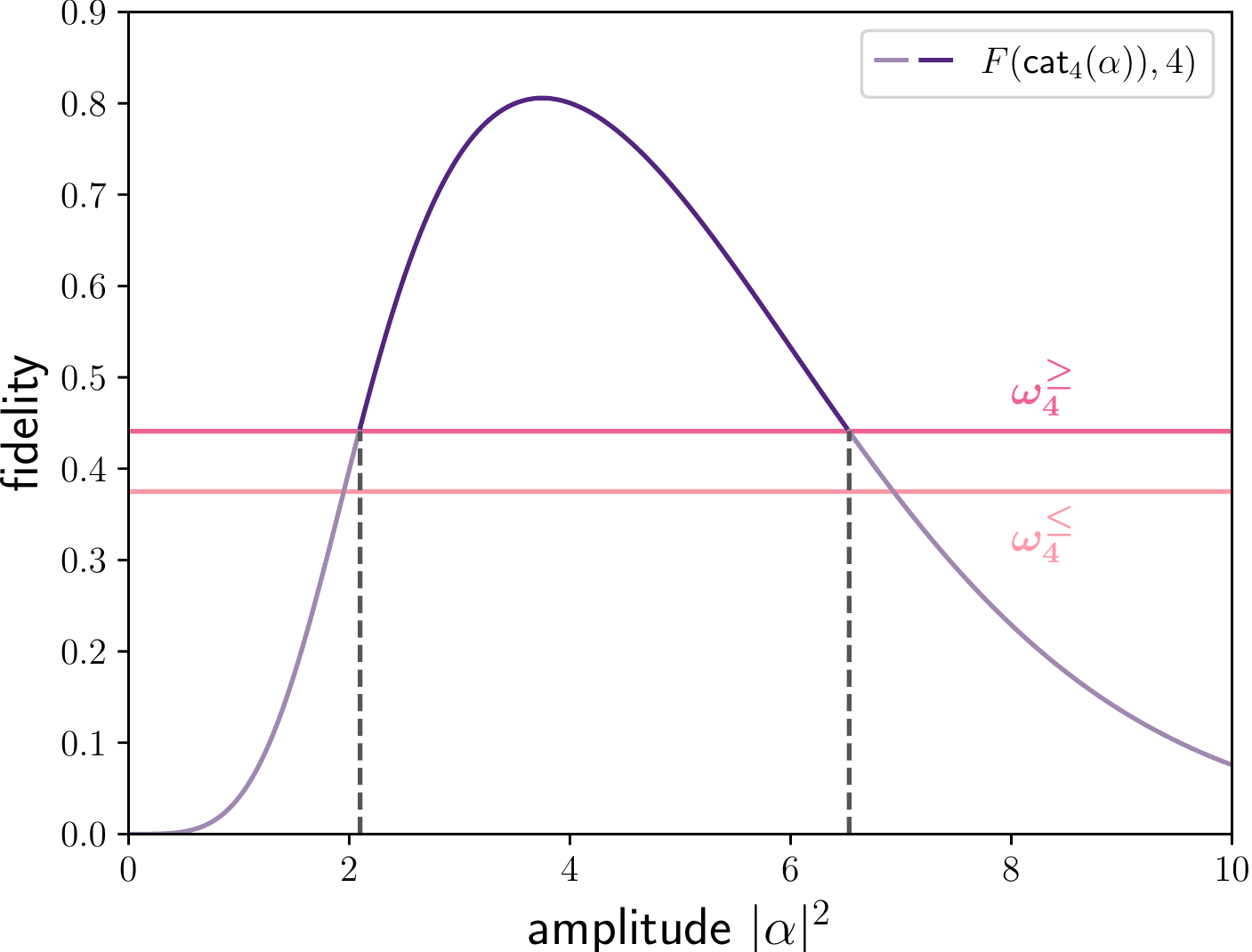}}
		\caption{(a) Fidelities of photon-subtracted squeezed vacuum states $\ket{\text{p-ssvs}(r)}$ with squeezing parameter $r\in\R$ with the Fock state $\ket1$. (b) Fidelities of cat states $\ket{\text{cat}_2(\alpha)}$ with amplitude $\alpha\in\mathbb C$ with the Fock state $\ket2$. (c) Fidelities of compass states $\ket{\text{cat}_4(\alpha)}$ with amplitude $\alpha\in\mathbb C$ with the Fock state $\ket4$. The dashed black lines delimit the intervals of parameter values where our witnesses from Table~\ref{tab:Fockbounds} can be used to detect Wigner negativity of the corresponding state, i.e., when the fidelity (purple curve) is above the witness upper bound (pink horizontal line). When it is below the witness lower bound, we are guaranteed that the witness (here $\ket4\!\bra4$) cannot be used to detect Wigner negativity of the state. 
		}
		\label{fig:examples}
	\end{center}
\end{figure}

We have
\begin{equation}\label{eq:fidecat}
    \left|\braket{2|\text{cat}_2(\alpha)}\right|^2=\frac{|\alpha|^4}{2\cosh(|\alpha|^2)},
\end{equation}
and 
\begin{equation}\label{eq:fidecompass}
    \left|\braket{4|\text{cat}_4(\alpha)}\right|^2=\frac{|\alpha|^8/12}{\cosh(|\alpha|^2)+\cos(|\alpha|^2)}.
\end{equation}

\noindent For $1.63\le|\alpha|^2\le2.59$, the fidelity $F(\text{cat}_2(\alpha),2)$ in~\refeq{fidecat} is greater than $\omega_2^\ge=\frac12$ and our witness corresponding to $n=2$ can be used to detect Wigner negativity of this state. Similarly, for $2.10\le|\alpha|^2\le6.70$, the fidelity $F(\text{cat}_4(\alpha),4)$ in~\refeq{fidecompass} is greater than $\omega_4^\ge=0.441$ and our witness corresponding to $n=4$ can be used to detect Wigner negativity of this state (see~\refig{examples} (b) and (c)).
Note that in all cases, the height difference when the purple curve is above the pink horizontal line directly provides a lower bound on the distance between the corresponding state and the set of states having a positive Wigner function.

Some quantum states will remain unnoticed by all single Fock state negativity witnesses. For example, the state $\rho_{0,1,2} := \frac19 \ket0\!\bra0 + \frac49 \ket1\!\bra1+ \frac49 \ket2\!\bra2$ has a negative Wigner function but is not detected by any of the single Fock state negativity witnesses, since the lower bounds for $n=1,2$ in Table~\ref{tab:Fockbounds} are higher than $\frac49$, and this state has fidelity $0$ with higher Fock states. However, it is detected by simple witnesses based on linear combinations of fidelities. For example, with $n=2$ and $\bm a=(1,1)$ we find numerically that the threshold value of the witness $\ket1\!\bra1+\ket2\!\bra2$ is less than $0.875$ when running the corresponding program \refprog{upperSDP} for $m=7$. 
And since $\Tr(\rho_{0,1,2} (\ket1\!\bra1+\ket2\!\bra2)) = \frac89 > 0.875$, this linear combination of Fock state fidelities can indeed detect Wigner negativity for this state.

However, some quantum states with a negative Wigner function will always go unnoticed by the previous witnesses because these witnesses are invariant under phase-space rotations while the Wigner function of those states becomes positive under random dephasing: consider for instance the superposition $\sqrt{1-\frac1s}\ket0+\frac1{\sqrt s}\ket1$, for $s>2$ (which under random dephasing is mapped to $(1-\frac1s)\ket0\!\bra0+\frac1s\ket1\!\bra1$). In that case, the Wigner negativity of such states can still be witnessed by using the displaced version of our witnesses. In particular, if any single-mode quantum state has a Wigner function negative at $\alpha\in\mathbb C$, then there is a choice of $n\in\mathbb N^*$ such that the witness in \refeq{witnessOmega} defined by $\bm a=(1,0,1,0,1,\dots)$ and the displacement amplitude $\alpha$ detects its negativity~\cite{chabaud2020certification}.
In practice, simpler witnesses may suffice to detect negativity, and the choice of witness will ultimately depend on the experimental state at hand.

\begin{figure}[ht!]
	\begin{center}
		\includegraphics[width=\columnwidth]{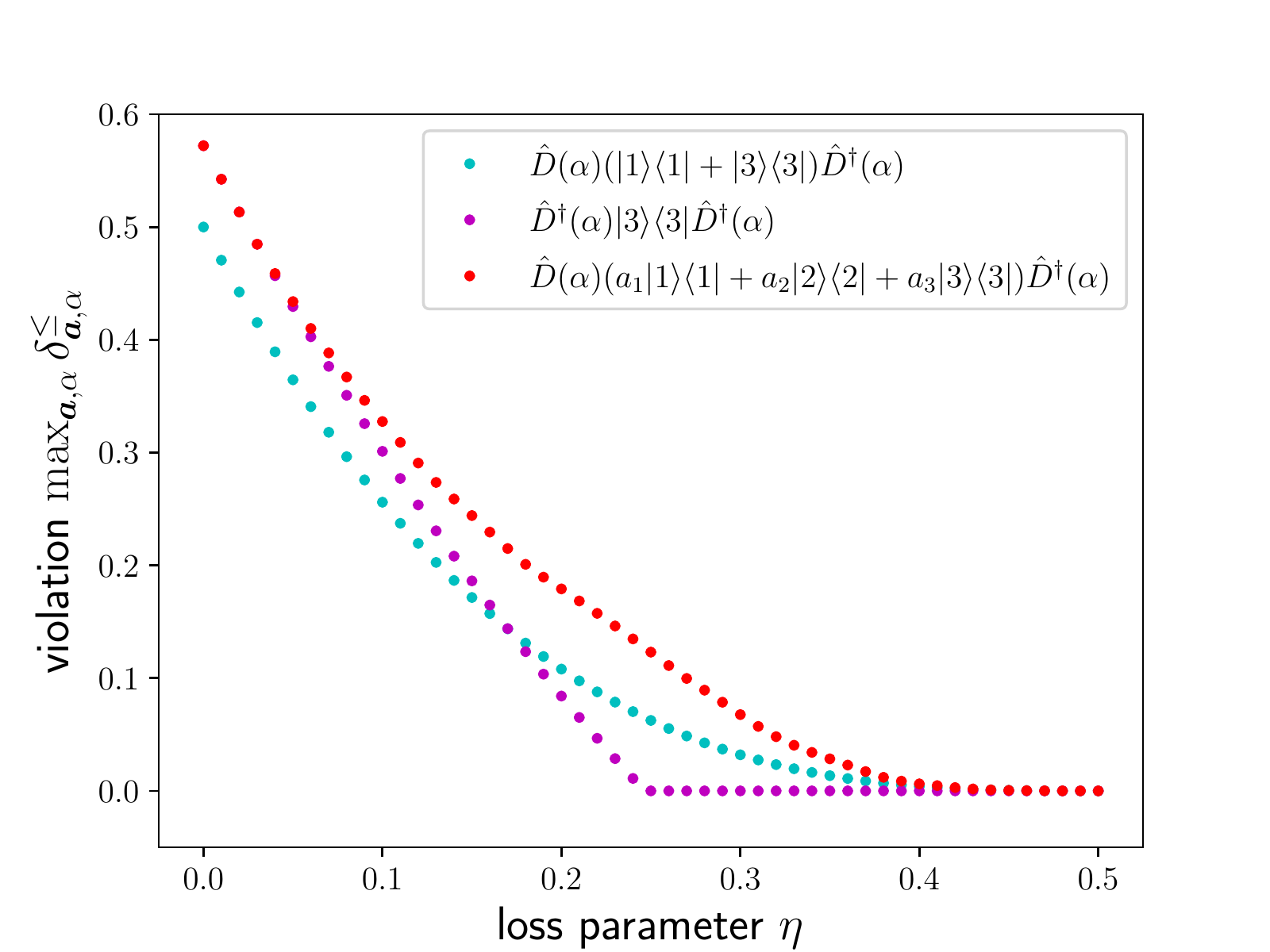}
		\caption{Lower bounds on the violation $\delta_{\bm a,\alpha} = \Tr(\hat\Omega_{\bm a,\alpha}\,\rho_{3,\eta})-\omega_{\bm a}$ for a lossy 3-photon Fock state $\rho_{3,\eta} = \eta^3\ket3\!\bra3+3\eta^2(1-\eta)\ket2\!\bra2
    +3\eta(1-\eta)^2\ket1\!\bra1+(1-\eta)^3\ket0\!\bra0$ for $\bm a = (a_1,a_2,a_3)$ with the loss parameter $\eta$. Precomputed bounds on threshold values for witnesses of the form $\hat \Omega_{\bm a}= a_1 \ket1\!\bra1 + a_2 \ket2\!\bra2 + a_3 \ket3\!\bra3 $ can be found in Appendix \ref{sec:weightedwitness_n3}. We use these values to find the witness $\hat\Omega_{\bm a,\alpha}$ giving the maximum lower bound $\delta_{\bm a,\alpha}^{\le} =\Tr(\hat\Omega_{\bm a,\alpha}\,\rho_{3,\eta})-\omega_{\bm a}^{\ge}$ on the violation $\delta_{\bm a,\alpha}$, for different values of the loss parameter $\eta$. These lower bounds are represented in red. In blue is the maximal violation that can be detected using the witnesses $\hat \Omega_{(1,0,1),\alpha}=\hat D(\alpha)(\ket1\!\bra1+\ket3\!\bra3)\hat D^\dag(\alpha)$~\cite{chabaud2020certification}. In violet is the maximal violation that can be detected using the more naive witness $\hat \Omega_{(0,0,1),\alpha}=D(\alpha)\ket3\!\bra3D^\dag(\alpha)$. Note that $\rho_{3,\eta}$ has a non-negative Wigner function for $\eta \geq 0.5$.
    }
		\label{fig:lossy}
	\end{center}
\end{figure}

Hereafter we discuss the heuristics for picking a good witness, with the theoretical example of the lossy $3$-photon Fock state:
\begin{align}
\begin{split}
\rho_{3,\eta} &:= (1-\eta)^3\ket3\!\bra3+3\eta(1-\eta)^2\ket2\!\bra2\\
&\;\;+3\eta^2(1-\eta)\ket1\!\bra1+\eta^3\ket0\!\bra0,
\end{split}
\end{align}
where $0\le\eta\le1$ is the loss parameter. Setting $\eta=0$ gives $\rho_{3,\eta}=\ket3\!\bra3$ while setting $\eta=1$ gives $\rho_{3,\eta}=\ket0\!\bra0$. This state has a non-negative Wigner function for $\eta\geq\frac12$. 
The fidelities of $\rho_{3,\eta}$ with displaced Fock states $\hat D(\alpha)\ket l$ are given by:
\begin{equation}\label{eq:fiderho3etal}
    \begin{aligned}
        F(\rho_{3,\eta},\hat D(\alpha)&\ket l)\\
        &=(1-\eta)^3|\braket{3|\hat D(\alpha)|l}|^2\\
        &+3\eta(1-\eta)^2|\braket{2|\hat D(\alpha)|l}|^2\\
        &+3\eta^2(1-\eta)|\braket{1|\hat D(\alpha)|l}|^2\\
        &+\eta^3|\braket{0|\hat D(\alpha)|l}|^2.    
    \end{aligned}
\end{equation}
where the coefficients of the displacement operator in Fock basis are given in Eq.~\eqref{eq:coefD}.
In an experimental scenario, the state would be unknown and these fidelities should be estimated using samples from a homodyne or heterodyne detection of the state translated by $\alpha$, and estimating the fidelities with Fock states~\cite{chabaud2020efficient,chabaud2020certification}. 

Following the heuristic detailed in the previous section, we have determined good Wigner negativity witnesses for $50$ values of the loss parameter $\eta$ between $0$ and $0.5$ as follows: for each value $\eta$, we have computed numerically the values of the fidelities in Eq.~\eqref{eq:fiderho3etal} for $l=1,2,3$, and for displacement parameters $\alpha=q/10+ip/10$ for all $q,p\in\llbracket0,10\rrbracket$. Using these values, we have computed the expectation value of the witnesses $\hat\Omega_{\bm a,\alpha}$ for multiple choices of $\bm a=(a_1,a_2,a_3)$ with $\max_ia_i=1$. We have used the corresponding precomputed bounds $\omega_{\bm a}^{\ge}$ on the threshold values of $\hat\Omega_{\bm a,\alpha}$ in Appendix~\ref{sec:weightedwitness_n3} to determine the witness leading to the maximal lower bound $\delta_{\bm a,\alpha}^{\le} :=\Tr(\hat\Omega_{\bm a,\alpha}\,\rho_{3,\eta})-\omega_{\bm a}^{\ge}$ on the violation $\delta_{\bm a,\alpha}=\Tr(\hat\Omega_{\bm a,\alpha}\,\rho_{3,\eta})-\omega_{\bm a}$ over the choice of $(\bm a,\alpha)$.

We have represented these violations for each value of the loss parameter $\eta$ in Fig.~\ref{fig:lossy}. For all values of $\eta$, we find that the optimal displacement parameter is $\alpha=0$. On the other hand, we find different optimal choices of $\bm a$ for different values of $\eta$. To illustrate the usefulness of the optimisation over the choice of witnesses parametrised by $(\bm a,\alpha)$, we have also represented the violations obtained when using the witnesses $\hat\Omega_{(1,0,1),\alpha}=\hat D(\alpha)(\ket1\!\bra1+\ket3\!\bra3)\hat D^\dag(\alpha)$ from~\cite{chabaud2020certification} for all values of $\eta$. In that setting, the violations obtained quantify how hard it is to detect the Wigner negativity of the state: a larger violation implies that a less precise estimate of the witness expectation value is needed to witness Wigner negativity.
In particular, we obtain that our optimised witnesses always provide a greater violation to detect negativity than the previous witnesses which will result in an easier experimental detection. 
We also represented the violation obtained when using the more naive witnesses $\hat\Omega_{(0,0,1),\alpha}=\hat D(\alpha)\ket3\!\bra3\hat D^\dag(\alpha)$ and we see that it is only useful when the loss parameter is smaller than $0.25$, while the optimised witnesses may detect negativity of the state $\rho_{3,\eta}$ up to $\eta=0.5$---when the Wigner functions becomes non-negative---provided the estimates of the fidelities are precise enough.

Overall, this procedure only amounts to a simple classical post-processing of samples from homodyne or heterodyne detection and yields a good witness for detecting Wigner negativity.

\subsection{Numerical implementation}
\label{sec:numerical}

Here we discuss numerical implementations of the semidefinite programs \refprog{lowerSDP} and \refprog{upperSDP}. All codes are available \href{https://archive.softwareheritage.org/swh:1:dir:d98f70e386783ef69bf8c2ecafdb7b328b19b7ec}{here} \cite{codes}.

We implemented the semidefinite programs with Python through the interface provided by PICOS~\cite{sagnol2012picos}. We first used the solver Mosek \cite{mosek} to solve these problems but, while the size of the semidefinite programs remains relatively low for small values of $n$ and $m$, binomial terms grow rapidly and numerical precision issues arise quickly (usually for $m=12$, $n \le m$). 
The linear constraints involving $Q_{ij}$ in the semidefinite programs come from a polynomial equality (see Lemma~\ref{lem:pospolyR}). While polynomial equalities are usually written in the canonical basis, a first trick is to express them in a different basis---for instance the basis $(1,\frac X{1!},\frac {X^2}{2!}, \dots)$---to counterbalance the binomial terms. 

However, this may not be sufficient to probe larger values of $m$. Instead, we used the solver SDPA-GMP~\cite{nakata2010numerical,fujisawa2002sdpa} which allows arbitrary precision arithmetic. While much slower, this solver is dedicated to solve problems requiring a lot of precision. Because our problems remain rather small, time efficiency is not an issue and this solver is particularly well-suited. All problems were initially solved on a regular laptop as warning flags on optimality were raised before the problems were too large. A high-performance computer\footnote{DELL PowerEdge R440, 384 Gb RAM, Intel Xeon Silver 4216 processor, 64 threads from LIP6. }---handling floating point arithmetic more accurately---was later used to compute further ranks in the hierarchy. 

Using the semidefinite programs \refprog{upperSDP} and \refprog{lowerSDP} for values of $m$ up to around $30$ and $\bm a=(0,0,\dots,0,1)$ (where the size $n$ of the vector $\bm a$ is ranging from $1$ to $10$), we have obtained upper and lower bounds for the threshold values of Wigner negativity witnesses corresponding to fidelities with Fock states from $1$ to $10$, reported in Table~\ref{tab:Fockbounds}.

We also computed upper and lower bounds on the threshold values of witnesses of the form:
\begin{equation}
    \hat \Omega_{(a_1,a_2,a_3)} = \sum_{k=1}^3a_k\ket k\!\bra k,
\end{equation}
where $\forall i \in \{1,2,3\},\, 0 \leq a_i \leq 1$ and $\max_i a_i =1$. We focused on these particular witnesses for experimental considerations as it is challenging to obtain fidelities with higher Fock states. We fix one coefficient equal to $1$ and vary each other $a_i$ from $0$ to $1$ with a step of $0.1$. The resulting bounds on the threshold values can be found in appendix \ref{sec:weightedwitness_n3}. 

We now turn to the mathematical proofs of our results, i.e., that the threshold values in \refeq{threshold} can be upper bounded and lower bounded by the optimal values of the converging hierarchies of semidefinite programs \refprog{upperSDP}$_{m\ge n}$ and \refprog{lowerSDP}$_{m\ge n}$, respectively.

The following section is rather technical as we dive into infinite-dimensional optimisation techniques to prove the convergence of the two hierarchies of semidefinite programs. Some readers may want to skip directly to section \ref{sec:multi}.

\section{Infinite-dimensional optimisation}
\label{sec:opti}

In this section we use infinite-dimensional optimisation techniques:
\begin{enumerate*}[label=(\roman*)]
    \item to phrase the computation of the witness threshold value introduced in \refeq{threshold} as an infinite-dimensional linear program in section~\ref{sec:LP},
    \item to derive two hierarchies of finite-dimensional semidefinite programs that upper bound and lower bound the threshold value in section~\ref{sec:SDP}, and
    \item to show, in section~\ref{sec:CVproof}, that the sequence of upper bounds converges to the threshold value computed over square-integrable functions and the sequence of lower bounds converges to the threshold value computed over Schwartz functions (see~\refig{structure}). Given the technicalities of the proofs above, we sketch them in section~\ref{sec:sketch} before detailing them in the following sections.
\end{enumerate*}

As a convention, except if specifically mentioned, we will use the terminology `relaxation' and `restriction' from the point of view of the primal program. We will refer to the hierarchy of semidefinite programs providing the upper bounds as a \textit{hierarchy of relaxations} because the obtained semidefinite programs are indeed relaxations of the primal program (while they are restrictions of the dual program). Likewise, we will refer to the hierarchy providing the lower bounds as a \textit{hierarchy of restrictions}.
\begin{figure}[t]
	\begin{center}
		\centering

\begin{tikzpicture}[scale=0.56]

\node[inner sep=0pt] (A) at (-0.8,-1) {};
\node[inner sep=0pt] (B) at (-0.8,8) {};
\node[inner sep=0pt] (C1) at (-0.8,0) {};
\node[inner sep=0pt] (C2') at (-0.8,2.75) {};
\node[inner sep=0pt] (C2) at (-0.8,4.25) {};
\node[inner sep=0pt] (C3) at (-0.8,7) {};

\draw[|-|] (-0.8,-1) -- (-0.8,-1);
\draw[-|] (-0.8,-1) -- (-0.8,0);
\draw[-|] (-0.8,0) -- (-0.8,2.75);
\draw[dashed] (-0.8,2.85) -- (-0.8,4.2);
\draw[|-|] (-0.8,4.25) -- (-0.8,7);
\draw[-|] (-0.8,7) -- (-0.8,8);

\node[left] (zero) at (A) {0};
\node[left] (one) at (B) {1};
\node[left] (omegainf) at (C1) {$\omega_n^{m,\le}$};
\node[left] (omega) at (C2) {$\omega_n^{L^2}$};
\node[left] (omega) at (C2') {$\omega_n^\mathcal{S}$};
\node[left] (omegasup) at (C3) {$\omega_n^{m,\ge}$};

\hypersetup{linkcolor=black}
\node[inner sep=0pt] (SDPsup) at (2,7) {\Large \ref{prog:upperSDPn}};
\node[inner sep=0pt] (D-SDPsup) at (9,7) {\Large \ref{prog:upperDSDPn}};

\node[inner sep=0pt] (LP) at (2,4.25) {\Large \ref{prog:LP}};
\node[inner sep=0pt] (D-LP) at (9,4.25) {\Large \ref{prog:DLP}};

\node[inner sep=0pt] (LPS) at (2,2.75) {\Large \ref{prog:LPS}};
\node[inner sep=0pt] (D-LPS) at (9,2.75) {\Large \ref{prog:DLPS}};

\node[inner sep=0pt] (SDPinf) at (2,0) {\Large \ref{prog:lowerSDPn}};
\node[inner sep=0pt] (D-SDPinf) at (9,0) {\Large \ref{prog:lowerDSDPn}};
\hypersetup{linkcolor=cyan}

\node[inner sep=0pt] (eqsup) at ($.54*(SDPsup)+.46*(D-SDPsup)$) 
{$\substack{\text{Theorem}~\ref{th:sdupper} \\ =\joinrel=}$};

\node[inner sep=0pt] (eq) at ($.54*(LP)+.46*(D-LP)$) 
{$\substack{\text{Theorem}~\ref{th:sdLP} \\ =\joinrel=}$};

\node[inner sep=0pt] (eq) at ($.54*(LPS)+.46*(D-LPS)$) 
{$\substack{\text{Theorem}~\ref{th:sdLP} \\ =\joinrel=}$};

\node[inner sep=0pt] (eqinf) at ($.54*(SDPinf)+.46*(D-SDPinf)$) 
{$\substack{\text{Theorem}~\ref{th:sdlower} \\ =\joinrel=}$};

\node[inner sep=0pt] (CVlu) at ($.5*(SDPsup)+.5*(LP)$)
{$\substack{\text{Theorem~\ref{th:upperCV}} \\\hspace{20pt} \big\downarrow m\rightarrow \infty}$};

\node[inner sep=0pt] (CVld) at ($.5*(SDPinf)+.5*(LPS)$) 
{$\substack{\hspace{20pt} \big\uparrow m\rightarrow \infty \\ \text{Theorem~\ref{th:lowerCV}}}$};

\node[inner sep=0pt] (CVru) at ($.5*(D-SDPsup)+.5*(D-LP)$) 
{$\substack{\text{Theorem~\ref{th:upperCV}} \\\hspace{20pt} \big\downarrow m\rightarrow \infty}$};

\node[inner sep=0pt] (CVrd) at ($.5*(D-SDPinf)+.5*(D-LPS)$) 
{$\substack{\hspace{20pt} \big\uparrow m\rightarrow \infty \\ \text{Theorem~\ref{th:lowerCV}}}$};

\end{tikzpicture}
		\hypersetup{linkcolor=black}
		\caption{Hierarchy of semidefinite relaxations converging to the linear program \ref{prog:LP} and hierarchy of semidefinite restrictions converging to the linear program \ref{prog:LPS}, together with their dual programs. The upper index $m$ denotes the level of the relaxation or restriction. On the left are the associated optimal values. The equal sign denotes strong duality, i.e., equality of optimal values, and the arrows denote convergence of the corresponding sequences of optimal values. Note that the question of the closing the gap between the values of \ref{prog:LP} and \ref{prog:LPS} is left open.
		}
		\hypersetup{linkcolor=cyan}
		\label{fig:structure}
	\end{center}
\end{figure}
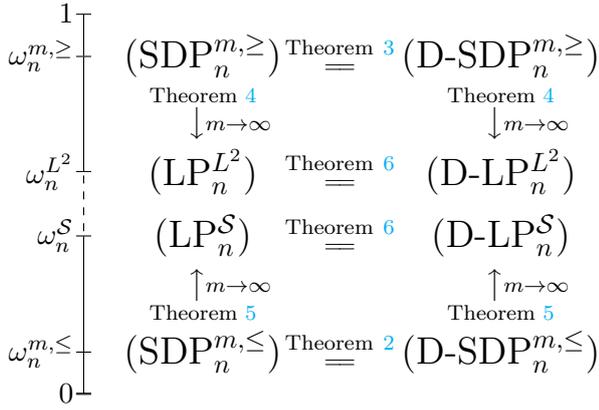

For clarity, we treat the case where the witnesses are given by the fidelity with a single Fock state, corresponding to the case where one entry of the vector $\bm a$ is equal to $1$ and all the other entries are $0$. The generalisation to linear combinations of fidelities with Fock states is straightforward by linearity. 

A linear program is an optimisation problem where variables are linearly constrained \cite{barvinok_02}. This can be expressed as:
\begin{equation}
\label{eq:gen_opti}
\Sup{\bm x\in K}\;f(\bm x),    
\end{equation}
for some linear real-valued function $f$ and some set $K \subset E$ where $E$ is the optimisation space which has the structure of a locally convex topological vector space and $K$ is specified by a set of linear constraints.

An element $\bm x \in E$ which belongs to $K$, i.e., which satisfies all the linear constraints, is called a \textit{feasible plan} or \textit{feasible solution}, and \textit{strictly feasible solution} when it strictly satisfies the constraints. The set of feasible plans is called the \textit{feasible set}. If the supremum in \refeq{gen_opti} is attained, a feasible plan that reaches it is called an \textit{optimal plan} or \textit{optimal solution}.

\subsection{Sketch of the proofs}
\label{sec:sketch}
This section aims at giving an overview of the rather technical proofs that follow. 
Our reasoning can be split into the following steps:

\subsubsection*{Obtaining an infinite linear program and its dual}

\begin{itemize}
    \item We exploit the phase-space rotational invariance of Fock states in Lemma~\ref{lemma:random_deph} to express the computation of threshold values \refeq{threshold} on states that are diagonal in the Fock basis.
    \item This provides a maximisation problem which can be rephrased as an infinite-dimensional linear program \refprog{LP} (resp.\ \refprog{LPS}) on the space of square-integrable functions $L^2(\R_+)$ (resp.\ Schwartz functions $\mathcal S(\R_+)$). We denote $\omega_n^{L^2}$ its optimal value (resp.\ $\omega_n^\mathcal S$).
    \item By duality, we recast this as an optimisation problem on measures which is given by the dual infinite-dimensional linear program \refprog{DLP} (resp.\ \refprog{DLPS}). 
\end{itemize}

\subsubsection*{Hierarchy of relaxations}
See upper part of figure~\ref{fig:structure}.
\begin{itemize}
    \item We relax the program \refprog{LP}: instead of asking for a positive function, we require that this function has a positive inner product with positive polynomials of fixed degree $m$. This degree fixes a level within a hierarchy of relaxations. These constraints can be cast as a positive semidefinite constraint and we thus obtain a hierarchy of semidefinite programs (see \refprog{upperSDPn}).
    \item For each level $m$, we derive the dual program \refprog{upperDSDPn} and show that strong duality holds in Theorem \ref{th:sdupper} by finding a strictly feasible solution of \refprog{upperSDPn}.
    \item We prove the convergence of the hierarchy of semidefinite relaxations towards \refprog{DLP} in Theorem~\ref{th:upperCV} by first showing that the feasible set of \refprog{upperSDPn} is compact. Then, we perform a diagonal extraction on a sequence of optimal solutions of \refprog{upperSDPn} and we finally show that this provides a feasible solution of \refprog{LP} proving that $\underset{m\rightarrow+\infty}{\lim}\omega^{m,\ge}_n=\omega_n^{L_2}$. 
\end{itemize}

\subsubsection*{Convergence of hierarchy of restrictions}
See bottom part of figure~\ref{fig:structure}.
\begin{itemize}
    \item We restrict the program \refprog{LP} (or equivalently the program \refprog{LPS}): instead of optimising over positive functions, we optimise over positive polynomials of fixed degree $m$. Again, this degree fixes a level within a hierarchy of restrictions. Using the fact that univariate positive polynomials are sum-of-squares which can be written as a semidefiniteness constraint, we obtain a hierarchy of semidefinite programs (see \refprog{lowerSDPn}).
    \item For each level $m$, we derive the dual program \refprog{lowerDSDPn} and show that strong duality holds in Theorem \ref{th:sdlower} by finding a strictly feasible solution of \refprog{lowerSDPn}.
    \item We prove the convergence of the hierarchy towards \refprog{DLPS} in Theorem~\ref{th:lowerCV} by first showing that the feasible set of \refprog{lowerDSDPn} is compact. This is highly nontrivial since it requires exhibiting an analytical feasible solution of \refprog{lowerSDPn} which is difficult in general and in the present case. Then, we perform a diagonal extraction on a sequence of optimal solutions of \refprog{lowerDSDPn}. A technicality arises as it does not necessarily provide a feasible solution of \refprog{DLP} which is why we introduce the linear program expressed over Schwartz functions.
    We then show that a diagonal extraction indeed provides a feasible solution of \refprog{DLPS} 
    proving that $\underset{m\rightarrow+\infty}{\lim}\omega^{m,\le}_n=\omega_n^{S}$. 
\end{itemize}

\subsection{Function spaces}
\label{sec:functionspaces}

We now review some function spaces which appear in the following sections, together with a few notations.

Hereafter, the half line of non-negative real numbers is denoted $\R_+$. The space of real square-integrable functions over $\R_+$ is denoted $L^2(\R_+)$ and is equipped with the usual scalar product:
\begin{equation}\label{eq:braket}
\braket{f,g}=\int_{\R_+}{f(x)g(x)dx},
\end{equation}
for $f,g\in L^2(\R_+)$.
This space is isomorphic to the space of square-summable real sequences indexed by $\N$, $l^2(\N)$, by considering the expansion in a countable basis. Such a basis is given, e.g., by the Laguerre functions~\cite{szego1959orthogonal},
modified here by a $(-1)^k$ prefactor to correspond to Fock state Wigner functions:
\begin{equation}\label{eq:Lagf}
    \mathcal L_k(x):=(-1)^kL_k(x)e^{-\frac x2},
\end{equation}
for all $k\in\mathbb N$ and all $x\in\mathbb R_+$, where $L_k(x)=\sum_{l=0}^k\frac{(-1)^l}{l!}\binom klx^l$ is the $k^{\text{th}}$ Laguerre polynomial. 
These functions form an orthonormal basis: for all $p,q\in\mathbb N$, $\braket{\mathcal L_p,\mathcal L_q}=\delta_{pq}$, where $\delta_{pq}$ is the Kronecker symbol.

The space $L^2(\R_+)$ is also isomorphic to its dual space ${L^2}'(\R_+)$ via
the 
Radon--Nikodym theorem \cite{Nikodym1930}:
elements of ${L^2}'(\R_+)$ can be identified by the Lebesgue measure on $\R_+$ times the corresponding function in $L^2(\R_+)$.

We write $\mathcal S(\R_+)$ the space of Schwartz functions over $\R_+$, i.e., the space of $C^{\infty}$ functions that go to $0$ at infinity faster than any inverse polynomial, as do their derivatives. $\mathcal S'(\R_+)$ is its dual space, the space of tempered distributions over $\R_+$. 
Note that $\mathcal S(\R_+)$ is dense in ${L^2}(\R_+)$.
We denote the space of rapidly decreasing real sequences by $\mathcal S(\mathbb N)$ (sequences that go to $0$ at infinity faster than any inverse polynomial), together with its dual space of slowly increasing real sequences $\mathcal S'(\mathbb N)$ (sequences that are upper bounded by a polynomial). The spaces $\mathcal S(\R_+)$ and $\mathcal S(\mathbb N)$ are isomorphic: any Schwartz function over $\R_+$ can be expanded uniquely in the basis of Laguerre functions with a rapidly decreasing sequence of coefficients. Similarly, the spaces $\mathcal S'(\R_+)$ and $\mathcal S'(\mathbb N)$ are also isomorphic: any tempered distribution over $\R_+$ can be written uniquely as a formal series of Laguerre functions with a slowly increasing sequence of coefficients~\cite{guillemot1971developpements}.
We extend the definition of the duality $\braket{\dummy,\dummy}$ in~\refeq{braket} to these spaces.

In order to denote non-negative elements of these spaces, we will use the notations $L^2_+(\R_+)$, ${L_+^2}'(\R_+)$, $\mathcal S_+(\R_+)$ and $\mathcal S_+'(\R_+)$. A distribution $\mu$ in ${L_+^2}'(\R_+)$ (resp.\ in $\mathcal S_+'(\R_+)$) satisfies: $\forall f\in L^2_+(\R_+)$ (resp.\ $\forall f\in \mathcal S_+(\R_+)$), $\braket{\mu,f}\ge0$.

For all $m\in\mathbb N$, we define the following space of truncated series of Laguerre functions over $\R_+$:
\begin{equation}
     \mathcal R_m(\R_+):=\text{span}_{\R}\{\mathcal L_k\}_{0\le k\le m},
\end{equation}
which is equal to the set of real polynomials over $\R_+$ of degree less or equal to $m$ multiplied by the function $x\mapsto e^{-\frac x2}$. We denote by $\mathcal R_{m,+}(\R_+)$ its subset of non-negative elements.

For all $\bm s\in\R^{\N}$, we define the associated formal series of Laguerre functions:
\begin{equation}\label{eq:isomorphism}
    f_{\bm s}:=\sum_{k\ge0}s_k\mathcal L_k,
\end{equation}
with the (formal) relation:
\begin{equation}\label{eq:coefLag}
    s_k=\braket{f_{\bm s},\mathcal L_k},
\end{equation}
for all $k\in\mathbb N$. We refer to $\bm s$ as the sequence of Laguerre moments of $f_{\bm s}$. We extend this definition to finite sequences by completing these sequences with zeros.
For $m\in\mathbb N$, we also define the matrix $A_{\bm s}$ (thus omitting the dependence in $m$) by
\begin{equation}\label{eq:momentmatrix}
    (A_{\bm s})_{0\le i,j\le m} := \begin{cases} 
      \sum\limits_{k=0}^l s_k \binom lkl! &\text{if } i+j=2l, \\
      0 &\text{otherwise.}
   \end{cases}
\end{equation}
$A_{\bm s}$ can be seen as the Laguerre moment matrix of the measure $f_{\bm s}$.
In what follows, we use standard techniques relating to the Stieltjes moment problem~\cite{reed1975ii}, which seeks conditions for a real sequence $\bm\nu=(\nu_k)_{k\in\N}\in\R^{\N}$ to be the sequence of moments $\int_{\R_+}x^kd\nu(x)$ of a non-negative distribution $\nu$ over $\R_+$. We adapt these techniques to the basis of Laguerre functions, rather than the canonical basis. In particular, we make use of the following result:

\begin{theorem}\label{th:RHLaguerre}
Let $\bm\mu=(\mu_k)_{k\in\N}\in\R^{\N}$. The sequence $\bm\mu$ is the sequence of Laguerre moments $\int_{\R_+}\mathcal L_k(x)d\mu(x)$ of a non-negative distribution $\mu$ supported on $\R_+$ if and only if
\begin{equation}
    \forall m\in\N,\forall g\in\mathcal R_{m,+}(\R_+),\;\braket{f_{\bm\mu},g}\ge0.
\end{equation}
\end{theorem}

\noindent We give a proof in Appendix~\ref{sec:app_RHtheo} for this result, which is based on the classic Riesz--Haviland theorem~\cite{riesz1923probleme,haviland1936momentum}.

\subsection{Linear program}
\label{sec:LP}

In this section, we phrase the computation of the witness threshold value introduced in \refeq{threshold} as an infinite-dimensional linear program, in the case where one entry of the vector $\bm a$ is equal to $1$ and all the other entries are $0$, the generalisation being straightforward by linearity.

Formally, we fix hereafter $n\in\mathbb N^*$ and we look for the witnesses threshold value $\omega_n$
\footnote{Here we write generically the threshold value as $\omega_n$ while we use the more precise notation $\omega_n^{L^2}$ (resp.\ $\omega_n^{\mathcal S}$) to refer to its computation in the space of square-integrable functions (resp.\ Schwartz functions).}
defined as
\begin{equation}\label{eq:formal_problem}
  \omega_n:=\sup_{\substack{\rho\in\mathcal D(\mathcal H)\\W_\rho\ge0}}\braket{n|\rho|n}.
\end{equation}
This is the maximal values such that for all states $\rho\in\mathcal D(\mathcal H)$:
\begin{equation}
  \braket{n|\rho|n}>\omega_n\quad\Rightarrow\quad\exists\alpha\in\mathbb C,\;W_\rho(\alpha)<0.
\end{equation}
Let $\mathcal C(\mathcal H)$ be the set of states that are invariant under phase-space rotations:
\begin{equation}
    \begin{aligned}
    \mathcal C(\mathcal H)\!:=\!\{\sigma\!&\in\!\mathcal D(\mathcal H) \text{ such that:} \\
    & \forall\varphi\in[0,2\pi], \eu^{\im\varphi \hat{n}}\sigma\eu^{-\im\varphi \hat{n}}\!=\!\sigma\},
    \end{aligned}
\end{equation}
where $\hat n=\hat a^\dag\hat a$ is the number operator.
The witnesses corresponding to the fidelity with a single Fock state feature a rotational symmetry in phase space, which we exploit in the following lemma. 

\begin{lemma}\label{lemma:random_deph}
The threshold value in~\refeq{formal_problem} can be expressed as
\begin{equation}
  \omega_n=\sup_{\substack{\sigma\in\mathcal C(\mathcal H)\\W_\sigma\ge0}}\braket{n|\sigma|n}.
\end{equation}
\end{lemma}

\begin{proof}
Let $\rho\in\mathcal D(\mathcal H)$. We start by applying a random dephasing to the state $\rho$:
\begin{equation}
    \sigma = \int_0^{2\pi}\frac{\dd\varphi}{2\pi}\eu^{\im\varphi \hat{n}}\rho\eu^{-\im\varphi \hat{n}}\in \mathcal C(\mathcal H).
\end{equation}
The random dephasing does not change the fidelity with any Fock state because of the rotational symmetry in phase space of the latter, i.e.,
\begin{equation}
    \forall n \in \N, \, \braket{n|\sigma|n} = \braket{n|\rho|n}.
\end{equation} 
Moreover, it can only decrease the maximum negativity of the Wigner function: for all $\alpha\in\mathbb C$,
\begin{equation}\label{eq:inequality_dephasing}
    \begin{aligned}
        W_{\sigma}(\alpha)
        &=\int_0^{2\pi} \frac{\dd\varphi}{2\pi}W_{e^{\im \varphi \hat{n}} \rho  e^{- \im \varphi \hat{n}}}(\alpha) \\
        &=\int_0^{2\pi} \frac{\dd\varphi}{2\pi}W_{\rho}(\alpha \eu^{\im\varphi}) \\
        &\ge \min_{\varphi\in[0,2\pi]}W_{\rho}(\alpha \eu^{\im\varphi}) \\
        &\ge \min_{\beta\in\mathbb C}W_{\rho}(\beta), 
    \end{aligned} 
\end{equation}
and taking the minimum over all $\alpha\in\mathbb C$ then gives
\begin{equation}
    \min_{\alpha\in\mathbb C}W_{\sigma}(\alpha)\ge\min_{\beta\in\mathbb C}W_{\rho}(\beta).
\end{equation}
In particular, applying a random dephasing to a Wigner positive state yields a Wigner positive mixtures of Fock states, which is invariant under phase-space rotations. Hence, we can restrict without loss of generality to states that are invariant under phase-space rotations when looking for the maximum fidelity of Wigner positive states with a given Fock state $\ket n$.
\end{proof}

\noindent Lemma~\ref{lemma:random_deph} ensures that the supremum in \refeq{formal_problem} can be computed over states that have a rotational symmetry in phase space. Such states $\sigma$ can be expanded diagonally in the Fock basis:
\begin{equation}
    \sigma = \sum_{k=0}^{\infty} F_k \ket{k}\!\bra{k},
\end{equation}
with $\sum_kF_k=1$ and $0 \leq F_k \leq 1$ for all $k\in\N$.
By linearity of the Wigner function: 
\begin{equation}
    \forall \alpha \in \C,\; W_{\sigma}(\alpha) = \sum_k F_k W_k(\alpha),
\end{equation}
where $W_k$ is the Wigner function of the $k^{th}$ Fock state~\cite{kenfack2004negativity}: 
\begin{equation}
    \forall \alpha \in \C, \, W_k(\alpha) = \frac{2}{\pi} \mathcal L_k(4|\alpha|^2),
\end{equation}
with $\mathcal L_k$ the $k^{\text{th}}$ Laguerre function, defined in~\refeq{Lagf}. As noted before, Fock states are invariant under phase-space rotations: their Wigner function only depends on the amplitude of the phase-space point considered. We fix $x= 4 \vert \alpha \vert^2 \in \R^+$ hereafter.

As we consider the optimisation over $L^2(\R_+)$ functions, we will denote the corresponding threshold value by $\omega_n^{L^2}$. With Lemma \ref{lemma:random_deph}, the computation of $\omega_n^{L^2}$ can thus be expressed as the following infinite-dimensional linear program:
\begin{fleqn}
\begin{equation*}
    \label{prog:LP}
        \tag*{(LP$_n^{L^2}$)}
        \begin{aligned}
            & \quad \text{Find }  (F_k)_{k \in \N} \in \ell^2(\N) \\
            & \quad \text{maximising } F_n \\
            & \quad \text{subject to} \\
            &\begin{dcases}
            \sum_k F_k = 1  \\
            \forall k \in \N, \;F_k \geq 0 \phsp{-.8}{.8} \\
            \forall x \in \R_+, \;\sum_k F_k\mathcal L_k(x) \geq 0.
            \end{dcases}
        \end{aligned}
\end{equation*}
\end{fleqn}
The first constraint ensures unit trace of the corresponding state $\sigma$, the second one ensures that its fidelity with each Fock state is non-negative, and the last one ensures that its Wigner function $W_{\sigma}$ is non-negative. Note that $\omega_n>0$ for all $n\in\N^*$, by considering a mixture of $\ket0$ and $\ket n$ with the vacuum component close enough to $1$. Note that, for consistency of notations, we use $L$ instead of $\ell$ in the label of the program.

We refer the interested reader to Appendix \ref{sec:app_std_form}
where this program is expressed in the canonical form of infinite-dimensional linear programs \cite{barvinok_02} and its dual is derived. This dual linear program reads:
\begin{fleqn}
\begin{equation*}
    \label{prog:DLP}
    \tag*{(D-LP$_n^{L^2}$)}
        \begin{aligned}
            & \quad \text{Find } y \in \R \text{ and } \mu \in {L^2}'(\R_+)\\
            & \quad \text{minimising } y\\
            & \quad \text{subject to}  \\
            & \begin{dcases}
            \forall k \neq n \in \N,\; y \geq \int_{\R_+}{\mathcal L_k}{d\mu}  \\
            y \geq 1 + \int_{\R_+}{\mathcal L_n}{d\mu} \phsp{-.8}{.8}\\
            \forall f \in L^2_+(\R_+), \; \langle \mu,f \rangle \geq  0.
            \end{dcases}
        \end{aligned}
\end{equation*}
\end{fleqn}
We also consider a restriction of \refprog{LP} by optimising over Schwartz functions rather than square-integrable functions as $\mathcal S(\R_+) \subset L^2(\R_+)$. The primal program becomes: 
\begin{fleqn}
\begin{equation*}
    \label{prog:LPS}
        \tag*{(LP$_n^{\mathcal S}$)}
        \begin{aligned}
            & \quad \text{Find } (F_k)_{k \in \N} \in \mathcal S(\N) \\
            & \quad \text{maximising } F_n \\
            & \quad \text{subject to} \\
            &\begin{dcases}
            \sum_k F_k = 1  \\
            \forall k \in \N, \;F_k \geq 0 \phsp{-.8}{.8} \\
            \forall x \in \R_+, \;\sum_k F_k\mathcal L_k(x) \geq 0,
            \end{dcases}
        \end{aligned}
\end{equation*}
\end{fleqn}
We denote its value by $\omega_n^\mathcal S$.
Because it is a restriction of \refprog{LP}, we have that $\omega_n^{L^2} \geq \omega_n^\mathcal S$.
Its dual linear program can be expressed as:
\begin{fleqn}
\begin{equation*}
    \label{prog:DLPS}
    \tag*{(D-LP$_n^{\mathcal S}$)}
        \begin{aligned}
            & \quad \text{Find } y \in \R \text{ and } \mu \in\mathcal S'(\R_+) \\
            & \quad \text{minimising } y \\
            & \quad \text{subject to}  \\
            & \begin{dcases}
            \forall k \neq n \in \N,\; y \geq \int_{\R_+}{\mathcal L_k}{d\mu}  \\
            y \geq 1 + \int_{\R_+}{\mathcal L_n}{d\mu} \phsp{-.8}{.8}\\
            \forall f \in \mathcal S_+(\R_+), \; \langle \mu,f \rangle \geq  0.
            \end{dcases}
        \end{aligned}
\end{equation*}
\end{fleqn}

We will show later that strong duality holds for the two pairs of programs.
However weak duality of linear programming already ensures that the optimal value $\omega_n^{L^2}$ of \refprog{LP} (resp.\ $\omega_n^{\mathcal S}$ of \refprog{LPS}) is upper bounded by the optimal value of \refprog{DLP} (resp.\ \refprog{DLPS}). Hence, a possible way of solving the optimisation \refprog{LP} is to exhibit a feasible solution for \refprog{LP} and a feasible solution for \refprog{DLP} that have the same value. 

For $n=1$, choosing $(F_k)_{k\in\mathbb N}=(\frac12,\frac12,0,0,\dots)$ gives a feasible solution for~(LP$_1^{L^2}$) 
(resp.\ (LP$_1^{\mathcal S}$))
with the value $\frac12$, while choosing $(y,\mu)=(\frac12,\frac12\delta(x))$, where $\delta$ is the Dirac delta function\footnote{Technically $\delta\notin{L^2}'(\R_+)$, but the result holds by considering a sequence of functions converging to a Dirac delta.} 
over $\R_+$, 
gives a feasible solution for~(D-LP$_1^{L^2}$) 
(resp.\ (D-LP$_1^{\mathcal S}$))
with the value $\frac12$. 
This shows that $\omega_1^{L^S} = \omega_1^{\mathcal S} =\frac12$.

Similarly, for $n=2$, choosing $(F_k)_{k\in\mathbb N}=(\frac12,0,\frac12,0,0,\dots)$ gives a feasible solution for~(LP$_2^{L^2}$) 
(resp. LP$_2^{\mathcal S}$)
with the value $\frac12$, while choosing $(y,\mu)=(\frac12,\frac e2\delta(x-2))$ gives a feasible solution for~(D-LP$_2$)
(resp. (D-LP$_2^{\mathcal S}$))
, up to a conjecture\footnote{We checked numerically the corresponding constraints $|L_k(2)|\le1$ for $k$ up to $10^3$ and, considering asymptotic behaviors, we conjecture that these hold for all $k\ge0$.}, with the value $\frac12$. This shows that $\omega_2^{L^2} = \omega_2^{S} = \frac12$.

While this approach is sensible for small values of $n$, finding optimal analytical solutions for higher values of $n$ seems highly nontrivial. Moreover, the infinite number of variables prevents us from performing the optimisation \refprog{LP} numerically. A natural workaround is to find finite-dimensional relaxations or restrictions of the original problem---thus providing upper and lower bounds for the optimal value $\omega_n^{L^2}$, respectively. This is the approach we follow in the next section.

\subsection{Hierarchies of semidefinite programs}
\label{sec:SDP}

Semidefinite programming is a convex optimisation technique in the cone of positive semidefinite matrices. It comes with a duality theory: if a primal semidefinite program is a maximisation problem, then one may deduce a dual minimisation problem which is again a semidefinite program. Like linear programs, these programs satisfy a weak duality condition: the optimal value of the primal problem is upper bounded by the optimal value of the dual program. The difference between the optimal values is called the duality gap. When there is no duality gap, we say that there is strong duality between the programs.

\subsubsection{Preliminaries}

In this section, we introduce preliminary technical lemmas and we refer the reader to Appendix~\ref{sec:app_prooflemmas} for their proofs.

We recall the following standard result, which comes from the fact that any univariate polynomial non-negative over $\R$ can be written as a sum of squares:

\begin{lemma}[\cite{hilbert1888darstellung}]\label{lem:pospolyR}
Let $p\in\mathbb N$ and let $P$ be a univariate polynomial of degree $2p$. Let $X=(1,x,\dots,x^p)$ be the vector of monomials. Then, $P$ is non-negative over $\R$ if and only if there exists a real $(p+1)\times(p+1)$ positive semidefinite matrix $Q$ such that for all $x\in\mathbb R$,
\begin{equation}
    P(x)=X^TQX.
\end{equation}
\end{lemma}

\noindent From this lemma we deduce the following characterisation of non-negative polynomials over $\R_+$:

\begin{lemma}\label{lem:pospolyR+}
Non-negative polynomials on $\R_+$ can be written as sums of polynomials of the form $\sum_{l=0}^px^l\sum_{i+j=2l}y_iy_j$, where $p\in\mathbb N$ and $y_i\in\mathbb R$, for all $0\le i\le p$.
\end{lemma}

\noindent Note that the characterisation in Lemma~\ref{lem:pospolyR+} differs from that of Stieltjes~\cite{reed1975ii}, which expresses non-negative polynomials over $\R_+$ as $x\mapsto A_1(x)+xA_2(x)$, where $A_1$ and $A_2$ are sums of squares. 
This slightly slows down the numerical resolution, which does not matter given the size of the programs considered. At the same time, this allows us to obtain more compact expressions for the semidefinite programs. 

We use the characterisation of Lemma~\ref{lem:pospolyR+} to obtain the following result: for $\bm s\in\R^{\N}$, the fact that the series $f_{\bm s}$ defined in~\refeq{isomorphism} has non-negative scalar product with non-negative truncated Laguerre series up to degree $m$ can be expressed as a positive semidefinite constraint involving the matrix $A_{\bm s}$ defined in~\refeq{momentmatrix}. Formally:

\begin{lemma}\label{lem:momentmatrix}
Let $m\ge n$ and let $\bm s\in\R^{\N}$. The following propositions are equivalent:
\begin{enumerate}[label=(\roman*)]
\item $\forall g\in\mathcal R_{m,+}(\R_+),\;\braket{f_{\bm s},g}\ge0$,
\item $A_{\bm s}\succeq0$.
\end{enumerate}
\end{lemma}

\noindent Using these results, we derive hierarchies of semidefinite relaxations and restrictions for the infinite-dimensional linear program \refprog{LP} in the following sections.

\subsubsection{Semidefinite relaxations}
\label{sec:SDPupper}

One way to obtain a relaxation of \refprog{LP} is to relax the constraint:
\begin{equation}\label{eq:conditionpos}
    \forall x\in\mathbb R_+,\;f_{\bm F}(x)=\sum_kF_k\mathcal L_k(x)\ge0.    
\end{equation}
Instead, one may impose the weaker constraint:
\begin{equation}\label{eq:conditionposscal}
    \forall g\in\mathcal R_{m,+}(\R_+),\;\braket{f_{\bm F},g}\ge0,
\end{equation}
for some fixed $m\ge n$. By Lemma~\ref{lem:momentmatrix}, this constraint may in turn be expressed as a positive semidefinite constraint on the $(m+1)\times(m+1)$ matrix $A_{\bm F}$ defined in~\refeq{momentmatrix}. 

Each choice of $m$ thus leads to a different semidefinite program, whose optimal value gets closer to $\omega_n$ as $m$ increases (since the constraint~\eqref{eq:conditionposscal} gets stronger when $m$ increases). Moreover, when replacing the constraint~\eqref{eq:conditionpos} by the constraint~\eqref{eq:conditionposscal} in \refprog{LP}, for $l>m$ the variables $F_l$  do not appear in the matrix $A_{\bm F}$, and are constrained only by $\sum_kF_k=1$ and $F_l\ge0$. Since $m\ge n$ and the optimal value is $F_n$, we may thus set $F_l=0$ for $l>m$ without loss of generality.
This gives a hierarchy of finite-dimensional semidefinite relaxations for \refprog{LP}, and the semidefinite relaxation of order $m$ is given by:
\begin{fleqn}
\begin{equation*}
   \label{prog:upperSDPn}
    \tag*{$(\text{SDP}^{m,\geq}_n)$}
        \begin{aligned}
            & \quad \text{Find } A\in\text{Sym}_{m+1} \text{ and } \bm F\in\R^{m+1} \\
            & \quad \text{maximising } F_n \\
            & \quad \text{subject to}\\
            & \begin{dcases}
                        \textstyle \sum_{k=0}^{m} F_k = 1  \\
                        \textstyle \forall k\le m, \; F_k \geq 0 \\
                        \textstyle \forall l\le m,\forall i+j=2l,\; A_{ij}=\sum\limits_{k=0}^l F_k \binom lkl!\hspace{-3cm}\\
                        \textstyle \forall l\in\llbracket1,m\rrbracket,\forall i+j=2l-1,\;A_{ij}=0\hspace{-3cm}\\
                        \textstyle A\succeq 0.
              \end{dcases}
        \end{aligned}
\end{equation*}
\end{fleqn}
Let us denote its optimal value by $\omega^{m,\geq}_n$. The sequence $\{\omega^{m,\geq}_n\}_{m\ge n}$ is a decreasing sequence and for all $m\ge n$, we have $\omega_n^{L^2}\le\omega^{m,\geq}_n$.

For each $m\ge n$, the program \refprog{upperSDPn} has a dual semidefinite program which is given by (see Appendix~\ref{sec:app_dualSDP} for a detailed derivation):
\begin{fleqn}
\begin{equation*}
    \label{prog:upperDSDPn}
    \tag*{$(\text{D-SDP}^{m,\geq}_n)$}
        \begin{aligned}
            & \quad \text{Find } Q \in \text{Sym}_{m+1}, \bm{\mu} \in \R^{m+1} \text{ and } y \in \R \hspace{-3cm}\\
            & \quad \text{minimising } y \\
            & \quad \text{subject to} \\
            & \begin{dcases} 
            \textstyle y\ge1+\mu_n  \\
            \textstyle \forall k\le m, \; y\ge \mu_k \\
            \textstyle \forall l\le m, \sum\limits_{i+j=2l} Q_{ij} = \textstyle \sum\limits_{k=l}^{m} \frac{(-1)^{k+l}}{l!}\binom kl \mu_k\hspace{-2cm} \\
            \textstyle Q \succeq 0.
            \end{dcases}
        \end{aligned}
\end{equation*}
\end{fleqn}
%
We show in Theorem~\ref{th:sdupper} that strong duality holds between the primal and the dual versions of this semidefinite program. In particular, numerical computations with either of these programs will yield the same optimal value.

\subsubsection{Semidefinite restrictions}
\label{sec:SDPlower}

A trivial way to obtain a restriction of \refprog{LP} (or equivalently of \refprog{LPS}) is to impose $F_l=0$ for $l>m$, for some $m\ge n$. What is less trivial is that this yields a finite-dimensional semidefinite program.
Indeed, the constraint~\eqref{eq:conditionpos} becomes
\begin{equation}
    \forall x\in\mathbb R_+,\;\sum_{k=0}^mF_k\mathcal L_k(x)\ge0,
\end{equation}
or equivalently:
\begin{equation}
    \forall x\in\mathbb R,\;\sum_{k=0}^m(-1)^kF_kL_k(x^2)\ge0,
\end{equation}
where we used~\refeq{Lagf}. By Lemma~\ref{lem:pospolyR}, writing $X=(1,x,\dots,x^m)$, this is equivalent to the existence of a positive semidefinite matrix $Q=(Q_{ij})_{0\le i,j\le m}$ such that for all $x\in\mathbb R$,
\begin{equation}\label{eq:FtoQ}
    \begin{aligned}
        \sum_{k=0}^m(-1)^kF_kL_k(x^2)&=X^TQX\\
        &=\sum_{l=0}^mx^l\sum_{i+j=l}Q_{ij}.
    \end{aligned}
\end{equation}
This is in turn equivalent to the linear constraints:
\begin{equation}
    \begin{aligned}
        &\forall l \in \llbracket 1,m \rrbracket, \;\sum\limits_{i+j=2l-1}Q_{ij} = 0,\\
        &\text{and}\\
        &\forall l\le m,\sum\limits_{i+j=2l}Q_{ij} = \frac{(-1)^l}{l!}\sum\limits_{k=l}^{m} (-1)^k \binom kl F_k,\hspace{-1cm}
    \end{aligned}
\end{equation}
by identifying the coefficients in front of each monomial in~\refeq{FtoQ}.

Hence, the restriction of \refprog{LP} obtained by imposing $F_l=0$ for $l>m$, for a fixed $m\ge n$, is a semidefinite program given by:
\begin{fleqn}
\begin{equation*}
    \label{prog:lowerSDPn}
    \tag*{$(\text{SDP}^{m,\leq}_n)$}
        \begin{aligned}
            & \quad \text{Find } Q\in\text{Sym}_{m+1} \text{ and } \bm{F}\in\R^{m+1}\\
            & \quad \text{maximising } F_n \\
            & \quad \text{subject to} \\
            &\begin{dcases}
                \textstyle \sum_{k=0}^{m} F_k = 1  \\
                \textstyle \forall k\le m, \; F_k \geq 0 \\                     \textstyle \forall l \in \llbracket 1,m \rrbracket, \;\sum\limits_{i+j=2l-1}Q_{ij} = 0\hspace{-2cm}\\
                \textstyle \forall l\le m,\sum\limits_{i+j=2l}Q_{ij} = \sum\limits_{k=l}^{m} \frac{(-1)^{k+l}}{l!}  \binom kl F_k \hspace{-2cm}\\
                \textstyle Q \succeq 0.
             \end{dcases}
        \end{aligned}
\end{equation*}
\end{fleqn}
Let us denote its optimal value by $\omega^{m,\leq}_n$.
Each choice of $m$ leads to a different semidefinite restriction of \refprog{LP} and \refprog{LPS}, whose optimal value gets closer to $\omega_n^\mathcal S$ as $m$ increases (since the constraint $F_l=0$ for $l>m$ gets weaker when $m$ increases). The sequence $\{\omega^{m,\leq}_n\}_{m\ge n}$ is thus an increasing sequence and for all $m\ge n$, we have $\omega^{m,\leq}_n\le\omega_n^\mathcal S\le\omega_n^{L^2}$.

For each $m\ge n$, the program \refprog{lowerSDPn} has a dual semidefinite program which is given by (see Appendix~\ref{sec:app_dualSDP} for a detailed derivation):
\begin{fleqn}
\begin{equation*}
   \label{prog:lowerDSDPn}
      \tag*{$(\text{D-SDP}^{m,\leq}_n)$}
      \begin{aligned}
            & \quad \text{Find } A\in\text{Sym}_{m+1}, \bm{\mu} \in \R^{m+1} \text{ and } y \in \R \hspace{-2cm}\\
            & \quad \text{minimising } y\\
            & \quad \text{subject to}\\
            &\begin{dcases}
                \textstyle y \geq 1 + \mu_n \\
                \textstyle \forall k \le m,\; y \geq \mu_k \\
                \textstyle \forall l\le m,\forall i+j=2l,\;A_{ij}=\sum\limits_{k=0}^l \mu_k \binom lkl!\hspace{-2cm}\\
                \textstyle A \succeq 0.
            \end{dcases}
        \end{aligned}
\end{equation*}
\end{fleqn}
We show in Theorem~\ref{th:sdlower} that strong duality holds between the primal and the dual versions of this semidefinite program.

\subsubsection{Strong duality of semidefinite programs}
\label{sec:SDPsd}

We fix $m\ge n$ and we show that strong duality holds both for the semidefinite restriction \refprog{lowerSDPn} and the semidefinite relaxation \refprog{upperSDPn}.

We first consider the semidefinite restrictions:

\begin{theorem}\label{th:sdlower}
Strong duality holds between the programs \refprog{lowerSDPn} and \refprog{lowerDSDPn}.
\end{theorem}

\begin{proof}
We make use of Slater condition for (finite-dimensional) semidefinite programs: strict feasibility of \refprog{lowerSDPn} implies strong duality between \refprog{lowerSDPn} and \refprog{lowerDSDPn}. 

In order to obtain a strictly feasible solution, we define $Q := \frac1{2^{m+1}-1}\text{Diag}_{k=0,\dots,m}(\frac1{k!}) \in \text{Sym}_{m+1}$ and $\bm F=(F_0,\dots,F_m)\in\R^{m+1}$, where for all $k \in \llbracket 0,m \rrbracket$, $F_k := \frac1{2^{m+1}-1} \binom{m+1}{k+1}$. 
Then $Q \succ 0$ and $F_k>0$ for all $k \in \llbracket 0,m \rrbracket$. Moreover, we have
\begin{equation}
    \begin{aligned}
        \sum_{k=0}^m F_k &= \frac1{2^{m+1}-1}\sum_{k=0}^m \binom{m+1}{k+1} \\
        &= \frac1{2^{m+1}-1}\left(\sum_{k=0}^{m+1} \binom{m+1}k-1\right)\\
        &= 1.
    \end{aligned}
\end{equation}
We also have $\sum_{i+j=2l-1}Q_{ij} = 0$ for all $l\in\llbracket1,m\rrbracket$, since $Q$ is diagonal. Furthermore, for all $l\le m$,
\begin{equation}
    \sum_{i+j=2l}Q_{ij}=Q_{ll}=\frac1{2^{m+1}-1}\frac1{l!},
\end{equation}
and
\begin{fleqn}
\begin{equation}
    \begin{aligned}
        &\frac{(-1)^l}{l!}\sum\limits_{k=l}^{m} (-1)^k  \binom kl F_k\\
        &=\frac1{2^{m+1}-1}\frac1{l!}\sum_{k=l}^m(-1)^{k-l}\binom km\binom{m+1}{k+1}\\
        &=\frac1{2^{m+1}-1}\frac1{l!}\binom ml\sum_{q=0}^{m-l}(-1)^q\frac{m+1}{q+l+1}\binom{m-l}q\hspace{-1cm}\\
        &=\frac1{2^{m+1}-1}\frac1{l!},
    \end{aligned}
\end{equation}
\end{fleqn}
where we used~\cite[(1.41)]{gould1972combinatorial} in the last line.
Therefore, $(Q,\bm F)$ is a strictly feasible solution of \refprog{lowerSDPn}, which implies strong duality.
\end{proof}

\noindent As a consequence of the proof of Theorem~\ref{th:sdlower}, we also obtain strong duality for the semidefinite relaxations: 

\begin{theorem}\label{th:sdupper}
Strong duality holds between the programs \refprog{upperSDPn} and \refprog{upperDSDPn}.
\end{theorem}

\begin{proof}
The program \refprog{upperSDPn} is a relaxation of \refprog{LP} and \refprog{lowerSDPn} is a restriction of \refprog{LP}, so \refprog{upperSDPn} is a relaxation of \refprog{lowerSDPn}. Hence, the strictly feasible solution of \refprog{lowerSDPn} derived in the proof of Theorem~\ref{th:sdlower} yields a strictly feasible solution $(A,\bm F)$ for \refprog{upperSDPn}: we set
$\bm F=(F_0,\dots,F_m)\in\R^{m+1}$, where for all $k \in \llbracket 0,m \rrbracket$, $F_k := \frac1{2^{m+1}-1} \binom{m+1}{k+1}$ and $A=A_{\bm F}\in\text{Sym}_{m+1}$, where $A_{\bm F}$ is defined in~\refeq{momentmatrix}. 

With Slater condition, this shows again that strong duality holds between the programs \refprog{upperSDPn} and \refprog{upperDSDPn}.
\end{proof}

\subsection{Convergence of the hierarchies of semidefinite programs}
\label{sec:CVproof}

From the previous sections, for $m\ge n$ the optimal values $\omega^{m,\geq}_n$ and $\omega^{m,\leq}_n$ of \refprog{upperSDPn} and \refprog{lowerSDPn} form decreasing and increasing sequences, respectively, which satisfy
\begin{equation}\label{eq:sandwich}
    0\le\omega^{m,\leq}_n\le\omega_n^{\mathcal S}\le\omega_n^{L^2}\le\omega^{m,\geq}_n\le1.
\end{equation}
Recall that $\omega_n^{L^2}$ is the optimal value of \refprog{LP} while $\omega_n^\mathcal S$ is the optimal value of \refprog{LPS}. These sequences thus both converge, and the remaining question is whether $(\omega^{m,\geq}_n)_m$ converges to $\omega_n^{L^2}$ and $(\omega^{m,\leq}_n)_m$ converges to $\omega_n^{\mathcal S}$. 
In this section, we show that this is indeed the case.

\subsubsection{Convergence of the sequence of upper bounds}

\begin{theorem}\label{th:upperCV}
The decreasing sequence of optimal values of \refprog{upperSDPn} converges to the optimal value of \refprog{LP}:
\begin{equation}
    \lim_{m\rightarrow+\infty}\omega^{m,\ge}_n=\omega_n^{L^2}.
\end{equation}
\end{theorem}

\noindent In order to prove this theorem, we extract a limit from a sequence of optimal solutions of \refprog{upperSDPn}, for $m\ge n$, and we show using Theorem~\ref{th:RHLaguerre} that it provides a feasible solution of \refprog{LP}.

\begin{proof}
For all $m\ge n$, the feasible set of \refprog{upperSDPn} is non-empty (consider, e.g., $\bm F=(1,0,0,\dots,0)\in\R^{m+1}$). Moreover, due to the constraints $\sum_{k=0}^mF_k=1$ and $F_k\ge0$ for all $k\le m$, the feasible set of \refprog{upperSDPn} is compact. Hence, the program \refprog{upperSDPn} has feasible optimal solutions, for all $m\ge n$, by diagonal extraction.

The matrix $A$ in \refprog{upperSDPn} is entirely fixed by the choice of $\bm F$. Let $(\bm F^m)_{m\ge n}$ be a sequence of optimal solutions of \refprog{upperSDPn}, for $m\ge n$. 
For each $m\ge n$, we have by optimality that $F_n^m=\omega_n^{m,\ge}$, and the sequence $(F^m_n)_{m\ge n}$ converges. We complete each tuple $\bm F^m=(F_0^m,F_1^m,\dots,F_m^m)\in\R^{m+1}$ with zeros to obtain a sequence in $\R^{\N}$, which we still denote $\bm F^m=(F_0^m,F_1^m,\dots,F_m^m,0,0,\dots)\in\R^{\N}$.

Performing a diagonal extraction $\phi$ on the sequence of optimal solutions $(\bm F^m)_{m\ge n}$, we obtain a sequence of sequences $(\bm F^{\phi(m)})_{m\ge n}$ such that each sequence $(F^{\phi(m)}_k)_{m\ge n}$ converges when $m\rightarrow+\infty$, for all $k\in\N$. Let $F_k$ denote its limit, for each $k\in\N$. We write $\bm F=(F_k)_{k\in\N}\in\R^{\N}$ the sequence of limits. 

For all $m\ge n$, $F_k^{\phi(m)}\ge0$ for all $k\in\mathbb N$ and $\sum_kF_k^{\phi(m)}=1$, so taking $m\rightarrow+\infty$ we obtain $F_k\ge0$ for all $k\in\mathbb N$, and $\sum_kF_k\le1$. Moreover,
\begin{equation}\label{eq:Fnlim}
    F_n=\lim_{m\rightarrow+\infty}F_n^m=\lim_{m\rightarrow+\infty}\omega^{m,\ge}_n.
\end{equation}
For all $m\ge n$, we have $\omega^{m,\ge}_n\ge\omega_n^{L^2}$, so $F_n\ge\omega_n^{L^2}>0$. In particular, $\sum_kF_k>0$, so without loss of generality we may assume that $\sum_kF_k=1$ (otherwise we can always replace $F_k$ by $\frac{F_k}{\sum_lF_l}$).

Let $f_{\bm F}=\sum_kF_k\mathcal L_k\in L^2(\R_+)$. By construction we have:
\begin{equation}
    \forall m\ge n,\forall g\in\mathcal R_{m,+}(\R_+),\;\braket{f_{\bm F},g}\ge0.
\end{equation}
Hence, by Theorem~\ref{th:RHLaguerre}, $\bm F$ is the sequence of Laguerre moments of a non-negative distribution over $\R_+$ (the Lebesgue measure times the function $f_{\bm F}$). In particular,
\begin{equation}
    \forall x \in \R_+, \;f_{\bm F}(x)=\sum_k F_k\mathcal L_k(x) \geq 0.    
\end{equation}
With the constraints $F_k\ge0$ for all $k\in\mathbb N$, and $\sum_kF_k\le1$, this implies that $\bm F$ is a feasible solution of \refprog{LP}, and in particular $F_n\le\omega_n^{L^2}$, since \refprog{LP} is a maximisation problem. Since we already had $F_n\ge\omega_n^{L^2}$ we obtain with~\refeq{Fnlim}:
\begin{equation}
    \lim_{m\rightarrow+\infty}\omega^{m,\ge}_n=\omega_n^{L^2},
\end{equation}
which concludes the proof.
\end{proof}
\noindent This immediately implies strong duality for programs \refprog{LP} and \refprog{DLP} because \begin{enumerate}[label=(\roman*)] 
\item we have weak duality between those programs so that the optimal value $\omega'^{L^2}_n$ of \refprog{DLP} upper bounds the optimal value of \refprog{LP} $\omega^{L^2}_n$ \ie $\omega^{L^2}_n \le \omega'^{L^2}_n$; 
\item we have strong duality between \refprog{upperSDPn} and \refprog{upperDSDPn} by Theorem~\ref{th:sdupper};
\item \refprog{upperSDPn} is a relaxation of \refprog{LP} so that $\forall m,\;\omega^{m,\ge}_n \ge\omega_n^{L^2}$;
\item \refprog{upperDSDPn} is a restriction of \refprog{DLP} so that $\forall m,\omega^{m,\ge}_n \ge\omega'^{L^2}_n\ge\omega^{L^2}_n$;
\item we showed that the optimal value of the hierarchy \refprog{upperSDPn} converges to $\omega^{L^2}_n$ \ie $\lim_{m\rightarrow+\infty}\omega^{m,\ge}_n=\omega_n^{L^2}$.
\end{enumerate}

\subsubsection{Convergence of the sequence of lower bounds}
 
\begin{theorem}\label{th:lowerCV}
The increasing sequence of optimal values of \refprog{lowerSDPn} converges to the optimal value of \refprog{LPS}:
\begin{equation}
    \lim_{m\rightarrow+\infty}\omega^{m,\le}_n=\omega_n^\mathcal S.
\end{equation}
\end{theorem}
 
\noindent The proof is similar to that of Theorem~\ref{th:upperCV} using the dual programs: we attempt to construct a feasible optimal solution of \refprog{DLP} by extracting a limit from a sequence of optimal solutions of \refprog{lowerDSDPn}, for $m\ge n$ and then conclude using the strong duality between \refprog{lowerSDPn} and \refprog{lowerDSDPn}, which we proved in Theorem~\ref{th:sdlower}. However, it turns out that \refprog{DLP} may not have feasible optimal solutions in ${L^2}'(\R_+)$ (as anticipated with the analytical optimal solutions for $n=1$ and $n=2$ from section~\ref{sec:LP}). 

To deal with this issue, we have extended the formulation of \refprog{DLP} to a larger space where it has feasible optimal solutions, namely the space of tempered distributions $\mathcal S'(\R_+)$ (see \refprog{DLPS}).

Note that the semidefinite restrictions \refprog{lowerSDPn} of \refprog{LP} are also restrictions of \refprog{LPS} for all $m\ge n$, while \refprog{LPS} is itself a restriction of \refprog{LP}, and \refprog{DLPS} is a relaxation of \refprog{DLP}. We denote by ${\omega_n^{\mathcal S'}}$ the optimal value of \refprog{DLPS}. Recall that the optimal value of \refprog{LPS} is denoted $\omega_n^{\mathcal S}$ and that of \refprog{LP} and \refprog{DLP} are denoted $\omega_n^{L^2}$ (by strong duality). By weak duality of linear programming we thus have
\begin{equation}\label{eq:chainopti}
    \omega_n^{m,\le}\le\omega_n^{\mathcal S}\le{\omega_n^{\mathcal S'}}\le\omega_n^{L^2},
\end{equation}
for all $m\ge n$ (see also figure~\ref{fig:structure}).

Before proving Theorem~\ref{th:lowerCV}, we introduce two intermediate technical results. The first one is a reformulation of \refprog{lowerDSDPn} over Schwartz space:

\begin{lemma}\label{lem:newlowerDSDPn}
For all $m\ge n$, the program \refprog{lowerDSDPn} is equivalent to the following program:
\begin{fleqn}
\begin{equation*}
    \label{prog:newlowerDSDPn}
    \tag*{(D-SDP$_n^{m,\le}$)}
        \begin{aligned}
            & \quad \text{\upshape Find } y \in \R \; \text{\upshape and } \bm\mu \in\mathcal S'(\N) \\
            & \quad \text{\upshape minimising } y \\
            & \quad \text{\upshape subject to}  \\
            & \begin{dcases}
            \forall k \neq n \in \N,\; y \ge\mu_k  \\
            y \geq 1 + \mu_n \\
            \forall g \in \mathcal R_{m,+}(\R_+), \; \langle f_{\bm\mu},g \rangle \geq  0,
            \end{dcases}
        \end{aligned}
\end{equation*}
\end{fleqn}
where $f_{\bm\mu}=\sum_k\mu_k\mathcal L_k$.
\end{lemma}

\noindent This reformulation uses Stieltjes characterisation of non-negative polynomials over $\R_+$~\cite{reed1975ii} and is detailed in Appendix~\ref{sec:app_newlowerDSDPn}.

The second result provides a nontrivial analytical solution to the primal program \refprog{lowerSDPn}:
let us define $\bm F^n=(F_k^n)_{k\in\N}\in\R^{\N}$ by \\
$\bullet$ if $n$ is even:
\begin{equation}\label{eq:Fkforneven}
        F_k^n:=\begin{cases}\frac1{2^n}\binom k{\frac k2}\binom{n-k}{\frac{n-k}2}&\text{when }k\le n, k\text{ even},\\0&\text{otherwise},\end{cases}
\end{equation}
$\bullet$ if $n$ is odd:
\begin{equation}\label{eq:Fkfornodd}
        F_k^n:=\begin{cases} \frac1{2^n}\frac{\binom n{\floor{\frac n2}}\binom{\floor{\frac n2}}{\floor{\frac k2}}^2 }{\binom nk},&\text{when }k\le n,\\0&\text{otherwise}. \end{cases}
\end{equation}
In both cases,
\begin{equation}\label{eq:boundFnn}
    \begin{aligned}
        F_n^n&=\frac1{2^n}\binom n{\floor{\frac n2}}\\
        &\ge\frac1{n+1}
    \end{aligned}
\end{equation}
where we used $\binom nj\le\binom n{\floor{\frac n2}}$ for all $j\in\llbracket0,n\rrbracket$, summed over $j$ from 0 to $n$. 

\begin{lemma}\label{lem:feasible}
For all $m\ge n$, $\bm F^n$ is a feasible solution of \refprog{lowerSDPn}. Moreover, it is optimal when $m=n$.
\end{lemma}

\noindent The proof of feasibility consists in checking that the constraints of \refprog{lowerSDPn} are satisfied by $\bm F^n$. To do so, we make use of Zeilberger's algorithm~\cite{zeilberger1991method}, a powerful algorithm for proving binomial identities. The proof of optimality for $m=n$ is obtained by deriving an analytical feasible solution of (D-SDP$_n^{n,\le}$) with the same optimal value. Note that the optimality of this solution does not play a role in the proof of convergence. We refer to Appendix~\ref{sec:app_feasible} for a detailed proof. 

As a consequence, we obtain the following analytical lower bound for the optimal value of \refprog{LP}:
\begin{equation}
    \omega_n\ge \omega_n^{n,\le}=\frac1{2^n}\binom n{\floor{\frac n2}}\underset{n\rightarrow+\infty}\sim\sqrt{\frac2{\pi n}},
\end{equation}
which is superseded by numerical bounds when $n\ge3$ (see Table~\ref{tab:Fockbounds}).

We now combine Lemma~\ref{lem:newlowerDSDPn} and Lemma~\ref{lem:feasible} to prove Theorem~\ref{th:lowerCV}.

\begin{proof}
The feasible set of \refprog{newlowerDSDPn} is non-empty, by considering the null sequence, which achieves value $1$. Without loss of generality, we add the constraint $y\le1$ in \refprog{DLPS} and \refprog{newlowerDSDPn}.

Let $m\ge n$ and let $(y,\bm\mu)\in\R\times S'(\N)$ be a feasible solution of \refprog{newlowerDSDPn} expressed in the form given by Lemma~\ref{lem:newlowerDSDPn}. 
The constraint $\braket{f_{\bm\mu},x\mapsto e^{-\frac x2}}\ge0$ implies $\mu_0\ge0$ and thus $y\ge0$. Without loss of generality, we may set $\mu_k=0$ for $k>m$, since these coefficients are only constrained by $\mu_k\le y\le1$. We also have $\mu_k\le1$ for all $k\in\N$.

By Lemma~\ref{lem:feasible}, $\bm F^l$ is a feasible solution of (SDP$_l^{m,\le}$) for all $l\le m$, so in particular $f_{\bm F^l}=\sum_{k=0}^lF_k^l\mathcal L_k\in\mathcal R_{m,+}(\R_+)$. Hence, $\bm\mu$ satisfies the constraint $\braket{f_{\bm\mu},f_{\bm F^l}}\ge0$, which gives
\begin{equation}
    \sum_{k=0}^l\mu_kF_k^l\ge0,
\end{equation}
for all $l\le m$.
Thus we have, for all $l\in\llbracket1,m\rrbracket$,
\begin{equation}\label{eq:mucompact}
    \begin{aligned}
        \mu_l&\ge-\frac1{F_l^l}\sum_{k=0}^{l-1}\mu_kF_k^l\\
        &\ge-\frac1{F_l^l}\sum_{k=0}^{l-1}F_k^l\\
        &=1-\frac1{F_l^l}\\
        &\ge-l,
    \end{aligned}
\end{equation}
where we used $F_l^l>0$ in the first line, $\mu_k\le1$ and $F_k^l\ge0$ in the second line, $\sum_{k=0}^lF_k^l=1$ in the third line and~\refeq{boundFnn} in the last line. With $\mu_k\le1$, we obtain $|\mu_k|\le k$ for $k\in\N^*$, and thus $|\mu_k|\le k+1$\footnote{Critically, this bound does not depend on the level $m$ of the hierarchy. } for all $k\in\mathbb N$. Hence, the feasible set of \refprog{newlowerDSDPn} is compact and the program \refprog{newlowerDSDPn} has feasible optimal solutions for all $m\ge n$, by diagonal extraction.

Let $(y^m,\bm\mu^m)_{m\ge n}$ be a sequence of optimal solutions of \refprog{newlowerDSDPn}, for $m\ge n$. By Theorem~\ref{th:sdlower}, we have strong duality between the programs \refprog{lowerSDPn} and \refprog{newlowerDSDPn}, so the optimal value of \refprog{newlowerDSDPn} is given by $\omega_n^{m,\le}$, for all $m\ge n$. By optimality $y^m=\omega_n^{m,\le}$, for all $m\ge n$, and the sequence $(y^m)_{m\ge n}$ converges. 

Performing a diagonal extraction $\phi$ on the sequence $(\bm\mu^m)_{m\ge n}$, we obtain a sequence of sequences $(\bm\mu^{\phi(m)})_{m\ge n}$ such that each sequence $(\mu^{\phi(m)}_k)_{m\ge n}$ converges when $m\rightarrow+\infty$, for all $k\in\N$. Let $\mu_k$ denote its limit, for each $k\in\N$. We write $\bm \mu=(\mu_k)_{k\in\N}\in\R^{\N}$ the sequence of limits. We also write 
\begin{equation}\label{eq:ylim}
    y:=\lim_{m\rightarrow+\infty}y^m=\lim_{m\rightarrow+\infty}\omega_n^{m,\le}.
\end{equation}
For all $m\ge n$, we have $\omega_n^{\phi(m),\le}\ge\mu_k^{\phi(m)}$ for all $k\in\mathbb N$ and $\omega_n^{\phi(m),\le}\ge1+\mu_n^{\phi(m)}$, so taking $m\rightarrow+\infty$ we obtain $y\ge\mu_k$ for all $k\in\mathbb N$ and $y\ge1+\mu_n$. By~\refeq{chainopti}, we have $\omega^{m,\le}_n\le\omega_n^{\mathcal S}$ for all $m\ge n$, so $y\le\omega_n^{\mathcal S}$.

Moreover, $|\mu_k^{\phi(m)}|\le k+1$ for all $k\in\mathbb N$, so taking $m\rightarrow+\infty$ we obtain $|\mu_k|\le k+1$ for all $k\in\mathbb N$, which implies that $\bm\mu\in\mathcal S'(\N)$~\cite{guillemot1971developpements}. Let $f_{\bm\mu}=\sum_k\mu_k\mathcal L_k\in \mathcal S'(\R_+)$. We have 
\begin{equation}
    \mu_k=\braket{f_{\bm\mu},\mathcal L_k}.
\end{equation}
By construction we also have:
\begin{equation}
    \forall m\ge n,\forall g\in\mathcal R_{m,+}(\R_+),\;\braket{f_{\bm \mu},g}\ge0.
\end{equation}
By Theorem~\ref{th:RHLaguerre}, this implies that the distribution $\mu:=f_{\bm\mu}(x)\in\mathcal S'(\R_+)$ is non-negative, i.e.,
\begin{equation}
    \forall f \in \mathcal S_+(\R_+), \;\braket{f_{\bm \mu},f}\ge0.    
\end{equation}
With the constraints $y\ge\mu_k$ for all $k\in\mathbb N$ and $y\ge1+\mu_n$, we have that $(y,\mu)$ is a feasible solution of \refprog{DLPS}, and in particular $y\ge{\omega_n'^{\mathcal S}}$, since \refprog{DLPS} is a minimisation problem. Since $y\le\omega_n^{\mathcal S}$ we obtain with~\refeq{chainopti} and~\refeq{ylim}:
\begin{equation}
    \lim_{m\rightarrow+\infty}\omega^{m,\le}_n=\omega_n^{\mathcal S}={\omega_n'^{\mathcal S}}.
\end{equation}
\end{proof}
 
\noindent As a direct corollary of the proofs of Theorem~\ref{th:upperCV} and Theorem~\ref{th:lowerCV} (in the same spirit than the remark after the proof of Theorem~\ref{th:upperCV}), we obtain the following strong duality result:
 
\begin{theorem}\label{th:sdLP}
Strong duality holds between the programs \refprog{LPS} and \refprog{DLPS} and between the programs \refprog{LP} and \refprog{DLP}.
\end{theorem}
 
\noindent For completeness, we give a different and more direct proof of the strong duality between \refprog{LP} and \refprog{DLP} in Appendix~\ref{sec:app_sdLP}.

We have shown the convergence of the semidefinite hierarchies of upper and lower bounds: $(\omega^{m,\ge}_n)_{m\ge n}$ towards the optimal value of \refprog{LP} and
$(\omega^{m,\le}_n)_{m\ge n}$ towards the optimal value of \refprog{LPS}.
By linearity, these results generalise straightforwardly to the case of witnesses corresponding to linear combinations of fidelities with displaced Fock states:
\begin{align}
    & \lim_{m\rightarrow+\infty}\!\omega^{m,\le}_{\bm a}=\omega_{\bm a}^\mathcal S, \\
    & \lim_{m\rightarrow+\infty}\!\omega^{m,\ge}_{\bm a}=\omega_{\bm a}^{L^2},
\end{align}
for all $n\in\N^*$ and all $\bm a\in[0,1]^n$.

We obtained two hierarchies providing numerical lower bounds and upper bounds on the threshold value. A natural question that arises is the following: is there a gap between the optimal values of \refprog{LP} and \refprog{LPS}? We leave this as an open question.





\section{Witnessing multimode Wigner negativity}
\label{sec:multi}

In this section we discuss the generalisation of our Wigner negativity witnesses to the more challenging multimode setting. Hereafter, $M$ denotes the number of modes.

\subsection{Multimode Wigner negativity witnesses}

Using multi-index notations (see Appendix~\ref{sec:app_multidef}), the single-mode Wigner negativity witnesses defined in Eq.~\eqref{eq:witnessOmega} are naturally generalised to
\begin{equation}\label{eq:witnessOmegamulti}
    \hat\Omega_{\bm a,\bm\alpha}:=\smashoperator{\sum_{\bm{1}\le\bm k\le\bm n}}
      a_{\bm k}\hat D(\bm\alpha)\ket{\bm k}\!\bra{\bm k}\hat D^\dag(\bm\alpha),
\end{equation}
for $\bm n=(n_1,\dots,n_M)\in\mathbb N^M\setminus\{\bm0\}$, $\bm a=(a_{\bm k})_{\bm{1}\le\bm k\le\bm n}\in[0,1]^{n_1\cdots n_M}$, 
with $\max_{\bm k}a_{\bm k}=1$, and $\bm\alpha\in\mathbb C^M$. Similar to the single-mode case, these POVM elements are weighted sums of multimode displaced Fock states projectors, and their expectation value for a quantum state $\bm\rho\in\mathcal D(\mathcal H^{\otimes M})$ is given by
\begin{equation}
    \Tr(\hat\Omega_{\bm a,\bm\alpha}\bm\rho)
    =\smashoperator{\sum_{\bm{1}\le\bm k\le\bm n}}
       a_{\bm k}F\left(\hat D^\dag(\bm\alpha)\bm\rho\hat D(\bm\alpha),\ket{\bm k}\right)
\end{equation}
where $F$ is the fidelity. Unlike in the single-mode case however, estimating this quantity with homodyne or heterodyne detection by direct fidelity estimation is no longer efficient when the number of modes becomes large. Instead, one may use robust lower bounds on the multimode fidelity from~\cite{chabaud2020efficient} which can be obtained efficiently with homodyne or heterodyne detection. A lower bound on the estimated experimental multimode fidelity will allow to detect Wigner negativity if it is larger than an upper bound on the threshold value associated to a given witness.

These lower bounds are obtained as follows: given a target multimode Fock state $\ket{\bm n}=\ket{n_1}\otimes\cdots\otimes\ket{n_M}$ and multiple copies of an $M$-mode experimental state $\bm \rho$, measure all single-mode subsystems of $\bm\rho$ and perform fidelity estimation with each corresponding single-mode target Fock state. That is, the samples obtained from the detection of the $i^{th}$ mode of $\bm\rho$ are used for single-mode fidelity estimation with the Fock state $\ket{n_i}$. Let $F_1,\dots,F_M$ be the single-mode fidelity estimates obtained and let $\tilde F(\bm\rho,\ket{\bm n}):=1-\sum_{i=1}^M(1-F_i)$. Then~\cite{chabaud2020efficient},
\begin{equation}\label{eq:robustfide}
    \hspace{-0.15cm}1\!-\!M(1\!-\!F(\bm\rho,\ket{\bm n}))\le\tilde F(\bm\rho,\ket{\bm n})\le F(\bm\rho,\ket{\bm n}).
\end{equation}
In particular, $\tilde F$ provides a good estimate of the multimode fidelity $F$ whenever $F$ is not too small. 

The same procedure is followed in the case of target displaced Fock states, with classical translations of the samples in order to account for the displacement parameters.

To each witness $\hat\Omega_{\bm a,\bm\alpha}$ is associated its threshold value:
\begin{equation}\label{eq:threshold_multi}
    \omega_{\bm a}\footnote{Again we write the threshold value as $\omega_{\bm a}$ generically and we will use the more precise notation $\omega_{\bm a}^{L^2}$ (resp. $\omega_{\bm a}^{\mathcal S}$) to refer to the optimisation over square-integrable functions (resp. Schwartz functions) on $\R^M$.}
    :=\sup_{\substack{\bm\rho\in \mathcal D(\mathcal H^{\otimes M})\\W_{\bm\rho} \geq 0}}\Tr\!\left(\hat\Omega_{\bm a,\bm\alpha}\,\bm\rho\right).
\end{equation}
With Eq.~\eqref{eq:robustfide}, if the value of the bound $\tilde F$ obtained experimentally is greater than $\omega_{\bm a}$, then the state $\bm\rho$ has a negative Wigner function. 

With the same arguments as in the single-mode case, the multimode Wigner negativity witnesses in Eq.~\eqref{eq:witnessOmegamulti} form a complete family and retain the interpretation from Lemma~\ref{lem:operational}: the violation of the threshold value provides a lower bound on the distance to the set of multimode states with non-negative Wigner function. However, the limited robustness of the bound $\tilde F$ may affect the performance of the witnesses in practical scenarios, in particular for witnesses that are sums of different projectors. Still, we show in section~\ref{sec:example_multi} the applicability of the method with a genuinely multimode example.

We first generalise the single-mode semidefinite programming approach for approximating the threshold values to the multimode case.

\subsection{Approximating the multimode threshold values}

By linearity, we restrict our analysis to the case of Wigner negativity witnesses that are projectors onto a single multimode Fock state $\ket{\bm n}$, for $\bm n\in\mathbb N^M\setminus\{\bm0\}$. We thus consider the computation of
\begin{equation}
    \omega_{\bm n}=\sup_{\substack{\bm\rho\in\mathcal D(\mathcal H^{\otimes M})\\W_{\bm\rho}\ge0}}\braket{\bm n|\bm\rho|\bm n}.
\end{equation}
A similar reasoning to the single-mode case shows that the computation of the corresponding threshold value in Eq.~\eqref{eq:threshold_multi} may be rephrased as the following infinite-dimensional linear program:
\begin{fleqn}
\begin{equation*}
    \label{prog:LP_multi}
        \tag*{(LP$_{\bm n}^{L^2}$)}
        \begin{aligned}
            & \quad \text{Find }  (F_{\bm k})_{\bm k \in \N^M} \in \ell^2 \\
            & \quad \text{maximising } F_{\bm n} \\
            & \quad \text{subject to} \\
            &\begin{dcases}
             \sum_{\bm k} F_{\bm k} =  1  \\
            \forall \bm k \in \N^M, \;F_{\bm k} \geq  0  \\
            \forall \bm x \in \R_+^M, \; \sum_{\bm k} F_{\bm k}\mathcal L_{\bm k}(\bm x) \geq  0,
            \end{dcases}
        \end{aligned}
\end{equation*}
\end{fleqn}
where the optimisation is over square-summable real sequences indexed by elements of $\mathbb N^M$. Its dual linear program reads
\begin{fleqn}
\begin{equation*}
    \label{prog:DLP_multi}
    \tag*{(D-LP$_{\bm n}^{L^2}$)}
        \begin{aligned}
            & \quad \text{Find } y \in \R \text{ and } \mu \in {L^2}'(\R_+^M)\\
            & \quad \text{minimising } y\\
            & \quad \text{subject to}  \\
            & \begin{dcases}
            \forall \bm k \neq \bm n \in \N^M,\;  y \geq \int_{\R_+^M}{\mathcal L_{\bm k}}{d\mu}  \\
            y \geq  1 + \int_{\R_+^M}{\mathcal L_{\bm n}}{d\mu} \\
            \forall f \in L^2_+(\R_+^M), \; \langle \mu,f \rangle \geq  0.
            \end{dcases}
        \end{aligned}
\end{equation*}
\end{fleqn}
We denote their optimal value $\omega_{\bm n}^{L^2}$---as in the single-mode case, we have strong duality between these programs.
We also introduce the programs over Schwartz functions:
\begin{fleqn}
\begin{equation*}
    \label{prog:LP_multiS}
        \tag*{(LP$_{\bm n}^{\mathcal S}$)}
        \begin{aligned}
            & \quad \text{Find }  (F_{\bm k})_{\bm k \in \N^M} \in \mathcal S(\N^M) \\
            & \quad \text{maximising } F_{\bm n} \\
            & \quad \text{subject to} \\
            &\begin{dcases}
             \sum_{\bm k} F_{\bm k} =  1  \\
            \forall \bm k \in \N^M, \;F_{\bm k} \geq  0  \\
            \forall \bm x \in \R_+^M, \; \sum_{\bm k} F_{\bm k}\mathcal L_{\bm k}(\bm x) \geq  0,
            \end{dcases}
        \end{aligned}
\end{equation*}
\end{fleqn}
and
\begin{fleqn}
\begin{equation*}
    \label{prog:DLP_multiS}
    \tag*{(D-LP$_{\bm n}^{\mathcal S}$)}
        \begin{aligned}
            & \quad \text{Find } y \in \R \text{ and } \mu \in {\mathcal S}'(\R_+^M)\\
            & \quad \text{minimising } y\\
            & \quad \text{subject to}  \\
            & \begin{dcases}
            \forall \bm k \neq \bm n \in \N^M,\;  y \geq \int_{\R_+^M}{\mathcal L_{\bm k}}{d\mu}  \\
            y \geq  1 + \int_{\R_+^M}{\mathcal L_{\bm n}}{d\mu} \\
            \forall f \in L^2_+(\R_+^M), \; \langle \mu,f \rangle \geq  0.
            \end{dcases}
        \end{aligned}
\end{equation*}
\end{fleqn}
We denote their optimal value $\omega_{\bm n}^{S}$.

In the single-mode case, we obtained hierarchies of SDP relaxations and restrictions for~\refprog{LP} by replacing constraints involving non-negative functions by constraints involving non-negative polynomials $P$ of fixed degree. 
We then exploited the existence of a sum-of-squares decomposition for non-negative monovariate polynomials.
In the multimode setting, the polynomials involved are multivariate, so that the set of non-negative polynomials over $\R$ of a given degree may be strictly larger than the set of sum-of-square polynomials~\cite{hilbert1888darstellung}. Instead, we replace directly constraints involving non-negative functions over $\R_+$ by constraints involving non-negative polynomials $P$ of fixed degree such that $\bm x\mapsto P(\bm x^2)$\footnote{We write $\bm x^2$ in short for $\bm x^{2 \bf 1} = (x_1^2,\dots,x_M^2)$.} has a sum-of-squares decomposition, implying that the multimode semidefinite relaxations and restrictions are possibly looser than their single-mode counterparts. 

Moreover, the dimension of the semidefinite programs increases exponentially with the level of the hierarchy $m$, as the number of $M$-variate monomials of degree less or equal to $m$ is given by $\binom{M+m}m$. This implies that the semidefinite programs remain tractable only for a constant number of levels.

In spite of these observations, and following similar steps to the single-mode case (see Appendix~\ref{sec:app_multiSDP}), the SDP relaxations providing upper bounds for the threshold value are given by
\begin{fleqn}
\begin{equation*}
   \label{prog:upperSDPn_multi}
    \tag*{$(\text{SDP}^{m,\geq}_{\bm n})$}
        \begin{aligned}
            & \quad \text{Find } A=(A_{\bm i\bm j})_{|\bm i|,|\bm j|\le m}\in\text{Sym}_{\binom{M+m}m}\hspace{-3cm}\\
            &\quad\text{and } \bm F=(F_{\bm k})_{|\bm k|\le m}\in\R^{{\binom{M+m}m}}\hspace{-4cm} \\
            & \quad \text{maximising } F_{\bm n} \\
            & \quad \text{subject to}\\
            & \begin{dcases}
                        \textstyle \sum_{|\bm k|\le m} F_{\bm k} =  1  \\
                        \textstyle \forall |\bm k|\le m, \; F_{\bm k} \geq  0 \\
                         \textstyle \forall|\bm l|\le m,\forall\bm i\!+\!\bm j\!=\!2\bm l,\; A_{\bm i\bm j}=\sum\limits_{\bm k\le\bm l}F_{\bm k}\binom{\bm l}{\bm k}\bm l!\hspace{-3cm}\\
                        \textstyle \forall|\bm r|\le2m,\bm r\!\neq\!2\bm l,\forall|\bm l|\le m,\forall\bm i\!+\!\bm j\!=\!\bm r,\;A_{\bm i\bm j}=0\hspace{-3cm}\\
                        \textstyle A\succeq 0,
              \end{dcases}
        \end{aligned}
\end{equation*}
\end{fleqn}
for all $m\ge|\bm n|$.
Similarly, the semidefinite restrictions providing lower bounds for the threshold value are given by
\begin{fleqn}
\begin{equation*}
   \label{prog:lowerSDPn_multi}
    \tag*{$(\text{SDP}^{m,\leq}_{\bm n})$}
        \begin{aligned}
            & \quad \text{Find } Q=(Q_{\bm i\bm j})_{|\bm i|,|\bm j|\le m}\in\text{Sym}_{\binom{M+m}m}\hspace{-4cm}\\
            &\quad\text{and } \bm F=(F_{\bm k})_{|\bm k|\le m}\in\R^{{\binom{M+m}m}}\hspace{-4cm} \\
            & \quad \text{maximising } F_{\bm n} \\
            & \quad \text{subject to}\\
            & \begin{dcases}
                        \textstyle \sum_{|\bm k|\le m} F_{\bm k} =  1  \\
                        \textstyle \forall |\bm k|\le m, \; F_{\bm k} \geq  0 \\
                        \textstyle \forall |\bm l|\le m,\sum_{\bm i+\bm j=2\bm l}Q_{\bm i\bm j}=\sum_{\bm k\ge\bm l} \frac{(-1)^{|\bm k|+|\bm l|}}{\bm l!}\binom{\bm k}{\bm l}F_{\bm k}\hspace{-4cm}\\
                        \textstyle \forall |\bm r|\le 2m,\bm r\neq2\bm l,\forall |\bm l|\le m,\sum_{\bm i+\bm j=\bm r}Q_{\bm i\bm j}=0\hspace{-4cm}\\
                        \textstyle Q\succeq 0,
              \end{dcases}
        \end{aligned}
\end{equation*}
\end{fleqn}
for all $m\ge|\bm n|$. In these semidefinite programs, the optimisations are over matrices and vectors indexed by elements of $\mathbb N^m$ with sum of coefficients lower that $m$.

While our proof of convergence of the single-mode hierarchy of upper bounds towards $\omega_{\bm n}^{L^2}$ transfers easily to the multimode setting, the proof of convergence of the hierarchy of lower bounds towards $\omega_{\bm n}^\mathcal S$ requires the analytical expression of feasible solutions for each level of the hierarchy. We show how to construct such solutions in the multimode case using products of single-mode feasible solutions---this requires introducing an equivalent hierarchy of restrictions, where constraints are expressed on polynomials of $M$ variables with the degree in each individual variable being less or equal to $m$, rather than on polynomials of degree $m$ (that is, constraints of the form $\bm k \leq \bm m$ where $\bm m=(m,\dots,m)\in\N^M$ rather than $\vert \bm k \vert \leq m$). Along the way, we also prove strong duality of the programs involved. We refer to Appendix~\ref{sec:app_multiCV} for the proofs.

Summarising our results, we find that the multimode semidefinite programs~\refprog{upperSDPn_multi} and~\refprog{lowerSDPn_multi} 
provide converging sequences of upper and lower bounds to the threshold values $\omega_{\bm n}^{L^2}$ and $\omega_{\bm n}^{\mathcal S}$, respectively.
We study a concrete application in the next section.

\subsection{Multimode example}
\label{sec:example_multi}

To illustrate the usefulness of our Wigner negativity witnesses in the multimode setting, we consider a lossy Fock state over two modes:
\begin{equation}
    \begin{aligned}
        \bm\rho_{1,1,\eta}&:=(1-\eta)^2\ket1\!\bra1\otimes\ket1\!\bra1\\
        &\;\;+\eta(1-\eta)\ket1\!\bra1\otimes\ket0\!\bra0\\
        &\;\;+\eta(1-\eta)\ket0\!\bra0\otimes\ket1\!\bra1\\
        &\;\;+\eta^2\ket0\!\bra0\otimes\ket0\!\bra0,
    \end{aligned}
\end{equation}
with loss parameter $0\le\eta\le1$. Setting $\eta=0$ gives $\bm\rho_{1,1,\eta}=\ket1\!\bra1\otimes\ket1\!\bra1$ while setting $\eta=1$ gives $\bm\rho_{1,1,\eta}=\ket0\!\bra0\otimes\ket0\!\bra0$. This state has a non-negative Wigner function for $\eta\geq\frac12$.

We also consider the multimode Wigner negativity witness $\ket1\!\bra1\otimes\ket1\!\bra1$, which is a projector onto the Fock state $\ket1\otimes\ket1$. Solving numerically the corresponding hierarchy~\refprog{lowerSDPn_multi} up to $m=3$, we obtain the lower bound $0.266$ and solving the hierarchy~\refprog{upperSDPn_multi} up to $m=10$, we obtain the upper bound $0.320$. 

A direct consequence of the numerical lower bound is that tensor product states are not the closest among Wigner positive states to tensor product states with a negative Wigner function. Indeed, the maximum achievable fidelity with the state $\ket1\otimes\ket1$ using Wigner positive tensor product states is equal to the square of the maximum achievable fidelity with the state $\ket1$ using single-mode Wigner positive states, that is $0.5^2=0.25<0.266$.

\begin{figure}[t]
	\begin{center}
		\includegraphics[width=\columnwidth]{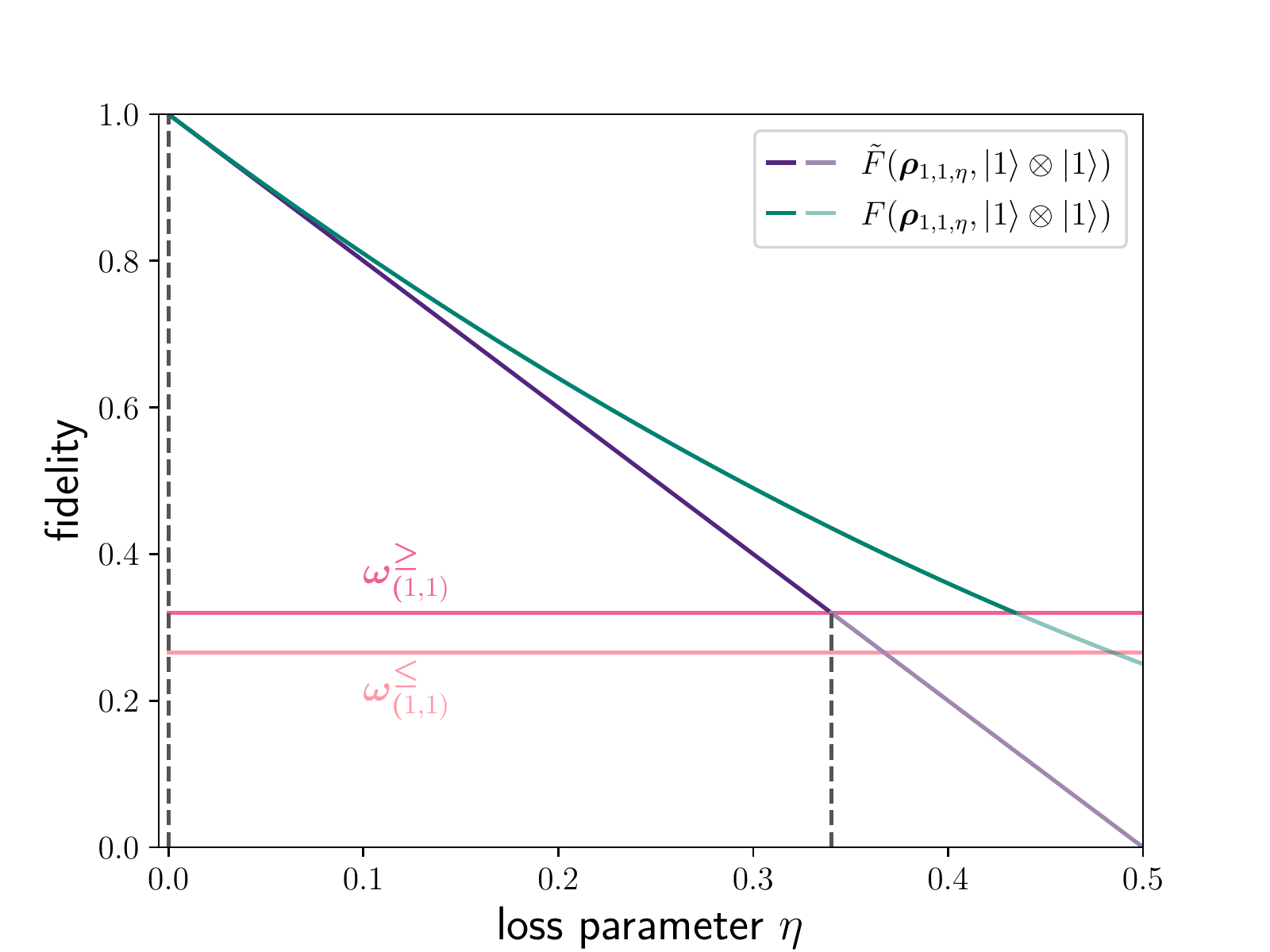}
		\caption{Witnessing Wigner negativity of the lossy Fock state $\bm\rho_{1,1,\eta}$ over two modes using the witness $\ket1\!\bra1\otimes\ket1\!\bra1$. The threshold value for that witness is upper bounded by $0.320$ and lower bounded by $0.266$.  The dashed red line delimits the interval of loss parameter values where the witness can be used to detect Wigner negativity of $\bm\rho_{1,1,\eta}$ efficiently, i.e., when the robust bound $\tilde F(\bm\rho_{1,1,\eta},\ket1\otimes\ket1)$ (blue curve) on the fidelity from Eq.~\eqref{eq:tildeF11} is above the witness upper bound (red line). When it is below the witness lower bound (black line), we are guaranteed that the witness cannot be used to detect Wigner negativity of the state. The fidelity $F(\bm\rho_{1,1,\eta},\ket1\otimes\ket1)$ is also depicted above (yellow curve). Note that $\bm\rho_{1,1,\eta}$ has a non-negative Wigner function for $\eta\ge0.5$.}
		\label{fig:multi}
	\end{center}
\end{figure}

We now use the upper bound to witness the Wigner negativity of the state $\bm\rho_{1,1\eta}$ (see Fig.~\ref{fig:multi}).
The fidelity between $\bm\rho_{1,1,\eta}$ and $\ket1\otimes\ket1$ is given by $F(\bm\rho_{1,1,\eta},\ket1\otimes\ket1)=(1-\eta)^2$, for all $0\le\eta\le1$. This fidelity is above the upper bound $0.320$ on the threshold value of the witness $\ket1\!\bra1\otimes\ket1\!\bra1$ when $\eta\le0.434$. 

However, in practice one would not obtain a precise estimate of the fidelity efficiently, but rather a robust lower bound on the fidelity computed from single-mode fidelities, which satisfies Eq.~\eqref{eq:robustfide}. In the worst case, the estimate obtained is closer to $1-2(1- F(\bm\rho_{1,1,\eta},\ket1\otimes\ket1))$ than to $F(\bm\rho_{1,1,\eta},\ket1\otimes\ket1)$. When the value of this robust lower bound is greater than the threshold value of the witness, this implies that the state has a negative Wigner function. 

In the present case, the two single-mode reduced states of $\bm\rho_{1,1,\eta}$ are the same, given by
\begin{equation}
    \Tr_2(\bm\rho_{1,1,\eta})=(1-\eta)\ket1\!\bra1+\eta\ket0\!\bra0,
\end{equation}
so the single-mode fidelities with $\ket1$ are equal for each mode and given by $1-\eta$. Hence, the robust lower bound on $F(\bm\rho_{1,1,\eta},\ket1\otimes\ket1)$ is given by
\begin{equation}\label{eq:tildeF11}
    \tilde F(\bm\rho_{1,1,\eta},\ket1\otimes\ket1)=1-2\eta.
\end{equation}
It is above the upper bound $0.320$ on the threshold value of the witness $\ket1\!\bra1\otimes\ket1\!\bra1$ when $\eta\le0.340$.

This example highlights the use of efficient and robust lower bounds on multimode fidelities rather than fidelity estimates~\cite{chabaud2020efficient}, in conjunction with our family of multimode witnesses to detect Wigner negativity of realistic experimental states.

\section{Conclusion and open problems}
\label{sec:conclusion}

Characterising quantum properties of physical systems is an important step in the development of quantum technologies, and negativity of the Wigner function, a necessary resource for any quantum computational speedup, is no exception.
In this work, we have derived a complete family of Wigner negativity witnesses which provide an operational quantification of Wigner negativity, both in the single-mode and multimode settings. 
In the context of quantum optical information processing, the main application of our method is in experimental scenarios, where it leads to robust and efficient certification of negativity of the Wigner function. Witnesses of Wigner negativity also provide witnesses of non-classicality although they are no longer complete in this case.

What is more, our witnesses also delineate the set of quantum states with positive Wigner function, and it would be interesting to understand whether additional insights on this set can be obtained using these witnesses.

The Wigner function has been extended to the discrete-variable setting~\cite{leonhardt1995quantum,gross2006hudson}, where it has been been linked to contextuality~\cite{spekkens2008negativity,delfosse2017equivalence,raussendorf2017contextuality}, a necessary resource for discrete-variable quantum computing~\cite{howard2014contextuality,bermejo2017contextuality}. A framework for treating contextuality and computing the amount of contextuality in continuous-variable settings has recently been developed~\cite{barbosa2019continuous}. As we obtained reliable Wigner negativity witnesses, it would be interesting to investigate the link between Wigner negativity and continuous-variable contextuality. 

Hierarchies of semidefinite programs (in particular with non-commutative variables~\cite{navascues2008convergent}) have found many recent applications in quantum information theory.
From an infinite-dimensional linear program, we were able to use numerically both a hierarchy of upper bounds and a hierarchy of lower bounds---thus obtaining a certificate for the optimality of these bounds by looking at their difference---whereas this only works in specific cases for the Lasserre hierarchy of upper bounds~\cite{lasserre2011new}. Can we find other interesting cases where we can exploit both hierarchies? Moreover, we obtained an analytical sequence of lower bounds for the threshold value of the program \refprog{LP}. Can we also get an analytical sequence of upper bounds? In particular, we anticipate that Fock states $\ket n$ get further away from the set of states having a positive Wigner function as $n$ increases and that $\omega_n = \mathcal{O}(\frac 1{\sqrt{n}})$ as $n\rightarrow+\infty$.
The question of the gap between the optimisation over square-integrable functions and Schwartz functions is also left open.

Finally, using our multimode Wigner negativity witnesses for studying the interplay between Wigner negativity and entanglement~\cite{PhysRevLett.119.183601} is a very interesting prospect which we leave for future work.

\subsection*{Acknowledgments}

U.\ Chabaud acknowledges stimulating discussions with S.\ Gribling, T.\ Freiman and T.\ Vidick. P.-E.\ Emeriau acknowledges interesting discussions with A.\ Oustry, E.\ Galv\~ao and R.\ Soares Barbosa. We thank J.\ Eisert for his valuable comments on a previous version of this work and P.\ Paule for providing access to the Mathematica package for implementing Zeilberger's algorithm. U.\ Chabaud acknowledges funding provided by the Institute for Quantum Information and Matter, an NSF Physics Frontiers Center (NSF Grant PHY-1733907). F.\ Grosshans   acknowledges  funding  from  the  ANR through the ANR-17-CE24-0035 VanQuTe project.


\typeout{}
\hypersetup{breaklinks=true}
\bibliographystyle{linksen}
\bibliography{bibliography}

\providecommand{\href}[2]{#2}\begingroup\raggedright\begin{thebibliography}{10}

\bibitem{lloyd1999quantum}
S.~Lloyd and S.~L. Braunstein, ``Quantum computation over continuous
  variables,'' \href{http://dx.doi.org/10.1007/978-94-015-1258-9_2}{in {\em
  Quantum Information with Continuous Variables}}, pp.~9--17.
\newblock Springer, 1999.

\bibitem{yokoyama2013ultra}
S.~Yokoyama, R.~Ukai, S.~C. Armstrong, C.~Sornphiphatphong, T.~Kaji, S.~Suzuki,
  J.-i. Yoshikawa, H.~Yonezawa, N.~C. Menicucci, and A.~Furusawa,
  ``Ultra-large-scale continuous-variable cluster states multiplexed in the
  time domain,'' \href{http://dx.doi.org/10.1038/nphoton.2013.287}{{\em Nature
  Photonics} {\bfseries 7}, 982 (2013)}.

\bibitem{Leonhardt-essential}
U.~Leonhardt, ``Essential Quantum Optics,''.
\newblock \href{http://dx.doi.org/10.1017/CBO9780511806117}{Cambridge
  University Press}, Cambridge, UK, 1st~ed., 2010.

\bibitem{moyal1949quantum}
J.~E. Moyal, ``Quantum mechanics as a statistical theory,''
  \href{http://dx.doi.org/10.1017/S0305004100000487}{in {\em Mathematical
  Proceedings of the Cambridge Philosophical Society}}, vol.~45, pp.~99--124,
  Cambridge University Press.
\newblock 1949.

\bibitem{wigner1997quantum}
E.~P. Wigner, ``On the quantum correction for thermodynamic equilibrium,''
  \href{http://dx.doi.org/10.1007/978-3-642-59033-7_8}{in {\em Part I: Physical
  Chemistry. Part II: Solid State Physics}}, pp.~110--120.
\newblock Springer, 1997.

\bibitem{lee1991measure}
C.~T. Lee, ``Measure of the nonclassicality of nonclassical states,''
  \href{http://dx.doi.org/10.1103/PhysRevA.44.R2775}{{\em Physical Review A}
  {\bfseries 44}, R2775 (1991)}.

\bibitem{giedke2002characterization}
G.~Giedke and J.~I. Cirac, ``Characterization of {G}aussian operations and
  distillation of {G}aussian states,''
  \href{http://dx.doi.org/10.1103/PhysRevA.66.032316}{{\em Physical Review A}
  {\bfseries 66}, 032316 (2002)}.

\bibitem{eisert2002distilling}
J.~Eisert, S.~Scheel, and M.~B. Plenio, ``Distilling {G}aussian states with
  {G}aussian operations is impossible,''
  \href{http://dx.doi.org/10.1103/PhysRevLett.89.137903}{{\em Physical Review
  Letters} {\bfseries 89}, 137903 (2002)}.

\bibitem{fiuravsek2002gaussian}
J.~Fiur{\'a}{\v{s}}ek, ``Gaussian transformations and distillation of entangled
  {G}aussian states,''
  \href{http://dx.doi.org/10.1103/PhysRevLett.89.137904}{{\em Physical Review
  Letters} {\bfseries 89}, 137904 (2002)}.

\bibitem{niset2009no}
J.~Niset, J.~Fiur{\'a}{\v{s}}ek, and N.~J. Cerf, ``No-go theorem for {G}aussian
  quantum error correction,''
  \href{http://dx.doi.org/10.1103/PhysRevLett.102.120501}{{\em Physical Review
  Letters} {\bfseries 102}, 120501 (2009)}.

\bibitem{ghose2007non}
S.~Ghose and B.~C. Sanders, ``Non-Gaussian ancilla states for continuous
  variable quantum computation via Gaussian maps,''
  \href{http://dx.doi.org/10.1080/09500340601101575}{{\em Journal of Modern
  Optics} {\bfseries 54}, 855--869 (2007)}.

\bibitem{Bartlett2002}
S.~D. Bartlett, B.~C. Sanders, S.~L. Braunstein, and K.~Nemoto, ``Efficient
  Classical Simulation of Continuous Variable Quantum Information Processes,''
  \href{http://dx.doi.org/10.1103/PhysRevLett.88.097904}{{\em Physical Review
  Letters} {\bfseries 88}, 097904 (2002)}.

\bibitem{chabaud2020classical}
U.~Chabaud, G.~Ferrini, F.~Grosshans, and D.~Markham, ``Classical simulation of
  Gaussian quantum circuits with non-Gaussian input states,''
  \href{http://arxiv.org/abs/2010.14363}{{\ttfamily arXiv:2010.14363}}.

\bibitem{hudson1974wigner}
R.~L. Hudson, ``When is the Wigner quasi-probability density non-negative?,''
  \href{http://dx.doi.org/10.1016/0034-4877(74)90007-X}{{\em Reports on
  Mathematical Physics} {\bfseries 6}, 249--252 (1974)}.

\bibitem{soto1983wigner}
F.~Soto and P.~Claverie, ``When is the Wigner function of multidimensional
  systems nonnegative?,'' \href{http://dx.doi.org/10.1063/1.525607}{{\em
  Journal of Mathematical Physics} {\bfseries 24}, 97--100 (1983)}.

\bibitem{mandilara2009extending}
A.~Mandilara, E.~Karpov, and N.~Cerf, ``Extending Hudson’s theorem to mixed
  quantum states,'' \href{http://dx.doi.org/10.1103/PhysRevA.79.062302}{{\em
  Physical Review A} {\bfseries 79}, 062302 (2009)}.

\bibitem{filip2011detecting}
R.~Filip and L.~Mi{\v{s}}ta~Jr, ``Detecting quantum states with a positive
  Wigner function beyond mixtures of Gaussian states,''
  \href{http://dx.doi.org/10.1103/PhysRevLett.106.200401}{{\em Physical Review
  Letters} {\bfseries 106}, 200401 (2011)}.

\bibitem{tan2020negativity}
K.~C. Tan, S.~Choi, and H.~Jeong, ``Negativity of quasiprobability
  distributions as a measure of nonclassicality,''
  \href{http://dx.doi.org/10.1103/PhysRevLett.124.110404}{{\em Physical review
  letters} {\bfseries 124}, 110404 (2020)}.

\bibitem{titulaer1965correlation}
U.~Titulaer and R.~Glauber, ``Correlation functions for coherent fields,''
  \href{http://dx.doi.org/10.1103/PhysRev.140.B676}{{\em Physical Review}
  {\bfseries 140}, B676 (1965)}.

\bibitem{kenfack2004negativity}
A.~Kenfack and K.~{\.Z}yczkowski, ``Negativity of the Wigner function as an
  indicator of non-classicality,''
  \href{http://dx.doi.org/10.1088/1464-4266/6/10/003}{{\em Journal of Optics B:
  Quantum and Semiclassical Optics} {\bfseries 6}, 396 (2004)}.

\bibitem{Mari2012}
A.~Mari and J.~Eisert, ``Positive Wigner Functions Render Classical Simulation
  of Quantum Computation Efficient,''
  \href{http://dx.doi.org/10.1103/PhysRevLett.109.230503}{{\em Physical Review
  Lett.} {\bfseries 109}, 230503 (2012)}.

\bibitem{garcia2020efficient}
L.~Garc{\'\i}a-{\'A}lvarez, C.~Calcluth, A.~Ferraro, and G.~Ferrini,
  ``Efficient simulatability of continuous-variable circuits with large Wigner
  negativity,'' \href{http://arxiv.org/abs/2005.12026}{{\ttfamily
  arXiv:2005.12026}}.

\bibitem{preskill2018quantum}
J.~Preskill, ``Quantum Computing in the {NISQ} era and beyond,''
  \href{http://dx.doi.org/10.22331/q-2018-08-06-79}{{\em Quantum} {\bfseries
  2}, 79 (2018)}.

\bibitem{eisert2020quantum}
J.~Eisert, D.~Hangleiter, N.~Walk, I.~Roth, D.~Markham, R.~Parekh, U.~Chabaud,
  and E.~Kashefi, ``Quantum certification and benchmarking,''
  \href{http://dx.doi.org/10.1038/s42254-020-0186-4}{{\em Nature Reviews
  Physics} {\bfseries 2}, 382--390 (2020)}.

\bibitem{d2003quantum}
G.~M. D'Ariano, M.~G. Paris, and M.~F. Sacchi, ``Quantum tomography,'' {\em
  Advances in Imaging and Electron Physics} {\bfseries 128}, 206--309 (2003),
  \href{http://arxiv.org/abs/quant-ph/0302028}{{\ttfamily
  arXiv:quant-ph/0302028}}.

\bibitem{lvovsky2009continuous}
A.~I. Lvovsky and M.~G. Raymer, ``Continuous-variable optical quantum-state
  tomography,'' \href{http://dx.doi.org/10.1103/RevModPhys.81.299}{{\em Reviews
  of Modern Physics} {\bfseries 81}, 299 (2009)}.

\bibitem{chabaud2020building}
U.~Chabaud, T.~Douce, F.~Grosshans, E.~Kashefi, and D.~Markham, ``Building
  Trust for Continuous Variable Quantum States,''
  \href{http://dx.doi.org/10.4230/LIPIcs.TQC.2020.3}{in {\em 15th Conference on
  the Theory of Quantum Computation, Communication and Cryptography}}.
\newblock 2020.

\bibitem{terhal2001family}
B.~M. Terhal, ``A family of indecomposable positive linear maps based on
  entangled quantum states,''
  \href{http://dx.doi.org/10.1016/S0024-3795(00)00251-2}{{\em Linear Algebra
  and its Applications} {\bfseries 323}, 61--73 (2001)}.

\bibitem{lewenstein2000optimization}
M.~Lewenstein, B.~Kraus, J.~I. Cirac, and P.~Horodecki, ``Optimization of
  entanglement witnesses,''
  \href{http://dx.doi.org/10.1103/PhysRevA.62.052310}{{\em Physical Review A}
  {\bfseries 62}, 052310 (2000)}.

\bibitem{mari2011directly}
A.~Mari, K.~Kieling, B.~M. Nielsen, E.~Polzik, and J.~Eisert, ``Directly
  estimating nonclassicality,''
  \href{http://dx.doi.org/10.1103/physrevlett.106.010403}{{\em Physical Review
  Letters} {\bfseries 106}, 010403 (2011)}.

\bibitem{kiesel2012universal}
T.~Kiesel and W.~Vogel, ``Universal nonclassicality witnesses for harmonic
  oscillators,'' \href{http://dx.doi.org/10.1103/PhysRevA.85.062106}{{\em
  Physical Review A} {\bfseries 85}, 062106 (2012)}.

\bibitem{chabaud2020certification}
U.~Chabaud, G.~Roeland, M.~Walschaers, F.~Grosshans, V.~Parigi, D.~Markham, and
  N.~Treps, ``Certification of non-Gaussian states with operational
  measurements,'' \href{http://arxiv.org/abs/2011.04320}{{\ttfamily
  arXiv:2011.04320}}.

\bibitem{lasserre2001global}
J.-B. Lasserre, ``Global optimization with polynomials and the problem of
  moments,'' \href{http://dx.doi.org/10.1137/S1052623400366802}{{\em SIAM
  Journal on optimization} {\bfseries 11}, 796--817 (2001)}.

\bibitem{parrilo2000structured}
P.~A. Parrilo, {\em Structured semidefinite programs and semialgebraic geometry
  methods in robustness and optimization}.
\newblock PhD thesis, \href{http://dx.doi.org/10.7907/2K6Y-CH43}{California
  Institute of Technology}, 2000.

\bibitem{lasserre2011new}
J.~B. Lasserre, ``A new look at nonnegativity on closed sets and polynomial
  optimization,'' \href{http://dx.doi.org/10.1137/100806990}{{\em SIAM Journal
  on Optimization} {\bfseries 21}, 864--885 (2011)}.

\bibitem{nielsen2002quantum}
M.~A. Nielsen and I.~L. Chuang, ``Quantum Computation and Quantum Information:
  10th Anniversary Edition,''.
\newblock \href{http://dx.doi.org/10.1017/CBO9780511976667}{Cambridge
  University Press}, New York, NY, USA, 10th~ed., 2011.

\bibitem{weedbrook2012gaussian}
C.~Weedbrook, S.~Pirandola, R.~Garc{\'\i}a-Patr{\'o}n, N.~J. Cerf, T.~C. Ralph,
  J.~H. Shapiro, and S.~Lloyd, ``Gaussian quantum information,''
  \href{http://dx.doi.org/10.1103/RevModPhys.84.621}{{\em Reviews of Modern
  Physics} {\bfseries 84}, 621 (2012)}.

\bibitem{wunsche1998laguerre}
A.~W{\"u}nsche, ``Laguerre 2D-functions and their application in quantum
  optics,'' \href{http://dx.doi.org/10.1088/0305-4470/31/40/017}{{\em Journal
  of Physics A: Mathematical and General} {\bfseries 31}, 8267 (1998)}.

\bibitem{royer1977wigner}
A.~Royer, ``Wigner function as the expectation value of a parity operator,''
  \href{http://dx.doi.org/10.1103/PhysRevA.15.449}{{\em Physical Review A}
  {\bfseries 15}, 449 (1977)}.

\bibitem{banaszek1999direct}
K.~Banaszek, C.~Radzewicz, K.~W{\'o}dkiewicz, and J.~Krasi{\'n}ski, ``Direct
  measurement of the Wigner function by photon counting,''
  \href{http://dx.doi.org/10.1103/PhysRevA.60.674}{{\em Physical Review A}
  {\bfseries 60}, 674 (1999)}.

\bibitem{cahill1969density}
K.~E. Cahill and R.~J. Glauber, ``Density operators and quasiprobability
  distributions,'' \href{http://dx.doi.org/10.1103/PhysRev.177.1882}{{\em
  Physical Review} {\bfseries 177}, 1882 (1969)}.

\bibitem{husimi1940some}
K.~Husimi, ``Some formal properties of the density matrix,''
  \href{http://dx.doi.org/10.11429/ppmsj1919.22.4_264}{{\em Proceedings of the
  Physico-Mathematical Society of Japan. 3rd Series} {\bfseries 22}, 264--314
  (1940)}.

\bibitem{richter1998determination}
T.~Richter, ``Determination of photon statistics and density matrix from double
  homodyne detection measurements,''
  \href{http://dx.doi.org/10.1080/09500349808230666}{{\em Journal of Modern
  Optics} {\bfseries 45}, 1735--1749 (1998)}.

\bibitem{chabaud2020efficient}
U.~Chabaud, F.~Grosshans, E.~Kashefi, and D.~Markham, ``Efficient verification
  of Boson Sampling,'' \href{http://arxiv.org/abs/2006.03520}{{\ttfamily
  arXiv:2006.03520}}.

\bibitem{ferraro2005gaussian}
A.~Ferraro, S.~Olivares, and M.~G. Paris, ``Gaussian states in continuous
  variable quantum information,''
  \href{http://arxiv.org/abs/quant-ph/0503237}{{\ttfamily
  arXiv:quant-ph/0503237}}.

\bibitem{albarelli2018resource}
F.~Albarelli, M.~G. Genoni, M.~G. Paris, and A.~Ferraro, ``Resource theory of
  quantum non-Gaussianity and Wigner negativity,''
  \href{http://dx.doi.org/10.1103/PhysRevA.98.052350}{{\em Physical Review A}
  {\bfseries 98}, 052350 (2018)}.

\bibitem{takagi2018convex}
R.~Takagi and Q.~Zhuang, ``Convex resource theory of non-Gaussianity,''
  \href{http://dx.doi.org/10.1103/PhysRevA.97.062337}{{\em Physical Review A}
  {\bfseries 97}, 062337 (2018)}.

\bibitem{zhuang2018resource}
Q.~Zhuang, P.~W. Shor, and J.~H. Shapiro, ``Resource theory of non-Gaussian
  operations,'' \href{http://dx.doi.org/10.1103/PhysRevA.97.052317}{{\em
  Physical Review A} {\bfseries 97}, 052317 (2018)}.

\bibitem{chabaud2020stellar}
U.~Chabaud, D.~Markham, and F.~Grosshans, ``Stellar representation of
  non-Gaussian quantum states,''
  \href{http://dx.doi.org/10.1103/PhysRevLett.124.063605}{{\em Physical Review
  Letters} {\bfseries 124}, 063605 (2020)}.

\bibitem{vandenberghe1996semidefinite}
L.~Vandenberghe and S.~Boyd, ``Semidefinite programming,''
  \href{http://dx.doi.org/10.1137/1038003}{{\em SIAM review} {\bfseries 38},
  49--95 (1996)}.

\bibitem{fiuravsek2013witnessing}
J.~Fiur{\'a}{\v{s}}ek and M.~Je{\v{z}}ek, ``Witnessing negativity of Wigner
  function by estimating fidelities of catlike states from homodyne
  measurements,'' \href{http://dx.doi.org/10.1103/PhysRevA.87.062115}{{\em
  Physical Review A} {\bfseries 87}, 062115 (2013)}.

\bibitem{PhysRevLett.119.183601}
M.~Walschaers, C.~Fabre, V.~Parigi, and N.~Treps, ``Entanglement and Wigner
  Function Negativity of Multimode Non-Gaussian States,''
  \href{http://dx.doi.org/10.1103/PhysRevLett.119.183601}{{\em Physical Review
  Letters} {\bfseries 119}, 183601 (2017)}.

\bibitem{codes}
U.~Chabaud and P.-E. Emeriau, ``Zeilberger's algorithm and Hierarchy of
  semidefinite programs.'' \textit{Software Heritage} repository
  \href{https://archive.softwareheritage.org/swh:1:dir:d98f70e386783ef69bf8c2ecafdb7b328b19b7ec/}{\texttt{swh:1:dir:d98f70e386783ef69\\bf8c2ecafdb7b328b19b7ec}}
  containing the numerical tools developed for this article.

\bibitem{ourjoumtsev2006generating}
A.~Ourjoumtsev, R.~Tualle-Brouri, J.~Laurat, and P.~Grangier, ``Generating
  optical Schr{\"o}dinger kittens for quantum information processing,''
  \href{http://dx.doi.org/10.1126/science.1122858}{{\em Science} {\bfseries
  312}, 83--86 (2006)}.

\bibitem{sanders1992entangled}
B.~C. Sanders, ``Entangled coherent states,''
  \href{http://dx.doi.org/10.1103/PhysRevA.45.6811}{{\em Physical Review A}
  {\bfseries 45}, 6811 (1992)}.

\bibitem{zurek2001sub}
W.~H. Zurek, ``Sub-Planck structure in phase space and its relevance for
  quantum decoherence,'' \href{http://dx.doi.org/10.1038/35089017}{{\em Nature}
  {\bfseries 412}, 712--717 (2001)}.

\bibitem{sagnol2012picos}
G.~Sagnol and M.~Stahlberg, ``Picos, a python interface to conic optimization
  solvers,'' in {\em Proceedings of the in 21st International Symposium on
  Mathematical Programming}.
\newblock 2012.

\bibitem{mosek}
M.~ApS, {\em MOSEK Optimizer API for Python 9.2.36}, 2019.
\newblock \url{https://docs.mosek.com/9.2/pythonapi/index.html}.

\bibitem{nakata2010numerical}
M.~Nakata, ``A numerical evaluation of highly accurate multiple-precision
  arithmetic version of semidefinite programming solver: SDPA-GMP,-QD
  and-DD.,'' in {\em 2010 IEEE International Symposium on Computer-Aided
  Control System Design}, pp.~29--34, IEEE.
\newblock 2010.

\bibitem{fujisawa2002sdpa}
K.~Fujisawa, M.~Kojima, K.~Nakata, and M.~Yamashita, {\em SDPA (SemiDefinite
  Programming Algorithm) User’s Manual—Version 6.2. 0}, 2002.

\bibitem{barvinok_02}
A.~Barvinok, ``A course in convexity,'', vol.~54 of {\em Graduate Studies in
  Mathematics}.
\newblock \href{http://dx.doi.org/10.1090/gsm/054}{American Mathematical
  Society}, 2002.

\bibitem{szego1959orthogonal}
G.~Szeg{\"o}, ``Orthogonal Polynomials, revised ed,''
  \href{http://dx.doi.org/10.1090/coll/023}{in {\em American Mathematical
  Society Colloquium Publications}}, vol.~23.
\newblock 1959.

\bibitem{Nikodym1930}
O.~Nikodym, ``Sur une généralisation des intégrales de M. J. Radon,''
  \href{http://dx.doi.org/10.4064/fm-15-1-131-179}{{\em Fundamenta
  Mathematicae} {\bfseries 15}, 131--179 (1930)}.

\bibitem{guillemot1971developpements}
M.~Guillemot-Teissier, ``D{\'e}veloppements des distributions en s{\'e}ries de
  fonctions orthogonales. S{\'e}ries de Legendre et de Laguerre,'' {\em Annali
  della Scuola Normale Superiore di Pisa-Classe di Scienze} {\bfseries 25},
  519--573 (1971).

\bibitem{reed1975ii}
M.~Reed and B.~Simon, ``II: Fourier Analysis, Self-Adjointness,'', vol.~2.
\newblock Elsevier, 1975.

\bibitem{riesz1923probleme}
M.~Riesz, ``Sur le probl\`eme des moments, Troisi\`eme Note,'' {\em Ark. Mat.
  Fys} {\bfseries 16}, 1--52 (1923).

\bibitem{haviland1936momentum}
E.~Haviland, ``On the momentum problem for distribution functions in more than
  one dimension. II,'' \href{http://dx.doi.org/10.2307/2371063}{{\em American
  Journal of Mathematics} {\bfseries 58}, 164--168 (1936)}.

\bibitem{hilbert1888darstellung}
D.~Hilbert, ``{\"U}ber die darstellung definiter formen als summe von
  formenquadraten,'' \href{http://dx.doi.org/10.1007/BF01443605}{{\em
  Mathematische Annalen} {\bfseries 32}, 342--350 (1888)}.

\bibitem{gould1972combinatorial}
H.~W. Gould, ``Combinatorial Identities: A standardized set of tables listing
  500 binomial coefficient summations,''.
\newblock Morgantown, W Va, 1972.

\bibitem{zeilberger1991method}
D.~Zeilberger, ``The method of creative telescoping,''
  \href{http://dx.doi.org/10.1016/S0747-7171(08)80044-2}{{\em Journal of
  Symbolic Computation} {\bfseries 11}, 195--204 (1991)}.

\bibitem{leonhardt1995quantum}
U.~Leonhardt, ``Quantum-state tomography and discrete Wigner function,''
  \href{http://dx.doi.org/10.1103/PhysRevLett.74.4101}{{\em Physical Review
  Letters} {\bfseries 74}, 4101 (1995)}.

\bibitem{gross2006hudson}
D.~Gross, ``Hudson’s theorem for finite-dimensional quantum systems,''
  \href{http://dx.doi.org/10.1063/1.2393152}{{\em Journal of mathematical
  physics} {\bfseries 47}, 122107 (2006)}.

\bibitem{spekkens2008negativity}
R.~W. Spekkens, ``Negativity and contextuality are equivalent notions of
  nonclassicality,''
  \href{http://dx.doi.org/10.1103/PhysRevLett.101.020401}{{\em Physical Review
  Letters} {\bfseries 101}, 020401 (2008)}.

\bibitem{delfosse2017equivalence}
N.~Delfosse, C.~Okay, J.~Bermejo-Vega, D.~E. Browne, and R.~Raussendorf,
  ``Equivalence between contextuality and negativity of the Wigner function for
  qudits,'' \href{http://dx.doi.org/10.1088/1367-2630/aa8fe3}{{\em New Journal
  of Physics} {\bfseries 19}, 123024 (2017)}.

\bibitem{raussendorf2017contextuality}
R.~Raussendorf, D.~E. Browne, N.~Delfosse, C.~Okay, and J.~Bermejo-Vega,
  ``Contextuality and Wigner-function negativity in qubit quantum
  computation,'' \href{http://dx.doi.org/10.1103/PhysRevA.95.052334}{{\em
  Physical Review A} {\bfseries 95}, 052334 (2017)}.

\bibitem{howard2014contextuality}
M.~Howard, J.~Wallman, V.~Veitch, and J.~Emerson, ``Contextuality supplies the
  `magic' for quantum computation,''
  \href{http://dx.doi.org/10.1038/nature13460}{{\em Nature} {\bfseries 510},
  351 (2014)}.

\bibitem{bermejo2017contextuality}
J.~Bermejo-Vega, N.~Delfosse, D.~E. Browne, C.~Okay, and R.~Raussendorf,
  ``Contextuality as a resource for models of quantum computation with
  qubits,'' \href{http://dx.doi.org/10.1103/PhysRevLett.119.120505}{{\em
  Physical Review Letters} {\bfseries 119}, 120505 (2017)}.

\bibitem{barbosa2019continuous}
R.~S. Barbosa, T.~Douce, P.-E. Emeriau, E.~Kashefi, and S.~Mansfield,
  ``Continuous-variable nonlocality and contextuality,''
  \href{http://arxiv.org/abs/1905.08267}{{\ttfamily arXiv:1905.08267}}.

\bibitem{navascues2008convergent}
M.~Navascu{\'e}s, S.~Pironio, and A.~Ac{\'\i}n, ``A convergent hierarchy of
  semidefinite programs characterizing the set of quantum correlations,''
  \href{http://dx.doi.org/10.1088/1367-2630/10/7/073013}{{\em New Journal of
  Physics} {\bfseries 10}, 073013 (2008)}.

\bibitem{curto2008analogue}
R.~E. Curto and L.~A. Fialkow, ``An analogue of the Riesz--Haviland theorem for
  the truncated moment problem,''
  \href{http://dx.doi.org/10.1016/j.jfa.2008.09.003}{{\em Journal of Functional
  Analysis} {\bfseries 255}, 2709--2731 (2008)}.

\bibitem{henrion_14}
D.~Henrion and M.~Korda, ``Convex computation of the region of attraction of
  polynomial control systems,''
  \href{http://dx.doi.org/10.1109/TAC.2013.2283095}{{\em IEEE Transactions on
  Automatic Control} {\bfseries 59}, 297--312 (2014)}.

\bibitem{lasserre_10}
J.-B. Lasserre, ``Moments, positive polynomials and their applications,''
  \href{http://dx.doi.org/10.1142/p665}{in {\em Series on Optimization and its
  Applications}}, vol.~1.
\newblock Imperial College Press, 2009.

\end{thebibliography}\endgroup

\onecolumn\newpage
\appendix

\section{Riesz--Haviland theorem in Laguerre basis}
\label{sec:app_RHtheo}

Let $\bm\nu=(\nu_l)_{l\in\N}\in\R^{\N}$. Let us introduce the Riesz functional
\begin{equation}\label{eq:Rieszfunc}
    \begin{aligned}
        L_{\bm\nu}: \quad\R[x]\quad&\longrightarrow\quad\R\\
        P(x)=\sum_{k=0}^pp_kx^k\quad&\longmapsto\quad\sum_{k=0}^pp_k\nu_k,
    \end{aligned}
\end{equation}
which maps real polynomials to real numbers. Let $K$ be a closed subset of $\R$. We say that $L_{\bm\nu}$ is $K$-non-negative if $L_{\bm\nu}(P)\ge0$ for all $P\in\R[x]$ non-negative on $K$.
We recall the classical Riesz--Haviland theorem~\cite{riesz1923probleme,haviland1936momentum} (see, e.g.,~\cite{curto2008analogue} for a recent formulation):

\begin{theorem}[Riesz--Haviland]\label{th:RH}
The sequence $\bm\nu=(\nu_k)_{k\in\N}\in\R^{\N}$ is the sequence of moments $\int_Kx^kd\nu(x)$ of a non-negative distribution $\nu$ supported on $K$ if and only if $L_{\bm\nu}$ is $K$-non-negative.
\end{theorem}

\noindent We prove a modified version of this result in the basis of Laguerre functions. To that end, we introduce the following change of basis:

\begin{lemma}\label{lemmapp:changeofbasis}
Let $\bm\mu,\bm\nu\in\R^{\N}$. For all $m\in\N$, the following conditions are equivalent:
\begin{enumerate}[label=(\roman*)]
\item $\forall k\in\llbracket0,m\rrbracket,\quad\mu_k=\sum_{l=0}^k\nu_l\frac{(-1)^{k+l}}{l!}\binom kl$,
\item $\forall l\in\llbracket0,m\rrbracket,\quad\nu_l=\sum_{k=0}^l\mu_k\binom lkl!$.
\end{enumerate}
\end{lemma}

\noindent As a direct consequence, we retrieve the formula:
\begin{equation}\label{eq:Lagtox}
    x^l=\sum_{k=0}^l(-1)^k\binom lkl!L_k(x),
\end{equation}
for all $l\in\N$ and all $x\in\R_+$.

\begin{proof}
(i)$\Rightarrow$(ii): suppose that
\begin{equation}\label{eq:mutonuchange}
    \forall k\in\llbracket0,m\rrbracket,\quad\mu_k=\sum_{p=0}^k\nu_p\frac{(-1)^{k+p}}{p!}\binom kp.
\end{equation}
Then, for all $l\in\llbracket0,m\rrbracket$,
\begin{equation}\label{eq:derivationmutonu}
    \begin{aligned}
        \sum_{k=0}^l\mu_k\binom lkl!&=\sum_{k=0}^l\sum_{p=0}^k\nu_p\frac{(-1)^{k+p}}{p!}\binom kp\binom lkl!\\
        &=\sum_{p=0}^l\nu_p\frac{l!}{p!}\binom lp\sum_{k=p}^l(-1)^{k-p}\binom{l-p}{k-p}\\
        &=\sum_{p=0}^l\nu_p\frac{l!}{p!}\binom lp\sum_{q=0}^{l-p}(-1)^q\binom{l-p}q\\
        &=\nu_l,
    \end{aligned}
\end{equation}
where we used~\refeq{mutonuchange} in the first line and the binomial theorem in the last line which imposes $l=p$.

(ii)$\Rightarrow$(i): suppose that
\begin{equation}\label{eq:nutomuchange}
    \forall l\in\llbracket0,m\rrbracket,\quad\nu_l=\sum_{p=0}^l\mu_p\binom lpl!.
\end{equation}
Then, for all $k\in\llbracket0,m\rrbracket$,
\begin{equation}
    \begin{aligned}
        \sum_{l=0}^k\nu_l\frac{(-1)^{k+l}}{l!}\binom kl&=\sum_{l=0}^k\sum_{p=0}^l\mu_p(-1)^{k+l}\binom lp\binom kl\\
        &=\sum_{p=0}^k\mu_p(-1)^{k+p}\binom kp\sum_{l=p}^k(-1)^{l-p}\binom{k-p}{l-p}\\
        &=\sum_{p=0}^k\mu_p(-1)^{k+p}\binom kp\sum_{q=0}^{k-p}(-1)^q\binom{k-p}q\\
        &=\mu_k,
    \end{aligned}
\end{equation}
where we used~\refeq{nutomuchange} in the first line and the binomial theorem in the last line which imposes $k=p$.

\end{proof}

\noindent We may now prove the Riesz--Haviland theorem in Laguerre basis (Theorem~\ref{th:RHLaguerre} from the main text):

\setcounter{theorem}{0}

\begin{theorem}
Let $\bm\mu=(\mu_k)_{k\in\N}\in\R^{\N}$. The sequence $\bm\mu$ is the sequence of Laguerre moments $\int_{\R_+}\mathcal L_k(x)d\mu(x)$ of a non-negative distribution $\mu$ supported on $\R_+$ if and only if
\begin{equation}
    \forall m\in\N,\forall g\in\mathcal R_{m,+}(\R_+),\;\braket{f_{\bm\mu},g}\ge0.
\end{equation}
\end{theorem}

\begin{proof}
Let $\bm\mu=(\mu_k)_{k\in\N}\in\R^{\N}$, and suppose that the sequence $\bm\mu$ is the sequence of Laguerre moments $\int_{\R_+}\mathcal L_k(x)d\mu(x)$ of a non-negative distribution $\mu$ supported on $\R_+$.

Let $m\ge0$ and let $g=\sum_{k=0}^mg_k\mathcal L_k\in\mathcal R_{m,+}(\R_+)$. The distribution $\mu$ is non-negative, so $\braket{\mu,g}\ge0$. Moreover,
\begin{equation}
    \begin{aligned}
        \braket{f_{\bm\mu},g}&=\sum_{k=0}^m{\mu_kg_k}\\
        &=\int_{\R_+}\sum_{k=0}^m{g_k\mathcal L_kd\mu}\\
        &=\braket{\mu,g}.
    \end{aligned}
\end{equation}
Hence, for all $m\in\N$ and all $g\in\mathcal R_{m,+}(\R_+)$, $\braket{f_{\bm\mu},g}\ge0$. 

Conversely, let $\bm\mu=(\mu_k)_{k\in\N}\in\R^{\N}$, and suppose that for all $m\in\N$ and all $g\in\mathcal R_{m,+}(\R_+)$, $\braket{f_{\bm\mu},g}\ge0$.
We define the sequence $\bm\nu=(\nu_l)_{l\in\N}\in\R^{\N}$ by
\begin{equation}\label{eq:nu_ldef}
    \nu_l:=\sum_{k=0}^l\mu_k\binom lkl!,
\end{equation}
for all $l\in\N$. 

Let $m\in\N$ and let $P(x)=\sum_{l=0}^mp_lx^l$ be a non-negative polynomial over $\R_+$. By~\refeq{Lagtox}, for all $x\in\R_+$,
\begin{equation}\label{eq:PinLag}
    \begin{aligned}
        P(x)&=\sum_{l=0}^mp_l\sum_{k=0}^l(-1)^k\binom lkl!L_k(x)\\
        &=\sum_{k=0}^m(-1)^kL_k(x)\left(\sum_{l=k}^mp_l\binom lkl!\right).
    \end{aligned}
\end{equation}
Let $g_P(x):=P(x)e^{-\frac x2}$, for $x\in\R_+$. We have $g_P\in\mathcal R_{m,+}(\R_+)$, so $\braket{f_{\bm\mu},g_P}\ge0$. Moreover, with~\refeq{PinLag}
\begin{equation}
    \begin{aligned}
        \braket{f_{\bm\mu},g_P}&=\sum_{k=0}^m\mu_k\left(\sum_{l=k}^mp_l\binom lkl!\right)\\
        &=\sum_{l=0}^m\left(\sum_{k=0}^l\mu_k\binom lkl!\right)p_l\\
        &=L_{\bm\nu}(P)
    \end{aligned}
\end{equation}
where we used~\refeq{nu_ldef} and the definition of the Riesz functional from~\refeq{Rieszfunc} in the last line. 
In particular, $L_{\bm\nu}(P)\ge0$, and this holds for all non-negative polynomials $P$ over $\R_+$. By the Riesz--Haviland theorem (Theorem~\ref{th:RH}), this implies that $\bm\nu$ is the sequence of moments of a non-negative distribution $\nu$ supported on $\R_+$.

Furthermore, we have that for all $k\in\N$:
\begin{equation}
    \begin{aligned}
        \mu_k&=\sum_{l=0}^k\nu_l\frac{(-1)^{k+l}}{l!}\binom kl\\
        &=\sum_{l=0}^k\frac{(-1)^{k+l}}{l!}\binom kl\int_{\R_+}x^ld\nu(x)\\
        &=\int_{\R_+}(-1)^k\sum_{l=0}^k\frac{(-1)^l}{l!}\binom klx^ld\nu(x)\\
        &=\int_{\R_+}(-1)^kL_k(x)d\nu(x)\\
        &=\int_{\R_+}\mathcal L_k(x)e^{\frac x2}d\nu(x)\\
    \end{aligned}
\end{equation}
where we used Lemma~\ref{lemmapp:changeofbasis} in the first line. Hence, $\bm\mu$ is the sequence of Laguerre moments of the distribution $\mu(x):=e^{\frac x2}\nu(x)$ supported on $\R_+$, which is non-negative since $\nu$ is non-negative.

\end{proof}

\newpage
\section{Theory for infinite-dimensional linear programs}
\label{sec:app_std_form}

This appendix is dedicated to expressing formally our linear program as presented in \cite[IV--(6.1)]{barvinok_02} so that readers unfamiliar with global optimisation may better understand why \refprog{LP} is indeed an infinite-dimensional linear program and how to derive its dual program. 
We will derive it for the problem over square-integrable functions. The exact same treatment can be applied for the case of optimising over Schwartz functions which will involve $\mathcal S(\R_ +)$ and its dual space $\mathcal S'(\R_+)$, the space of tempered distributions.
We recall our initial program \refprog{LP}:
\leqnomode
\begin{flalign*}
    \label{prog:LP_recall}
        \tag*{(LP$_n^{L^2}$)}
        \hspace{4cm}
        \left\{
        \begin{aligned}
            & & & \Sup{(F_k)_{k \in \N} \in \R^{\N}} F_n  \\
            & \text{subject to} & & \sum_k F_k = 1 \\
            & \text{and} & & \forall k \in \N, \quad F_k \geq 0 \\
            & \text{and} & & \forall x \in \R_+, \quad \sum_k F_k \mathcal L_k(x) \geq 0.
        \end{aligned}
    \right. &&
\end{flalign*}
\reqnomode

Let us introduce the spaces:
\begin{itemize}
    \item $E_1 = \ell^2 \times L^2(\R_ +)$\footnote{Recall that via expansion on a basis of $L^2(\R_ +)$, $L^2(\R_ +)$ and $\ell^2$ are isomomorphic and that the family of Laguerre functions forms a basis of $L^2(\R_ +)$.}.
    \item $F_1 = \ell^2 \times L^2(\R_ +)$\footnote{The spaces $L^2(\R_ +)$ and ${L^2}'(\R_ +)$ are isomomorphic by the Radon--Nikodym theorem. We associate a measure in the dual of ${L^2}(\R_+)$ with the Lebesgue measure on $\R_+$ times the corresponding function in $L^2(\R_+)$.} the dual space of $E_1$.
    \item $E_2 = \R \times L^2(\R_ +)$.
    \item $F_2 = \R \times L^2(\R_ +)$ the dual space of $E_2$.
\end{itemize}
We also define the dualities $\langle \dummy , \dummy \rangle_1 : E_1 \times F_1 \longrightarrow \R$ and $\langle \dummy , \dummy \rangle_2 : E_2 \times F_2 \longrightarrow \R$ as follows:
\begin{equation}
    \begin{aligned}
    & \forall e_1 = ((u_k),f) \in E_1, \forall f_1 = ((v_k),\mu) \in F_1, & & \langle e_1,f_1 \rangle_1 \defeq \sum_k u_k v_k + \int_{\R_ +}{f}{d\mu},  \\ 
    & \forall e_2 = (x,f) \in E_2, \forall f_2 = (y,\mu) \in F_2, & & \langle e_2,f_2 \rangle_2 \defeq  xy + \int_{\R_ +}{f}{d\mu}.
    \end{aligned}
\end{equation}
Let $\fdec{A}{E_1}{E_2}$ be the following linear transformation:
\begin{equation}
    \forall e_1=((u_k),f) \in E_1, \quad A(e_1) \defeq \left(\sum_k u_k,x \in \R_ + \mapsto f(x) - \sum_k u_k \mathcal L_k(x)\right),
\end{equation}
and $\fdec{A^*}{F_2}{F_1}$ be defined as:
\begin{equation}
    \forall f_2=(y,\mu) \in F_2, \quad A^*(f_2) \defeq \left( (y-\int_{\R_ +}{\mathcal L_k}{d\mu})_{k \in \N},\mu \right).
\end{equation}
We can easily verify that $A^*$ is the dual transformation of $A$, i.e., $\forall e_1 \in E_1, \forall f_2 \in F_2$ we have $\langle A(e_1),f_2 \rangle_2 = \langle e_1, A^*(f_2) \rangle_1 $.

Recall that $L^2_+(\R_+)$ is the cone of non-negative functions in $L^2(\R_+)$ and $\ell^2_+$ the cone of sequences in $\ell^2$ with non-negative coefficients.
We will optimise in the convex cones $K_1 = \ell^2_+ \times L^2_+(\R_+) \subset E_1$ and $K_2 = \{ 0 \}$. The dual cones are then respectively: $K_1^* = \{f_1 \in F_1: \forall e_1 \in K_1, \; \langle e_1,f_1 \rangle \geq 0 \}$ and $K_2^* = F_2$. 

We can now rewrite the problem \refprog{LP} as a standard linear program in convex cones. We choose the vector function in the objective to be $c_n = ( (\delta_{kn})_k, \mathbf{0}) \in F_1$ and we also set $b = (1,\mathbf{0}) \in E_2$ for the constraints. The standard form of~\refprog{LP_recall} in the sense of \cite{barvinok_02} can be written as follows:
\leqnomode
\begin{flalign*}
    \label{prog:LPstdform}
    \tag*{(LP$_n^{L^2}$)}
    \hspace{2.8cm}\left\{
    \begin{aligned}
        &  & & \sup_{e_1 \in E_1} \langle e_1,c_n \rangle_1\\
        & \text{subject to} & & A(e_1) = b \\
        & \text{and} & & e_1 \geq_{K_1} 0. \\
    \end{aligned}
    \right. &&
\end{flalign*}
The standard form of the dual \refprog{DLP} of problem \refprog{LPstdform} can be expressed as follows:
\begin{flalign*}
    \label{prog:DLPstdform}
    \tag*{(D-LP$_n^{L^2}$)}
    \hspace{2.8cm}\left\{
    \begin{aligned}
        & & & \inf_{f_2 \in F_2} \langle b,f_2 \rangle_2 \\
        & \text{subject to} & & A^*(f_2) \geq_{K_1^*} c_n, \\
    \end{aligned}
    \right. &&
\end{flalign*}
which can be expanded as:
\begin{flalign*}
    \label{prog:DLP_recall}
    \tag*{(D-LP$_n^{L^2}$)}
    \hspace{2.8cm}\left\{
        \begin{aligned}
            & & & \Inf{\substack{y \in \R \\ \mu \in {L^2}'(\R_+)}} y  \\
            & \text{subject to} & & \forall k\in \N, \quad y \geq \int_{\R_+}{\mathcal L_k}{d\mu} \\
            & \text{and} & & y \geq 1 + \int_{\R_+}{\mathcal L_n}{d\mu} \\
            & \text{and} & & \forall f \in L^2_+(\R_+),\quad \langle \mu,f \rangle=\int_{\R_+}fd\mu \geq  0. \\
        \end{aligned}
    \right. &&
\end{flalign*}
\reqnomode

\noindent Note that a similar derivation holds for the more general form where one uses a linear combination of fidelities with Fock states. The displaced Fock states version can be obtained by classical post-processing as detailed in section~\ref{sec:Wigner}. For $n \in \N^*$ and some vector $\bm a = (a_1,a_2, \dots, a_n)\in[0,1]^n$, the computation of:
\begin{equation}
    \omega_{\bm a}:=\sup_{\substack{\rho\in\mathcal D(\mathcal H)\\W_\rho \geq 0}}\Tr\left(\hat\Omega_{\bm a,0}\rho\right),
\end{equation}
can be expressed as:
\leqnomode
\begin{flalign*}
    \label{prog:LP_a}
        \tag*{($\text{LP}_{\bm a}^{L^2}$)}
        \hspace{2.8cm}
        \left\{
        \begin{aligned}
            & & & \Sup{(F_k)_{k \in \N} \in \R^{\N}} \sum_{k=1}^n a_k F_k  \\
            & \text{subject to} & & \sum_k F_k = 1 \\
            & \text{and} & & \forall k \in \N, \quad F_k \geq 0 \\
            & \text{and} & & \forall x \in \R_+, \quad \sum_k F_k \mathcal L(x) \geq 0.
        \end{aligned}
    \right. &&
\end{flalign*}
Its dual reads:
\begin{flalign*}
    \label{prog:DLP_a}
    \tag*{($\text{D-LP}_{\bm a}^{L^2}$)}
    \hspace{2.8cm}\left\{
        \begin{aligned}
            & & & \Inf{\substack{y \in \R \\ \mu \in {L^2}'(\R_+)}} y  \\
            & \text{subject to} & & \forall k \le n \in \N, \quad y \geq a_k + \int_{\R_+}{\mathcal L_k}{d\mu} \\
            & \text{and} & & \forall k > n \in \N, \quad y \geq \int_{\R_+}{\mathcal L_k}{d\mu} \\
            & \text{and} & & \forall f \in L^2_+(\R_+), \quad\langle \mu,f \rangle=\int_{\R_+}fd\mu \geq  0. \\
        \end{aligned}
    \right. &&
\end{flalign*}
\reqnomode

\newpage
\section{Proof of technical lemmas}
\label{sec:app_prooflemmas}

\setcounter{lemmapp}{2}

In this section we prove the technical lemmas from section~\ref{sec:SDP}. For completeness, we include the proof of Lemma~\ref{lem:pospolyR} below:

\begin{lemmapp}[\cite{hilbert1888darstellung}]
Let $p\in\mathbb N$ and let $P$ be a univariate polynomial of degree $2p$. Let $X=(1,x,\dots,x^p)$ be the vector of monomials. Then, $P$ is non-negative over $\R$ if and only if there exists a sum of squares decomposition for $P$, i.e., a real $(p+1)\times(p+1)$ positive semidefinite matrix $Q$ such that for all $x\in\mathbb R$,
\begin{equation}
    P(x)=X^TQX.
\end{equation}
\end{lemmapp}

\begin{proof} If for all $x\in\mathbb R$, $P(x)=X^TQX$ with $X=(1,x,\dots,x^p)$ for $Q$ positive semidefinite, then $P$ is clearly non-negative over $\R$.

Conversely, suppose the univariate polynomial $P$ of degree $2p$ is non-negative over $\R$. It can thus be written as a sum of squares of polynomials of degree at most $p$ (e.g., by considering the factorisation of $P$ and writing each term in the product as a sum of squares, given that its zeros on the real line have even multiplicity and that each complex zero is associated to a conjugate zero with same multiplicity): for all $x\in\mathbb R$,
\begin{equation}
P(x)=\sum_i S_i^2(x),
\end{equation}
for some real polynomials $S_i$ of degree at most $p$. Then, for some vectors of coefficients $\bm s_i\in\mathbb R^{p+1}$ we have
\begin{equation}
    \begin{aligned}
P(x)&=\sum_i(\bm s_i^TX)^2\\
    &=\sum_i(X^T\bm s_i)(\bm s_i^TX)\\
    &=X^T\left(\sum_i\bm s_i\bm s_i^T\right)X.
    \end{aligned}
\end{equation}
Setting $Q:=\sum_i\bm s_i\bm s_i^T \succeq 0$ completes the proof.
\end{proof} 

\noindent We now turn to the proof of Lemma~\ref{lem:pospolyR+}:

\begin{lemmapp}\label{lemmapp:pospolyR+}
Non-negative polynomials on $\R_+$ can be written as sums of polynomials of the form $\sum_{l=0}^px^l\sum_{i+j=2l}y_iy_j$, where $p\in\mathbb N$ and $y_i\in\mathbb R$, for all $0\le i\le p$.
\end{lemmapp}

\begin{proof}
Let $P$ be a univariate polynomial of degree $p$ which is non-negative on $\R_+$.
Writing $X=(1,x,\dots,x^p)$, the polynomial $x\mapsto P(x^2)$ of degree $2p$ is non-negative on $\R$, so by Lemma~\ref{lem:pospolyR} there exists a real positive semidefinite matrix $Q=(Q_{ij})_{0\le i,j\le p}$ such that for all $x\in\mathbb R$:
\begin{equation}
    \begin{aligned}
        P(x^2)&=X^TQX\\
              &=\sum_{k=0}^{2p}x^k\sum_{i+j=k}Q_{ij}\\
              &=\sum_{l=0}^px^{2l}\sum_{i+j=2l}Q_{ij},
    \end{aligned}
\end{equation}
where the last line comes from the fact that $x\mapsto P(x^2)$ has no monomial of odd degree. Hence, for all $x\in\mathbb R_+$,
\begin{equation}
    P(x)=\sum_{l=0}^px^l\sum_{i+j=2l}Q_{ij}.
\end{equation}
$Q$ is a real $(p+1)\times(p+1)$ positive semidefinite matrix, so via Cholesky decomposition
\begin{equation}
    Q=\sum_{k=0}^p\bm y^{(k)}\bm y^{(k)T},
\end{equation}
where $\bm y^{(k)}\in\R^{p+1}$ for all $k\in\llbracket 0,p \rrbracket$. We finally obtain, for all $x\in\mathbb R_+$,
\begin{equation}
    \begin{aligned}
        P(x)&=\sum_{l=0}^px^l\sum_{i+j=2l}\sum_{k=0}^p\left(\bm y^{(k)}\bm y^{(k)T}\right)_{ij}\\
            &=\sum_{k=0}^p\left(\sum_{l=0}^px^l\sum_{i+j=2l}\bm y_i^{(k)}\bm y_j^{(k)}\right).
    \end{aligned}
\end{equation}
\end{proof}

\noindent We recall a few definitions from the main text.
For $\bm s\in\R^{\N}$, we define the associated formal series of Laguerre functions:
\begin{equation}
    f_{\bm s}:=\sum_{k\ge0}s_k\mathcal L_k,
\end{equation}
where for all $x\in\R_+$, $\mathcal L_k(x)=(-1)^kL_k(x)e^{-\frac x2}$, with $L_k(x)=\sum_{l=0}^k\frac{(-1)^l}{l!}\binom klx^l$ the $k^{th}$ Laguerre polynomial. For $m\in\mathbb N$, we also define the associated matrix $A_{\bm s}$ by
\begin{equation}
    (A_{\bm s})_{0\le i,j\le m} = \begin{cases} 
      \sum_{k=0}^l s_k \binom lkl! &\text{when } i+j=2l, \\
      0 &\text{otherwise.}
   \end{cases}
\end{equation}
For all $m\in\mathbb N$, the set of series of Laguerre functions over $\R_+$ truncated at $m$ is denoted $\mathcal R_m(\R_+)$, and $\mathcal R_{m,+}(\R_+)$ denotes its subset of non-negative elements.

\begin{lemmapp}\label{lemmapp:momentmatrix}
Let $m\in\mathbb N$ and let $\bm s\in\R^{\N}$. The following propositions are equivalent:
\begin{enumerate}[label=(\roman*)]
\item $\forall g\in\mathcal R_{m,+}(\R_+),\;\braket{f_{\bm s},g}\ge0$,
\item $A_{\bm s}\succeq0$.
\end{enumerate}
\end{lemmapp}

\begin{proof}
By Lemma~\ref{lemmapp:pospolyR+}, any non-negative polynomial over $\R_+$ of degree less or equal to $m$ can be expressed as a sum of polynomials of the form $\sum_{l=0}^mx^l\sum_{i+j=2l}y_iy_j$, where $Y=(y_0,\dots,y_m)\in\mathbb R^{m+1}$. Hence, any non-negative truncated Laguerre series (the elements of $\mathcal R_{m,+}(\R_+)$) can be expressed as as sum of terms of the form $e^{-\frac x2}\sum_{l=0}^mx^l\sum_{i+j=2l}y_iy_j$. By linearity, it is sufficient to check that the scalar products with these expressions are non-negative.

For all $k\in\mathbb N$ we have
\begin{equation}\label{eq:sk}
    s_k=\int_{\R_+}\mathcal L_k(x)f_{\bm s}(x)dx.
\end{equation}
Thus,
\begin{equation}
    \begin{aligned}
          A_{\bm s}\succeq0&\Leftrightarrow\forall Y\in\R^{m+1},\;Y^TA_{\bm s}Y\ge0\\
                           &\Leftrightarrow\forall Y\in\R^{m+1},\;\sum_{i,j=0}^my_iy_j(A_{\bm s})_{ij}\ge0\\
                           &\Leftrightarrow\forall Y\in\R^{m+1},\;\sum_{l=0}^m\sum_{i+j=2l}^my_iy_j\sum_{k=0}^ls_k\binom lkl!\ge0\\
                           &\Leftrightarrow\forall Y\in\R^{m+1},\;\int_{\R_+}\sum_{l=0}^m\sum_{i+j=2l}^my_iy_j\sum_{k=0}^l\binom lkl!\mathcal L_k(x)f_{\bm s}(x)dx\ge0\\
                           &\Leftrightarrow\forall Y\in\R^{m+1},\;\int_{\R_+}\sum_{l=0}^m\sum_{i+j=2l}^my_iy_j\sum_{k=0}^l(-1)^k\binom lkl!L_k(x)e^{-\frac x2}f_{\bm s}(x)dx\ge0\\
                           &\Leftrightarrow\forall Y\in\R^{m+1},\;\int_{\R_+}\left(e^{-\frac x2}\sum_{l=0}^mx^l\sum_{i+j=2l}^my_iy_j\right)f_{\bm s}(x)dx\ge0\\
                           &\Leftrightarrow\forall Y\in\R^{m+1},\;\left\langle f_{\bm s},x\mapsto e^{-\frac x2}\sum_{l=0}^mx^l\sum_{i+j=2l}^my_iy_j\right\rangle\ge0\\
                           &\Leftrightarrow\forall g\in\mathcal R_{m,+}(\R_+),\;\braket{f_{\bm s},g}\ge0,
    \end{aligned}
\end{equation}
where we used~\refeq{sk} in the fourth line and~\refeq{Lagtox} in the sixth line.
\end{proof}

\newpage
\section{Dual semidefinite programs}
\label{sec:app_dualSDP}

In this section, we detail the derivation of the dual semidefinite programs \refprog{upperDSDPn} and \refprog{lowerDSDPn}. The generalisations for \refprog{upperSDP} and \refprog{lowerSDP} are straightforward.

A standard form for a semidefinite program is given by~\cite{vandenberghe1996semidefinite}:
\leqnomode
\begin{flalign*}
   \label{prog:SDP_std}
    \tag*{(SDP)}
    \hspace{2.8cm}\left\{
        \begin{aligned}
            & & & \Sup{X\in\text{Sym}_N} \Tr(C^TX)\\
            & \text{subject to} & & \forall i \in \llbracket 1,M \rrbracket,\quad\Tr(B^{(i)}X)=b_i \\
            & \text{and } & & X \succeq 0,
        \end{aligned}
    \right. &&
\end{flalign*}
\reqnomode
where $M,N\in\mathbb N$, $\bm b=(b_1,\dots,b_M)\in\R^M$, $C\in\text{Sym}_N$, and $B^{(i)}\in\text{Sym}_N$ for all $i \in \llbracket 1,M \rrbracket$. Its dual semidefinite program reads:
\leqnomode
\begin{flalign*}
   \label{prog:DSDP_std}
    \tag*{(D-SDP)}
    \hspace{2.8cm}\left\{
        \begin{aligned}
            & & & \Inf{\bm y\in\R^M}\bm b^T\bm y\\
            & \text{subject to} & & \sum_{i=1}^My_iB^{(i)} \succeq C.
        \end{aligned}
    \right. &&
\end{flalign*}
\reqnomode
%

\subsection{Dual program for the semidefinite relaxations}

We fix $m\ge n$ and we recall below the expression of \refprog{upperSDPn}:
\leqnomode
\begin{flalign*}
   \label{prog:upperSDPn_recall}
    \tag*{$(\text{SDP}^{m,\geq}_n)$}
    \hspace{2.8cm}\left\{
        \begin{aligned}
            & & & \Sup{\substack{A\in\text{Sym}_{m+1}\\ \bm F\in\R^{m+1}}} F_n  \\
            & \text{subject to} & & \sum_{k=0}^{m} F_k = 1  \\
            & \text{and} & & \forall k\le m, \; F_k \geq 0 \\
            & \text{and} & & \forall l\le m,\forall i+j=2l,\quad A_{ij}=\sum_{k=0}^l F_k \binom lkl!\\
            & \text{and} & & \forall l\in\llbracket1,m\rrbracket,\forall i+j=2l-1,\quad A_{ij}=0\\
            & \text{and} & & A \succeq 0.
        \end{aligned}
    \right. &&
\end{flalign*}
\reqnomode
To put \refprog{upperSDPn_recall} in the standard form \refprog{SDP_std} we set $N=2\times(m+1)$ and $M=1+(m+1)^2$. For all $r\in\N^*$ and all $i,j\in\llbracket 1,r \rrbracket$, let $E^{(i,j)}_r$ be the $r\times r$ matrix whose $(i,j)$ entry is $1$ and all other entries are $0$. We set
\begin{equation}
    \begin{aligned}
    &X=\text{Diag}_{k=0,\dots,m}(F_k)\oplus A\in\text{Sym}_N,\\
    &C=E^{(n,n)}_N=E^{(n,n)}_{m+1}\oplus\mymathbb 0_{m+1}\in\text{Sym}_N,\\
    &\bm b=(1,0,0,\dots,0)\in\R^M,\\
    &B^{(0)}=\mymathbb 1_{m+1}\oplus\mymathbb 0_{m+1}\in\text{Sym}_N,\\
    &\forall i,j\in\llbracket0,m\rrbracket,\quad B^{(i,j)}=\begin{cases}\text{Diag}_{k=0,\dots,m}\left(-\binom lkl!\right)\oplus \left(\frac12E^{(i,j)}_{m+1}+\frac12E^{(j,i)}_{m+1}\right)&\text{when } i+j=2l,\\\mymathbb 0_{m+1}\oplus \left(\frac12E^{(i,j)}_{m+1}+\frac12E^{(j,i)}_{m+1}\right) &\text{otherwise,}\end{cases}
    \end{aligned}
\end{equation}
with the convention $\binom lk=0$ when $k>l$. The matrix $B^{(0)}$ corresponds to the constraint $\sum_{k=0}^mF_k=1$, and we denote the corresponding dual variable $y\in\R$. Similarly, the matrices $B^{(i,j)}$ correspond to the $(m+1)^2$ constraints defining the symmetric matrix $A$, and we denote the corresponding dual variables $Q_{ij}\in\R$, with $Q_{ij}=Q_{ji}$ for all $i,j\in\llbracket 0,m \rrbracket$. We write $Q=(Q_{ij})_{0\le i,j\le m}$.
The standard form \refprog{DSDP_std} of the dual program \refprog{upperDSDPn} thus reads:
\leqnomode
\begin{flalign*}
   \label{prog:upperDSDPn_std}
    \tag*{(D-SDP$_n^{m,\ge}$)}\hspace{2.8cm}\left\{
        \begin{aligned}
            & & & \Inf{\substack{Q\in\text{Sym}_{m+1}\\y\in\R}} y\\
            & \text{subject to} & & \text{Diag}_{k=0,\dots,m}\left[y-\sum_{l=0}^m\sum_{i+j=2l}Q_{ij}\binom lkl!\right]\oplus\frac12Q \succeq E^{(n,n)}_{m+1}\oplus\mymathbb 0_{m+1}.\hspace{-4cm}
        \end{aligned}
    \right. &&
\end{flalign*}
\reqnomode 
Due to the block-diagonal structure of the matrices involved, the positive semidefinite constraint above is equivalent to the following constraints:
\begin{equation}
   \left\{ \begin{aligned}
    &y\ge1+\sum_{l=0}^m\sum_{i+j=2l}Q_{ij}\binom lnl!,\\
    &\forall k\in\llbracket0,m\rrbracket\setminus\{n\},\quad y\ge\sum_{l=0}^m\sum_{i+j=2l}Q_{ij}\binom lkl!,\\
    &Q\succeq0.
    \end{aligned}\right.
\end{equation}
For $k\in\llbracket0,m\rrbracket$, we define
\begin{equation}
    \mu_k:=\sum_{l=0}^m\sum_{i+j=2l}Q_{ij}\binom lkl!\in\R.
\end{equation}
We obtain the program:
\leqnomode
\begin{flalign*}
   \label{prog:upperDSDPn_std2}
    \tag*{(D-SDP$_n^{m,\ge}$)}\hspace{2.8cm}\left\{
        \begin{aligned}
            & & & \Inf{\substack{Q\in\text{Sym}_{m+1}\\y,\bm\mu\in\R\times\R^{m+1}}} y\\
            & \text{subject to} & & y\ge1+\mu_n\\
            &\text{and} & & \forall k\in\llbracket0,m\rrbracket\setminus\{n\},\quad y\ge\mu_k\\
            &\text{and} & & \forall k\in\llbracket0,m\rrbracket,\quad \mu_k=\sum_{l=0}^m\sum_{i+j=2l}Q_{ij}\binom lkl!\hspace{-2cm}\\
            &\text{and} & & Q\succeq0.     
        \end{aligned}
    \right. &&
\end{flalign*}
\reqnomode
Finally, in order to obtain the form of \refprog{upperDSDPn} from the main text we prove the following result:

\begin{lemma}\label{lem:revert} Let $\bm u,\bm v\in\R^{m+1}$. The following propositions are equivalent:
\begin{enumerate}[label=(\roman*)]
\item $\forall k\in\llbracket0,m\rrbracket,\quad u_k=\sum_{l=0}^mv_l\binom lkl!$,
\item $\forall l\in\llbracket0,m\rrbracket,\quad v_l=\sum_{k=l}^m \frac{(-1)^{l+k}}{l!} \binom kl u_k$.
\end{enumerate}
\end{lemma}

\begin{proof} The proof is similar to that of Lemma~\ref{lemmapp:changeofbasis}. Note that we could start the first sum at $l=k$ since $\binom lk = 0 $ for $l < k$ but for convenience we start it at $l=0$.

(i)$\Rightarrow$(ii): suppose that
\begin{equation}\label{eq:mutoQ}
    \forall k\in\llbracket0,m\rrbracket,\quad u_k=\sum_{p=0}^mv_l\binom pkp!.
\end{equation}
Then, for all $l\in\llbracket0,m\rrbracket$,
\begin{equation}
    \begin{aligned}
        \frac{(-1)^l}{l!}\sum_{k=l}^m(-1)^k\binom kl u_k&= \frac{(-1)^l}{l!}\sum_{k=l}^m(-1)^k\binom kl\sum_{p=0}^mv_p\binom pkp!\\
        &=\sum_{p=0}^mv_p\frac{p!}{l!}\sum_{k=l}^m(-1)^{k+l}\binom pk\binom kl\\
        &=\sum_{p=l}^mv_p\frac{p!}{l!}\sum_{k=l}^p(-1)^{k+l}\binom pk\binom kl\\
        &=\sum_{p=l}^mv_p\frac{p!}{l!}\sum_{k=l}^p(-1)^{k+l}\frac{p!k!}{k!(p-k)!l!(k-l)!}\\
        &=\sum_{p=l}^mv_p\frac{p!}{l!}\binom pl\sum_{q=0}^{p-l}(-1)^q\binom{p-l}q\\
        &=v_l,
    \end{aligned}
\end{equation}
where we used~\refeq{mutoQ} in the first line, the fact that $\binom pk=0$ if $k>p$ in the third line, $q:=k-l$ in the fifth line, and the binomial theorem in the last line which imposes $p=l$.

(ii)$\Rightarrow$(i): suppose that
\begin{equation}\label{eq:Qtomu}
    \forall l\in\llbracket0,m\rrbracket,\quad v_l=\frac{(-1)^l}{l!}\sum_{p=l}^m(-1)^p\binom pl u_p.
\end{equation}
Then, for all $k\in\llbracket0,m\rrbracket$,
\begin{equation}
    \begin{aligned}
        \sum_{l=0}^mv_l\binom lkl!&=\sum_{l=0}^m\frac{(-1)^l}{l!}\sum_{p=l}^m(-1)^p\binom pl u_p\binom lkl!\\
        &=\sum_{p=0}^mu_p(-1)^p\sum_{l=0}^p(-1)^l\binom pl\binom lk\\
        &=\sum_{p=k}^mu_p(-1)^p\sum_{l=k}^p(-1)^l\binom pl\binom lk\\
        &=\sum_{p=k}^mu_p(-1)^p\sum_{l=k}^p(-1)^l\frac{p!l!}{l!(p-l)!k!(l-k)!}\\
        &=\sum_{p=k}^mu_p(-1)^{p-k}\binom pk\sum_{q=0}^{p-k}(-1)^q\binom{p-k}q\\
        &=u_k,
    \end{aligned}
\end{equation}
where we used~\refeq{Qtomu} in the first line, the fact that $\binom lk=0$ if $k>l$ in the third line, $q:=l-k$ in the fifth line, and the binomial theorem in the last line which imposes $p=k$.

\end{proof}

\noindent Combining Lemma~\ref{lem:revert} for $u_k=\mu_k$ and $v_l=\sum_{i+j=2l}Q_{ij}$ for all $k,l\in\llbracket0,m\rrbracket$ with the previous expression of \refprog{upperDSDPn_std2} we finally obtain:
\leqnomode
\begin{flalign*}
    \tag*{(D-SDP$_n^{m,\ge}$)}\hspace{2.8cm}\left\{
        \begin{aligned}
            & & & \Inf{\substack{Q\in\text{Sym}_{m+1}\\y,\bm\mu\in\R\times\R^{m+1}}} y\\
            & \text{subject to} & & y\ge1+\mu_n\\
            &\text{and} & & \forall k\in\llbracket0,m\rrbracket\setminus\{n\},\quad y\ge\mu_k\\
            &\text{and} & & \forall l\in\llbracket0,m\rrbracket,\quad\sum_{i+j=2l}Q_{ij}=\frac{(-1)^l}{l!}\sum_{k=l}^m(-1)^k\binom kl \mu_k\hspace{-2cm}\\
            &\text{and} & & Q\succeq0.     
        \end{aligned}
    \right. &&
\end{flalign*}
\reqnomode
Note that the constraint $y\ge1+\mu_n$ implies the constraint $y\ge\mu_n$.

\subsection{Dual program for the semidefinite restrictions}

The derivation is analogous to that of the previous section.
We fix $m\ge n$ and we recall below the expression of \refprog{lowerSDPn}:
\leqnomode
\begin{flalign*}
   \label{prog:lowerSDPn_recall}
    \tag*{$(\text{SDP}^{m,\leq}_n)$}
    \hspace{2.8cm}\left\{
        \begin{aligned}
            & & & \Sup{\substack{Q\in\text{Sym}_{m+1}\\\bm{F}\in\R^{m+1}}} F_n  \\
            & \text{subject to} & & \sum_{k=0}^{m} F_k = 1  \\
            & \text{and} & & \forall k\le m, \; F_k \geq 0 \\
            & \text{and} & & \forall l \in \llbracket 1,m \rrbracket, \quad\sum_{i+j=2l-1}Q_{ij} = 0\\
            & \text{and} & & \forall l\le m,\sum_{i+j=2l}Q_{ij} = \frac{(-1)^l}{l!}\sum_{k=l}^{m} (-1)^k  \binom kl F_k\\
            & \text{and} & & Q \succeq 0.
        \end{aligned}
    \right. &&
\end{flalign*}
\reqnomode
To put \refprog{lowerSDPn_recall} in the standard form \refprog{SDP_std} we set $N=2\times(m+1)$ and $M=1+m+(m+1)$. For all $r\in\N^*$ and all $i,j\in\llbracket 1,r \rrbracket$, recall that $E^{(i,j)}_r$ denotes the $r\times r$ matrix whose $(i,j)$ entry is $1$ and all other entries are $0$. We set
\begin{equation}
    \begin{aligned}
    &X=\text{Diag}_{k=0,\dots,m}(F_k)\oplus Q\in\text{Sym}_N,\\
    &C=E^{(n,n)}_N=E^{(n,n)}_{m+1}\oplus\mymathbb 0_{m+1}\in\text{Sym}_N,\\
    &\bm b=(1,0,0,\dots,0)\in\R^M,\\
    &{B'}^{(0)}=\mymathbb 1_{m+1}\oplus\mymathbb 0_{m+1}\in\text{Sym}_N,\\
    &\forall l\in\llbracket 1,m \rrbracket,\quad {B'}^{(l)}=\mymathbb 0_{m+1}\oplus\left(\sum_{i+j=2l-1}E^{(i,j)}_{m+1}\right),\\
    &\forall l\in\llbracket 0,m \rrbracket,\quad B^{(l)}=\text{Diag}_{k=0,\dots,m}\left(-\frac{(-1)^{k+l}}{l!}\binom kl\right)\oplus\left(\sum_{i+j=2l}E^{(i,j)}_{m+1}\right),
    \end{aligned}
\end{equation}
with the convention $\binom kl=0$ when $l>k$. The matrix ${B'}^{(0)}$ corresponds to the constraint $\sum_{k=0}^mF_k=1$ and we denote the corresponding dual variable $y\in\R$. Similarly,  the matrices ${B'}^{(l)}$ correspond to the $m$ constraints $\sum_{i+j=2l-1}Q_{ij} = 0$, and we denote the corresponding dual variables $\nu_l'\in\R$. Finally, the matrices $B^{(l)}$ correspond to the $m+1$ constraints $\sum_{i+j=2l}Q_{ij} = \frac{(-1)^l}{l!}\sum_{k=l}^{m} (-1)^k  \binom kl F_k$, and we denote the corresponding dual variables $\nu_l\in\R$.

The standard form \refprog{DSDP_std} of the dual program \refprog{lowerDSDPn} thus reads:
\leqnomode
\begin{flalign*}
   \label{prog:lowerDSDPn_std}
    \tag*{(D-SDP$_n^{m,\le}$)}\hspace{3.2cm}\left\{
        \begin{aligned}
            & & & \Inf{\substack{\bm\nu,\bm\nu'\in\R^{m+1}\times\R^m\\y\in\R}} y\\
            & \text{subject to} & & \text{Diag}_{k=0,\dots,m}\left[y-\sum_{l=0}^k\nu_l\frac{(-1)^{k+l}}{l!}\binom kl\right]\\
            & & &\quad\oplus\left(\sum_{l=0}^m\sum_{i+j=2l}\nu_lE^{(i,j)}_{m+1}+\sum_{l=1}^m\sum_{i+j=2l-1}\nu_l'E^{(i,j)}_{m+1}\right) \succeq E^{(n,n)}_{m+1}\oplus\mymathbb 0_{m+1}.\hspace{-4cm}
        \end{aligned}
    \right. &&
\end{flalign*}
\reqnomode
Due to the block-diagonal structure of the matrices involved, the positive semidefinite constraint above is equivalent to the following constraints:
\begin{equation}
   \left\{ \begin{aligned}
    &y\ge1+\sum_{l=0}^n\nu_l\frac{(-1)^{n+l}}{l!}\binom nl,\\
    &\forall k\in\llbracket0,m\rrbracket\setminus\{n\},\quad y\ge\sum_{l=0}^k\nu_l\frac{(-1)^{k+l}}{l!}\binom kl,\\
    &\left(\sum_{l=0}^m\sum_{i+j=2l}\nu_lE^{(i,j)}_{m+1}+\sum_{l=1}^m\sum_{i+j=2l-1}\nu_l'E^{(i,j)}_{m+1}\right)\succeq0.
    \end{aligned}\right.
\end{equation}
Let us define $A=(A_{ij})_{0\le i,j\le m}$ by
\begin{equation}
    A:=\sum_{l=0}^m\sum_{i+j=2l}\nu_lE^{(i,j)}_{m+1}+\sum_{l=1}^m\sum_{i+j=2l-1}\nu_l'E^{(i,j)}_{m+1},
\end{equation}
or equivalently
\begin{equation}\label{eq:Atonu}
    A_{ij}=\begin{cases}\nu_l&\text{when }i+j=2l,\\\nu_l'&\text{when }i+j=2l-1,\end{cases}
\end{equation}
for all $i,j\in\llbracket0,m\rrbracket$.
For $k\in\llbracket0,m\rrbracket$, we also define
\begin{equation}
    \mu_k:=\sum_{l=0}^k\nu_l\frac{(-1)^{k+l}}{l!}\binom kl\in\R.
\end{equation}
By Lemma~\ref{lemmapp:changeofbasis}, the following conditions are equivalent:
\begin{enumerate}[label=(\roman*)]
\item $\forall k\in\llbracket0,m\rrbracket,\quad\mu_k=\sum_{l=0}^k\nu_l\frac{(-1)^{k+l}}{l!}\binom kl$,
\item $\forall l\in\llbracket0,m\rrbracket,\quad\nu_l=\sum_{k=0}^l\mu_k\binom lkl!$.
\end{enumerate}
With~\refeq{Atonu} we thus have
\begin{equation}
    A_{ij}=\sum_{k=0}^l\mu_k\binom lkl!\quad\text{when }i+j=2l,
\end{equation}
and we obtain the following expression for \refprog{lowerDSDPn}:
\leqnomode
\begin{flalign*}
    \tag*{(D-SDP$_n^{m,\le}$)}\hspace{2.8cm}\left\{
        \begin{aligned}
            & & & \Inf{\substack{A\in\text{Sym}_{m+1}\\y,\bm\mu\in\R\times\R^{m+1}}} y\\
            & \text{subject to} & & y \geq 1 + \mu_n\\
            & \text{and} & &\forall k\in\llbracket0,m\rrbracket\setminus\{n\},\quad y \geq \mu_k\\
            & \text{and} & & \forall l \le m,\forall i+j=2l,\quad A_{ij}=\sum_{k=0}^l\mu_k\binom lkl!\\
            & \text{and} & & A \succeq 0.
        \end{aligned}
    \right. &&
\end{flalign*}
\reqnomode
Note that the constraint $y\ge1+\mu_n$ implies the constraint $y\ge\mu_n$.


\newpage
\section{Proof of Lemma~\ref{lem:newlowerDSDPn}}
\label{sec:app_newlowerDSDPn}

In this section, we prove the following result:

\begin{lemmapp}
For all $m\ge n$, the program \refprog{lowerDSDPn} is equivalent to the following program:
\leqnomode
\begin{flalign*}
    \tag*{(D-SDP$_n^{m,\le}$)}\hspace{2.8cm}\left\{
        \begin{aligned}
            & & & 
            \Inf{\substack{y \in \R \\ \bm\mu \in\mathcal S'(\N)}} y  \\
            & \text{\upshape subject to} & & y \geq 1 + \mu_n\\
            & \text{\upshape and} & & \forall k \neq n \in \N,\; y \ge\mu_k\\
            & \text{\upshape and} & & \forall g \in \mathcal R_{m,+}(\R_+), \; \langle f_{\bm\mu},g \rangle \geq  0,
        \end{aligned}
    \right. &&
\end{flalign*}
\reqnomode
where $f_{\bm\mu}=\sum_k\mu_k\mathcal L_k$.
\end{lemmapp}

\begin{proof}
We first obtain a reformulation of \refprog{lowerSDPn} and we derive its dual program. This reformulation is obtained using Stieltjes characterisation of non-negative polynomials over $\R_+$ rather than Lemma~\ref{lem:pospolyR+}:

\begin{lemma}[\cite{reed1975ii}]\label{lemmapp:Stieltjes1}
Let $m\in\mathbb N$ and let $P$ be a univariate polynomial of degree $m$. Let $a_1=\left\lfloor\frac m2\right\rfloor$ and $a_2=\left\lfloor\frac{m-1}2\right\rfloor$. For all $q\in\mathbb N$, let $X_q=(1,x,\dots,x^q)$ be the vector of univariate monomials up to degree $q$. Then, $P$ is non-negative over $\R_+$ if and only if there exist sum of squares polynomials $A_1$ and $A_2$ of degree $2a_1$ and $2a_2$, respectively, such that $P(x)=A_1(x)+xA_2(x)$ for all $x\in\mathbb R_+$ , or equivalently, if and only if there exist real positive semidefinite matrices $A_1$ and $A_2$ of size $(a_1+1)\times(a_1+1)$ and $(a_2+1)\times(a_2+1)$, respectively, such that for all $x\in\mathbb R_+$,
\begin{equation}
    P(x)=X_{a_1}^TA_1X_{a_1}+xX^T_{a_2}A_2X_{a_2}.
\end{equation}
\end{lemma}

\noindent By Lemma~\ref{lemmapp:Stieltjes1}, the program \refprog{lowerSDPn}, obtained by imposing $F_k=0$ for $k>m$ in \refprog{LP}, is equivalent to the following program:
\leqnomode
\begin{flalign*}
   \label{prog:lowerSDPn_Stieltjes}
    \tag*{$(\text{SDP}^{m,\leq}_n)$}
    \hspace{2.8cm}\left\{
        \begin{aligned}
            & & & \Sup{\substack{A_1,A_2\in\text{Sym}_{a_1+1}\times\text{Sym}_{a_2+1}\\\bm{F}\in\R^{m+1}}} F_n  \\
            & \text{subject to} & & \sum_{k=0}^{m} F_k = 1  \\
            & \text{and} & & \forall k\le m, \; F_k \geq 0 \\
            & \text{and} & & \forall l\le m,\frac{(-1)^l}{l!}\sum_{k=l}^{m} (-1)^k\binom klF_k=\sum_{\substack{i+j=l\\0\le i,j\le a_1}}(A_1)_{ij}+\sum_{\substack{i+j=l-1\\0\le i,j\le a_2}}(A_2)_{ij}\hspace{-1cm} \\
            & \text{and} & & A_1 \succeq 0\\
            & \text{and} & & A_2 \succeq 0.
        \end{aligned}
    \right. &&
\end{flalign*}
\reqnomode
To put \refprog{lowerSDPn_Stieltjes} in the standard form \refprog{SDP_std} we set $N=(m+1)+(a_1+1)+(a_2+1)$ and $M=1+(m+1)$. For all $r\in\N^*$ and all $i,j\in\llbracket 1,r \rrbracket$, recall that $E^{(i,j)}_r$ denotes the $r\times r$ matrix whose $(i,j)$ entry is $1$ and all other entries are $0$. We set
\begin{equation}
    \begin{aligned}
    &X=\text{Diag}_{k=0,\dots,m}(F_k)\oplus A_1\oplus A_2\in\text{Sym}_N,\\
    &C=E^{(n,n)}_N=E^{(n,n)}_{m+1}\oplus\mymathbb 0_{a_1+1}\oplus\mymathbb 0_{a_2+1}\in\text{Sym}_N,\\
    &\bm b=(1,0,0,\dots,0)\in\R^M,\\
    &{B}^{(-1)}=\mymathbb 1_{m+1}\oplus\mymathbb 0_{a_1+1}\oplus\mymathbb 0_{a_2+1}\in\text{Sym}_N,\\
    &\forall l\in\llbracket 0,m \rrbracket,\quad B^{(l)}=\text{Diag}_{k=0,\dots,m}\left(-\frac{(-1)^{k+l}}{l!}\binom kl\right)\oplus\left(\sum_{\substack{i+j=l\\0\le i,j\le a_1}}E^{(i,j)}_{a_1+1}\right)\oplus\left(\sum_{\substack{i+j=l-1\\0\le i,j\le a_2}}E^{(i,j)}_{a_2+1}\right),
    \end{aligned}
\end{equation}
with the convention $\binom kl=0$ when $l>k$. The matrix ${B}^{(-1)}$ corresponds to the constraint $\sum_{k=0}^mF_k=1$ and we denote the corresponding dual variable $y\in\R$. Similarly, the matrices $B^{(l)}$ correspond to the $m+1$ other linear constraints, and we denote the corresponding dual variable $\bm\nu\in\R^{m+1}$.
The standard form \refprog{DSDP_std} of the dual program \refprog{lowerDSDPn} thus reads:
\leqnomode
\begin{flalign*}
   \label{prog:lowerDSDPn_Stieltjes_std}
    \tag*{(D-SDP$_n^{m,\le}$)}\hspace{2.8cm}\left\{
        \begin{aligned}
            & & & \Inf{\substack{\bm\nu\in\R^{m+1}\\y\in\R}} y\\
            & \text{subject to} & & \text{Diag}_{k=0,\dots,m}\left[y-\sum_{l=0}^k\nu_l\frac{(-1)^{k+l}}{l!}\binom kl\right]\\
            & & &\quad\oplus\left(\sum_{l=0}^m\sum_{\substack{i+j=l\\0\le i,j\le a_1}}\nu_lE^{(i,j)}_{a_1+1}\!\right)\\
            & & &\quad\oplus\left(\sum_{l=1}^m\sum_{\substack{i+j=l-1\\0\le i,j\le a_2}}\nu_lE^{(i,j)}_{a_2+1}\!\right)\succeq E^{(n,n)}_{m+1}\oplus\mymathbb 0_{a_1}\oplus\mymathbb 0_{a_2}.\hspace{-4cm}
        \end{aligned}
    \right. &&
\end{flalign*}
\reqnomode
Due to the block-diagonal structure of the matrices involved, the positive semidefinite constraint above is equivalent to the following constraints:
\begin{equation}\label{eq:constraintsStieltjes}
   \left\{ \begin{aligned}
    &y\ge1+\sum_{l=0}^n\nu_l\frac{(-1)^{n+l}}{l!}\binom nl,\\
    &\forall k\in\llbracket0,m\rrbracket\setminus\{n\},\quad y\ge\sum_{l=0}^k\nu_l\frac{(-1)^{k+l}}{l!}\binom kl,\\
    &\sum_{l=0}^m\sum_{\substack{i+j=l\\0\le i,j\le a_1}}\nu_lE^{(i,j)}_{a_1+1}\succeq0\\
    &\sum_{l=1}^m\sum_{\substack{i+j=l-1\\0\le i,j\le a_2}}\nu_lE^{(i,j)}_{a_2+1}\succeq0.
    \end{aligned}\right.
\end{equation}
For $k\in\llbracket0,m\rrbracket$, we define
\begin{equation}
    \mu_k:=\sum_{l=0}^k\nu_l\frac{(-1)^{k+l}}{l!}\binom kl\in\R.
\end{equation}
By Lemma~\ref{lemmapp:changeofbasis}, the following conditions are equivalent:
\begin{enumerate}[label=(\roman*)]
\item $\forall k\in\llbracket0,m\rrbracket,\quad\mu_k=\sum_{l=0}^k\nu_l\frac{(-1)^{k+l}}{l!}\binom kl$,
\item $\forall l\in\llbracket0,m\rrbracket,\quad\nu_l=\sum_{k=0}^l\mu_k\binom lkl!$.
\end{enumerate}
We thus have
\begin{equation}\label{eq:mutonu2}
    \forall l\in\llbracket0,m\rrbracket,\quad\nu_l=\sum_{k=0}^l\mu_k\binom lkl!
\end{equation}
Let us introduce the moment matrices $M_\nu\in\text{Sym}_{a_1+1}$ and $M^{(1)}_\nu\in\text{Sym}_{a_2+1}$:
\begin{align}
    \forall i,j\in\llbracket0,a_1\rrbracket,\; &(M_\nu)_{i,j}:=\nu_{i+j}, \label{eq:Mnuij} \\
    \forall i,j\in\llbracket0,a_2\rrbracket,\; &(M_\nu^{(1)})_{ij}:=\nu_{i+j+1}.
\end{align}
The constraints~\eqref{eq:constraintsStieltjes} are equivalent to:
\begin{equation}\label{eq:constraintsStieltjes2}
   \left\{ \begin{aligned}
    &y\ge1+\mu_n,\\
    &\forall k\in\llbracket0,m\rrbracket\setminus\{n\},\quad y\ge\mu_l,\\
    &M_\nu\succeq0\\
    &M_\nu^{(1)}\succeq0.
    \end{aligned}\right.
\end{equation}
We complete the vector $\bm\mu$ with zeros to obtain an element of $\mathcal S'(\N)$. We have $f_{\bm\mu}=\sum_k\mu_k\mathcal L_k\in\mathcal S'(\R_+)$. We prove the following result, analogous to Lemma~\ref{lemmapp:momentmatrix}:

\begin{lemma}\label{lemmapp:momentmatrix2}
The following propositions are equivalent:
\begin{enumerate}[label=(\roman*)]
\item $\forall g\in\mathcal R_{m,+}(\R_+),\;\braket{f_{\bm\mu},g}\ge0$,
\item $M_\nu\succeq0$ and $M_\nu^{(1)}\succeq0$.
\end{enumerate}
\end{lemma}

\begin{proof}
The proof is similar to that of Lemma~\ref{lemmapp:momentmatrix}:
for all $k\in\mathbb N$ we have
\begin{equation}\label{eq:mukf}
    \mu_k=\int_{\R_+}\mathcal L_k(x)f_{\bm\mu}(x)dx.
\end{equation}
Thus,
\begin{equation}\label{eq:Mnupos}
    \begin{aligned}
          M_\nu\succeq0&\Leftrightarrow\forall Y\in\R^{a_1+1},\;Y^TM_\nu Y\ge0\\
                           &\Leftrightarrow\forall Y\in\R^{a_1+1},\;\sum_{i,j=0}^{a_1}y_iy_j(M_\nu)_{ij}\ge0\\
                           &\Leftrightarrow\forall Y\in\R^{a_1+1},\;\sum_{l=0}^{2a_1}\sum_{i+j=l}y_iy_j\nu_l\ge0\\
                           &\Leftrightarrow\forall Y\in\R^{a_1+1},\;\sum_{l=0}^{2a_1}\sum_{i+j=l}y_iy_j\sum_{k=0}^l\mu_k\binom lkl!\ge0\\
                           &\Leftrightarrow\forall Y\in\R^{a_1+1},\;\int_{\R_+}\sum_{l=0}^{2a_1}\sum_{i+j=l}y_iy_j\sum_{k=0}^l\binom lkl!\mathcal L_k(x)f_{\bm\mu}(x)dx\ge0\\
                           &\Leftrightarrow\forall Y\in\R^{a_1+1},\;\int_{\R_+}\sum_{l=0}^{2a_1}\sum_{i+j=l}y_iy_j\sum_{k=0}^l(-1)^k\binom lkl!L_k(x)e^{-\frac x2}f_{\bm\mu}(x)dx\ge0\\
                           &\Leftrightarrow\forall Y\in\R^{a_1+1},\;\int_{\R_+}e^{-\frac x2}\sum_{l=0}^{2a_1}x^l\sum_{i+j=l}y_iy_jf_{\bm\mu}(x)dx\ge0\\
                           &\Leftrightarrow\forall Y\in\R^{a_1+1},\;\int_{\R_+}e^{-\frac x2}\left(\sum_{k=0}^{a_1}y_kx^k\right)^2f_{\bm\mu}(x)dx\ge0\\
                           &\Leftrightarrow\forall Y\in\R^{a_1+1},\;\left\langle f_{\bm\mu},x\mapsto e^{-\frac x2}\left(\sum_{k=0}^{a_1}y_kx^k\right)^2\right\rangle\ge0,
    \end{aligned}
\end{equation}
where we used~\refeq{Mnuij} in the third line,~\refeq{mutonu2} in the fourth line,~\refeq{mukf} in the fifth line and~\refeq{Lagtox} in the seventh line. Similarly, 
\begin{equation}\label{eq:Mnu1pos}
     M_\nu^{(1)}\succeq0\Leftrightarrow\forall Y\in\R^{a_2+1},\;\left\langle f_{\bm\mu},x\mapsto e^{-\frac x2}x\left(\sum_{k=0}^{a_2}y_kx^k\right)^2\right\rangle\ge0.
\end{equation}
Combining~\refeq{Mnupos} and~\refeq{Mnu1pos} with Lemma~\ref{lemmapp:Stieltjes1} we obtain
\begin{equation}
     M_\nu\succeq0\text{ and }M_\nu^{(1)}\succeq0\Leftrightarrow\forall g\in\mathcal R_{m,+}(\R_+),\;\braket{f_{\bm\mu},g}\ge0,
\end{equation}
by linearity.

\end{proof} 

\noindent Combining Lemma~\ref{lemmapp:momentmatrix2} with the constraints~\eqref{eq:constraintsStieltjes2} finally yields:

\leqnomode
\begin{flalign*}
    \tag*{(D-SDP$_n^{m,\le}$)}
    \hspace{2.8cm}\left\{
        \begin{aligned}
            & & & \Inf{\substack{y \in \R \\ \bm\mu \in\mathcal S'(\N)}} y\\
            & \text{subject to} & & \forall k \neq n \in \N,\; y \ge\mu_k \\
            & \text{and} & & y \geq 1 + \mu_n \\
            & \text{and} & & \forall g \in \mathcal R_{m,+}(\R_+), \; \langle f_{\bm\mu},g \rangle \geq  0.
        \end{aligned}
    \right. &&
\end{flalign*}
\reqnomode

\end{proof}

\newpage
\section{Proof of Lemma~\ref{lem:feasible}}
\label{sec:app_feasible}

We recall the definition of $\bm F^n=(F_k^n)_{k\in\N}\in\R^{\N}$:
\begin{itemize}
    \item if $n$ is even:
\begin{align}
        F_k^n:=&\begin{cases}\frac1{2^n}\binom k{\frac k2}\binom{n-k}{\frac{n-k}2}&\text{when }k\le n, k\text{ even},\\0&\text{otherwise},\end{cases} \label{eq:Fkfornevenapp}\\
\intertext{\item if $n$ is odd:}
        F_k^n:=&\begin{cases} \frac1{2^n}\frac{\binom n{\floor{\frac n2}}\binom{\floor{\frac n2}}{\floor{\frac k2}}^2 }{\binom nk},&\text{when }k\le n,\\0&\text{otherwise}.\label{eq:Fkfornoddapp} \end{cases}
\end{align}
\end{itemize}
In both cases,
\begin{equation}
        F_n^n =\frac1{2^n}\binom n{\floor{\frac n2}}.
\end{equation}
We extrapolated these analytical expressions from numerical values. Running \refprog{lowerSDPn} for several values of $n$ and $m$ allowed us to deduce these sequences (we acknowledge here the great help from \href{https://oeis.org}{oeis.org}). 
 
We start by showing two results, corresponding to $n$ even and $n$ odd, respectively, where we make use of Zeilberger's algorithm, a powerful algorithm for proving binomial identities \cite{zeilberger1991method}.  Given a holonomic function, this algorithm outputs a recurrence relation that it satisfies, thus reducing the proof of identity between binomial expressions to the verification that the initialisation is correct. A Mathematica notebook is available for the implementation of Zeilberger's algorithm  \cite{codes}.

\begin{lemma}
\label{lemma:QdecFnkeven}
For $n \in \N$ even:
\begin{equation}
    \sum_{k=0}^{n} (-1)^k F_k^n L_k(x) = \sum_{l=0}^n x^l \sum_{i+j=2n-2l} p_{n-i} p_{n-j},
    \label{eq:polyeqeven}
\end{equation}
where:
\begin{align}\label{eq:coefpnk}
    p_{n-k}:=&\begin{cases} \sqrt{\frac{1}{2^n n!} \binom{n}{\frac{n}{2}} },&\text{when }k=0,\\
    (-1)^{\frac k2} 2^{\frac k2} (\frac k2)!  \binom{\frac{n}{2}}{\frac k2}^2 p_n,&\text{when }0<k\le n,\;k \text{ even },\\
    0,&\text{otherwise}.
\end{cases}
\end{align}
\end{lemma}

\noindent The coefficients $p_{n-k}$ (and $q_{n-k}$ later on) were found by hand when looking for an analytical sum of squares decomposition.

\begin{proof}
To prove the polynomial equality~\eqref{eq:polyeqeven}, we start by equating the coefficients in $x^{l}$ for $l \in \llbracket 0,n \rrbracket$ which gives:
\begin{equation}    \label{eq:claim_SOSeven}
    \begin{aligned}
        \frac{(-1)^l}{l!}\sum_{k=l}^n (-1)^k  \binom{k}{l}F_k^n 
        &= \sum_{i+j=2n-2l} p_{n-i} p_{n-j}\\
        &= \sum_{i=0}^{2n-2l} p_{n-i} p_{n-(2n-2l-i)}\\
        &= \sum_{i=0}^{n-l} p_{n-2i} p_{n-(2n-2l-2i)},
    \end{aligned}
\end{equation}
where we used $p_{n-i}=0$ for $n-i<0$ in the second line. 
These are equalities between holonomic functions of parameters $n$ and $l$ that are trivial for $l>n$.

\begin{itemize}
\item For $l\le n$ even, because $F^n_k=0$ for $k$ odd, \refeq{claim_SOSeven} becomes:
\begin{equation}
    \sum_{k=\frac l2}^{\frac n2} \frac1{l!} \binom{2k}{l} F_{2k}^n = \sum_{i=0}^{n-l} p_{n-2i} p_{n-(2n-2l-2i)},
\end{equation}
that is, taking into account the parity of $l=2s$ and $n=2t$,
\begin{equation}
    \sum_{k=s}^{t} \frac1{(2s)!} \binom{2k}{2s} F_{2k}^{2t} = \sum_{i=0}^{2t-2s} p_{2t-2i} p_{2t-(4t-4s-2i)}.
    \label{eq:Zeil_neven_leven_pre}
\end{equation}
Inserting the expressions from~\refeq{Fkfornevenapp} and~\refeq{coefpnk}, we thus have to check the identity:
\begin{equation}\label{eq:Zeil_neven_leven}
    \sum_{k=s}^t\frac1{2^{2t}(2s)!}\binom{2k}{2s}\binom{2k}{k}\binom{2t-2k}{t-k}= \sum_{i=0}^{2t-2s}\frac{i!(2t-2s-i)!}{2^{2s}(2t)!}\binom{2t}t\binom ti^2\binom t{2t-2s-i}^2,
\end{equation}
for all $t\in\N$ and all $s\le t$ (with the convention $\binom kj=0$ for $j>k$).
We ran Zeilberger's algorithm to show that the right-hand side and the left-hand side of~\refeq{Zeil_neven_leven_pre} both satisfy the following recurrence relation, for all $s,t\in\N$:
\begin{equation}
    2(t+1)^2 S(s,t)+(-2s^2-4t^2+4st+5s-11t-8) S(s,t+1)+(s-t-2) (2s-2t-3) S(s,t+2) =0.
\end{equation}
It remains to check that the initialisation is correct. For all $s\in\N$, this recurrence relation in $t$ is of order $2$. Since the identities in~\refeq{Zeil_neven_leven} are trivial when $l>n$, i.e., $s>t$, we thus only need to check~\refeq{Zeil_neven_leven} for $(s,t)=(0,0)$, $(s,t)=(0,1)$, and $(s,t)=(1,1)$, which is straightforward: we obtain the values $1$, $1$ and $\frac14$ respectively, for both sides of~\refeq{Zeil_neven_leven}.
\item For $l\le n$ odd, \refeq{claim_SOSeven} becomes:
\begin{equation}
    - \sum_{k=\frac{l+1}{2}}^{\frac n2} \frac1{l!} \binom{2k}{l} F_{2k}^n = \sum_{i=0}^{n-l} p_{n-2i} p_{n-(2n-2l-2i)},
\end{equation}
that is, taking into account the parity of $l=2s+1$ and $n=2t$:
\begin{equation}
    - \sum_{k=s+1}^{t} \frac{1}{(2s+1)!} \binom{2k}{2s+1} F_{2k}^{2t} = \sum_{i=0}^{2t-2s-1} p_{2t-2i} p_{2t-(4t-4s-2i-2)}.
    \label{eq:Zeil_neven_lodd_pre}
\end{equation}
Inserting the expressions from~\refeq{Fkfornevenapp} and~\refeq{coefpnk}, we thus have to check the identity:
\begin{equation}\label{eq:Zeil_neven_lodd}
    \begin{aligned}
        \sum_{k=s+1}^{t} \frac1{2^{2t}(2s+1)!}&\binom{2k}{2s+1}\binom{2k}k\binom{2t-2k}{t-k} \\
        &= \sum_{i=0}^{2t-2s-1}\frac{i!(2t-2s-i-1)!}{2^{2s+1}(2t)!}\binom{2t}t\binom ti^2\binom t{2t-2s-i-1}^2,
    \end{aligned}
\end{equation}
for all $t\in\N$ and all $s\le t$ (with the convention $\binom kj=0$ for $j>k$).
Likewise, we ran Zeilberger's algorithm to show that the right-hand side and the left-hand side of~\refeq{Zeil_neven_lodd_pre} both satisfy the following recurrence relation, for all $s,t\in\N$:
\begin{equation}
    2(t+1)^2 S(s,t) + (-2s^2-4t^2+4st+3s-9t-6) S(s,t+1)+(s-t-1)(2s-2t-3) S(s,t+2) =0.
\end{equation}
It remains to check that the initialisation is correct. For all $s\in\N$, this recurrence relation in $t$ is of order $2$. Since the identities in~\refeq{Zeil_neven_lodd} are trivial when $l>n$, i.e., $s\ge t$, we thus only need to check~\refeq{Zeil_neven_lodd} for $(s,t)=(0,1)$, which is straightforward: we obtain the value $1$ for both sides of~\refeq{Zeil_neven_lodd}.
\end{itemize}

\end{proof}

\begin{lemma}
\label{lemma:QdecFnkodd}
For $n \in \N$ odd:
\begin{equation}
    \sum_{k=0}^{n} (-1)^k F_k^n L_k(x) = \sum_{l=0}^n x^l \sum_{i+j=2l} q_i q_j
    \label{eq:polyeqodd}
\end{equation}
\end{lemma}
where:
\begin{align}
    q_{n-k}:=&\begin{cases} \sqrt{\frac{1}{2^n n!} \binom{n}{\floor{\frac{n}{2}}} },&\text{when }k=0,\\
    (-1)^{\frac k2}2^{\frac k2}(\frac k2)! \frac{n+1}{n-k+1}\binom{\floor{\frac{n}{2}}}{\frac k2}^2q_n,&\text{when }0<k< n,\;k \text{ even},\\
    0.&\text{otherwise}.
\label{eq:coefqnk}
\end{cases}
\end{align}

\begin{proof}

Unlike the case where $n$ is even, $\bm F^n$ is non-zero for $k\le n $ odd, and the expression of $F^n_k$ depends on the parity of $k$. Thus we cannot use directly~\refeq{Fkfornoddapp} in Zeilberger's algorithm as we did in the previous lemma, hence the development below in order to obtain expressions that the algorithm can take as inputs. We fix $n$ odd and $l\le n$.

We start by equating coefficents in $x^l$ in~\refeq{polyeqodd}:
\begin{equation}
    \begin{aligned}
        \frac{(-1)^l}{l!}\sum_{k=l}^n (-1)^k  \binom{k}{l}F_k^n &= \sum_{i+j=2n-2l} q_{n-i} q_{n-j} \\
        &= \sum_{i=0}^{n-l} q_{n-2i} q_{n-(2n-2l-2i)}.
    \end{aligned}
\label{eq:claim_SOSodd}
\end{equation}
\begin{itemize}
    \item For $l$ even, writing $l=2s$ and $n=2t+1$, the left-hand side of~\refeq{claim_SOSodd} becomes:
    \begin{equation}
        \begin{aligned}
            \frac{1}{(2s)!}\sum_{k=2s}^{2t+1} & (-1)^k  \binom{k}{2s}F_k^{2t+1} \\ &= \frac{1}{(2s)!}\sum_{k=0}^{2t+1-2s} (-1)^k  \binom{k+2s}{2s}F_{k+2s}^{2t+1}\\
            &= \frac{1}{(2s)!}\sum_{k=0}^{t-s} \left( \binom{2k+2s}{2s}F_{2k+2s}^{2t+1} - \binom{2k+2s+1}{2s}F_{2k+2s+1}^{2t+1} \right)\\
            &= \frac{\binom{2t+1}{t}}{2^{2t+1}(2s)!} \sum_{k=0}^{t-s} \left(  \frac{\binom{2k+2s}{2s} \binom{t}{k+s}^2}{\binom{2t+1}{2k+2s}} - \frac{\binom{2k+2s+1}{2s} \binom{t}{k+s}^2}{\binom{2t+1}{2k+2s+1}} \right)\\
            &= q_{2t+1}^2 \frac{(2t+1)!}{(2s)!} \sum_{k=0}^{t-s}  \frac{\binom{t}{k+s}^2 \binom{2k+2s}{2s}}{\binom{2t+1}{2k+2s}} \left( 1 - \frac{(2k+2s+1)^2}{(2k+1)(2t-2s-2k+1)} \right).
        \end{aligned}
    \end{equation}
    With~\refeq{coefqnk} and~\refeq{claim_SOSodd}, we thus have to check the identity:
    \begin{equation}
        \begin{aligned}
          \frac{(2t+1)!}{(2s)!} &\sum_{k=0}^{t-s}  \frac{\binom{t}{k+s}^2 \binom{2k+2s}{2s}}{\binom{2t+1}{2k+2s}} \left( 1 - \frac{(2k+2s+1)^2}{(2k+1)(2t-2s-2k+1)} \right) \\ 
          &=  \sum_{i=0}^{2t+1-2s} \frac{q_{2t+1-2i} q_{2t+1-(4t+2-4s-2i)}}{q_{2t+1}^2} \\
          &= -\sum_{i=0}^{2t+1-2s}2^{2t-2s+1}\frac{(2t+2)^2 i! (2t+1-2s-i)!}{(2t-2i+2)(4s+2i-2t)}\binom ti^2\binom t{2t-2s-i+1}^2,
        \end{aligned}
        \label{eq:Zeil_nodd_leven}
    \end{equation}
    for all $t\in\N$ and all $s\le t$ (with the convention $\binom kj=0$ for $j>k$).
    Zeilberger's algorithm certifies that the right-hand side and the left-hand side of~\refeq{Zeil_nodd_leven} both satisfy for all $s,t\in\N$:
    \begin{equation}
        \begin{aligned}
          -32 &(t+2)^3 (t+1)^2 (t+3) S(s,t) \\ 
          &+ 4(t+3)(t+2)(2s^2+4t^2-4st-7s+15t+14) S(s,t+1) \\
          & + (-2s+2t+5)(s-t-2)S(s,t+2) =0.
        \end{aligned}
    \end{equation}
    It remains to check that the initialisation is correct. For all $s\in\N$, this recurrence relation in $t$ is of order $2$. Since the identities in~\refeq{Zeil_nodd_leven} are trivial when $l>n$, i.e., $s>t$, we thus only need to check~\refeq{Zeil_nodd_leven} for $(s,t)=(0,0)$, $(s,t)=(0,1)$, and $(s,t)=(1,1)$, which is straightforward: we obtain the values $0$, $0$ and $-8$ respectively, for both sides of~\refeq{Zeil_nodd_leven}.
    \item For $l$ odd, writing $l=2s+1$ and $n=2t+1$, the left-hand side of~\refeq{claim_SOSodd} becomes:
    \begin{equation}
        \begin{aligned}
            \frac{-1}{(2s+1)!}\sum_{k=2s+1}^{2t+1} & (-1)^k  \binom{k}{2s+1}F_k^{2t+1} \\ 
            &=  \frac{-1}{(2s+1)!} \sum_{k=0}^{2t-2s} (-1)^{k+2s+1} \binom{k+2s+1}{2s+1} F_{k+2s+1}^{2t+1} \\
            &= \frac 1{(2s+1)!} \sum_{k=0}^{t-s} \left( - \binom{2k+2s}{2s+1} F_{2k+2s}^{2t+1} + \binom{2k+2s+1}{2s+1} F_{2k+2s+1}^{2t+1} \right) \\
            &= q_{2t+1}^2 \frac{(2t+1)!}{(2s+1)!} \sum_{k=0}^{t-s} \left( - \frac{\binom{2k+2s}{2s+1} \binom{t}{k+s}^2 }{ \binom{2t+1}{2k+2s} } + \frac{\binom{2k+2s+1}{2s+1} \binom{t}{k+s}^2 }{ \binom{2t+1}{2k+2s+1} } \right) \\
            &= q_{2t+1}^2 \frac{(2t+1)!}{(2s+1)!} \sum_{k=0}^{t-s} \frac{\binom{2k+2s}{2s+1} \binom{t}{k+s}^2 }{ \binom{2t+1}{2k+2s} } \left(-1 + \frac{(2k+2s+1)^2}{2k(2t-2k-2s+1)} \right),
        \end{aligned}
    \end{equation}
    where we used that $\binom{2k+2s}{2s+1}=0$ for $k=0$ in the third line. Note that when factorising we introduced an indeterminate form in the last line that Zeilberger's algorithm can resolve. This is necessary since the algorithm cannot deal with differences of binomial terms.  
    With~\refeq{coefqnk} and~\refeq{claim_SOSodd}, we thus have to check the identity:
    \begin{equation}\label{eq:Zeil_nodd_lodd}
        \begin{aligned}
          \frac{(2t+1)!}{(2s+1)!} & \sum_{k=0}^{t-s} \frac{\binom{2k+2s}{2s+1} \binom{t}{k+s}^2 }{ \binom{2t+1}{2k+2s} } \left(-1 + \frac{(2k+2s+1)^2}{2k(2t-2k-2s+1)} \right) \\ 
          &= \sum_{i=0}^{2t-2s} \frac{q_{2t+1-2i} q_{2t+1-(4t-4s-2i)} }{q_{2t+1}^2} \\
          &=\sum_{i=0}^{2t-2s}2^{2t-2s}\frac{(2t+2)^2 i! (2t-2s-i)!}{(2t-2i+2)(4s+2i-2t+2)}\binom ti^2\binom t{2t-2s-i}^2,
        \end{aligned}
    \end{equation}
    for all $t\in\N$ and all $s\le t$ (with the convention $\binom kj=0$ for $j>k$).
    Zeilberger's algorithm certifies that both the right-hand side and the left-hand side of~\refeq{Zeil_nodd_lodd} satisfy for all $s \le t$:
    \begin{equation}
        \begin{aligned}
          -32 &(t+2)^3 (t+1)^2 (t+3) S(s,t) \\ 
          &+ 4(t+3)(t+2)(2s^2+4t^2-4st-5s+13t+11) S(s,t+1) \\
          & + (-2s+2t+3)(s-t-2)S(s,t+2) =0.
        \end{aligned}
    \end{equation}
    It remains to check that the initialisation is correct. For all $s\in\N$, this recurrence relation in $t$ is of order $2$. Since the identities in~\refeq{Zeil_nodd_lodd} are trivial when $l>n$, i.e., $s>t$, we thus only need to check~\refeq{Zeil_nodd_lodd} for $(s,t)=(0,0)$, $(s,t)=(0,1)$, and $(s,t)=(1,1)$, which is straightforward: we obtain the values $1$, $16$ and $1$ respectively, for both sides of~\refeq{Zeil_nodd_lodd}.
\end{itemize}

\end{proof}

\noindent Having derived these identities, we now recall the lemma from the main text we wish to prove:
\begin{lemmapp}
For all $m\ge n$, $\bm F^n$ is a feasible solution of \refprog{lowerSDPn}. Moreover, it is optimal when $m=n$.
\end{lemmapp}
\begin{proof}
The proof has three parts. In the first we focus on $n$ even, and in the second on $n$ odd. 
In the last part, we exhibit a feasible solution of \refprog{lowerDSDPn} for $m=n$ with the same optimal value $F_n^n$. 

\paragraph{$n$ even, $m \ge n$:} we check that $\bm F^n$ defined in~\refeq{Fkfornevenapp} satisfies all the constraints of \refprog{lowerSDPn}.

\begin{itemize}
    \item For all $k \in \N$, $F_k^n \ge 0$.
    \item We have
    \begin{equation}\label{eq:sumFkneven}
        \begin{aligned}
            \sum_{k=0}^{\infty} F_k^n &=  \sum_{\substack{k=0 \\ \text{even}}}^{n} \frac1{2^{n}}\binom k{\frac k2}\binom{n-k}{\frac{n-k}2}\\
            & = \sum_{k=0}^{\frac n2} \frac1{2^{n}}\binom{2k}{k}\binom{n-2k}{\frac n2 -k} \\
            & = 1,
        \end{aligned}
    \end{equation}
    where the last equality follows from \cite[(3.90)]{gould1972combinatorial}.
    \item We have to show that $x\mapsto\sum_{k=0}^{n} (-1)^k F_k^n L_k(x^2)$ is a positive function on $\R$.  
    Due to Lemma \ref{lem:pospolyR}, we aim to find a sum of squares decomposition for this polynomial. Guided by numerical results, we look for a polynomial $P(x) \defeq \sum_{i=0}^{n} p_i x^i$ such that:
    \begin{equation}
        \begin{aligned}
          \sum_{k=0}^{n} (-1)^k F_k^n L_k(x^2) & = P^2(x^2) \\
          & = \left( \sum_{i=0}^n p_i x^{2i} \right)^2 \\
          & = \sum_{l=0}^n \left( \sum_{\substack{i+j=l \\ 0 \leq i,j \leq n}} p_i p_j \right) x^{2l},
        \end{aligned}
    \end{equation}
    and the sought coefficients are given by Lemma~\ref{lemma:QdecFnkeven}, which concludes the first part of the proof.
\end{itemize}

\paragraph{$n$ odd, $m \ge n$:} similarly, we check that $\bm F^n$ defined in~\refeq{Fkfornoddapp} satisfies all the constraints of \refprog{lowerSDPn}.

\begin{itemize}
    \item For all $k \in \N$, $F_k^n \ge 0$.
    \item Writing $n=2t+1$, from~\refeq{Fkfornoddapp} we have
    \begin{equation}
        \begin{split}
            & \forall s\le t, \; F^{2t+1}_{2s} = \frac1{2^{2t+2}}\left(1-\frac{s}{t+1}\right)\binom{2s}s\binom{2t+2-2s}{t+1-s}  \\
            & \forall s\in\llbracket1,t+1\rrbracket, \; F^{2t+1}_{2s-1} =   \frac1{2^{2t+2}}\frac{s}{t+1}\binom{2s}s\binom{2t+2-2s}{t+1-s}.
        \end{split}
    \end{equation}
    In particular, for all $s\le t$ we have
    \begin{equation}
       F^{2t+1}_{2s-1}+F^{2t+1}_{2s}=F^{2t+2}_{2s},
    \end{equation}
    where $F^{2t+2}_{2s}$ is defined in \refeq{Fkfornevenapp}. Since $F^{2t+2}_{2s+1}=0$ for all $s\le t$, and $F_k^n=0$ for all $k>n$, we have:
    \begin{equation}
        \begin{aligned}
            \sum_{k=0}^{\infty} F_k^n&=\sum_{k=0}^{\infty} F_k^{n+1} \\
            & = 1,
        \end{aligned}
    \end{equation}
    where we used~\refeq{sumFkneven}.
    \item We have to show that $x\mapsto\sum_{k=0}^{n} (-1)^k F_k^n L_k(x^2)$ is a positive function on $\R$.  
    Due to Lemma \ref{lem:pospolyR}, we aim to find a sum of squares decomposition for this polynomial. Guided by numerical results, we look for a polynomial $Q(x) \defeq \sum_{i=0}^{n} q_i x^i$ such that:
\begin{equation}
    \begin{aligned}
        \sum_{k=0}^{n} (-1)^k F_k^n L_k(x^2) & = Q^2(x^2) \\
            & = \left( \sum_{i=0}^n q_i x^{2i} \right)^2 \\
            & = \sum_{l=0}^n \left( \sum_{\substack{i+j=l \\ 0 \leq i,j \leq n}} q_i q_j \right) x^{2l},
    \end{aligned}
\end{equation}
and the sought coefficients are given by Lemma~\ref{lemma:QdecFnkodd}, which concludes the second part of the proof.
\end{itemize}

\noindent We thus obtained a feasible solution of \refprog{LP} for all $n\in\N^*$, which is feasible for \refprog{lowerSDPn} for all $m\ge n$. 

\paragraph{Optimality for $m=n$:} we now find an analytical solution of the dual (D-SDP$_n^{n,\le}$) with the same optimal value as the primal (SDP$_n^{n,\le}$), by finding a Cholesky decomposition for the matrix appearing in the dual program \refprog{lowerDSDPn} for $m=n$. The coefficients of the Cholesky decomposition are given by the triangular matrix $L$ with coefficients:
\begin{equation}
    \begin{split}
        \forall j \leq i, \quad & l_{2i,2j} = 2^i i! \binom ij \\ 
        \forall j \leq i, \quad & l_{2i+1,2j+1} = 2^{i+1/2}\frac{(i+1)!}{\sqrt{j+1}} \binom ij \\
        & l_{nn} = 0\\
        & l_{ij}=0\quad\text{otherwise}.
    \end{split}
\end{equation}
Once again, these analytical expressions were extrapolated from numerical values.
Then, $A = L L^T$ is a positive semidefinite matrix given by:
\begin{equation}
    A_{ij} = \sum_{k=0}^{\min(i,j)} l_{ik} l_{jk}.
\end{equation}
Now $l_{ij}$ is non-zero only when $i$ and $j$ have the same parity, so for all $k \in \llbracket 0,\min(i,j) \rrbracket$, $i$ and $j$ must have the same
parity than $k$ for $l_{ik} l_{jk}$ to be non-zero.

\begin{itemize}
    \item Suppose $i = 2i'$, $j = 2j'$, $i' \leq j'$ and $(i',j') \neq (n,n)$. Furthermore let $l=\frac{i+j}2$.
    \begin{equation}
      \begin{aligned}
        A_{ij} &= \sum_{\substack{k=0 \\ k \text{ even}}}^{i} l_{2i', k} l_{2j', k'}\\
        &= \sum_{k=0}^{i'} 2^{i'} i'! \binom{i'}k 2^{j'} j'! \binom{j'}k \\
        &= 2^l i'! (l-i')! \sum_{k=0}^{i'} \binom{i'}k\binom{l-i'}k \\ 
        &= 2^l l!.
        \end{aligned}
    \end{equation}
    \item Suppose $i = 2i'+1$, $j = 2j'+1$, $i' \leq j'$ and $(i',j') \neq (n,n)$. Furthermore let $l=\frac{i+j}2$.
    \begin{equation}
      \begin{aligned}
        A_{ij} &= \sum_{\substack{k=0 \\ k \text{ odd}}}^{i} l_{2i'+1 ,k} l_{2j'+1 ,k}\\
        &= \sum_{k=0}^{i'} 2^{i'+1/2} \frac{(i'+1)!}{\sqrt{k+1}} \binom{i'}k 2^{j'+1/2} \frac{(j'+1)!}{\sqrt{k+1}} \binom{j'}k \\
        &= 2^l i'! (l-i')! \sum_{k=0}^{i'} \binom{i'+1}{k+1}\binom{l-i'-1}k \\
        &= 2^l l!.      
    \end{aligned}
    \end{equation}
    \item Suppose $n = 2t$:
    \begin{equation}
      \begin{aligned}
        A_{nn} &= \sum_{k=0}^{t} l_{2t,2k}^2 \\
        &= 2^{2t} (t!)^2 \sum_{k=0}^{t-1} \binom tk^2\\
        &= 2^n (t!)^2 \left( \binom{2t}t - 1  \right) \\
        &= 2^n n! \left( 1- \binom n{\floor*{\frac n2}}^{-1}  \right).
    \end{aligned}
    \end{equation}
    \item Suppose $n = 2t+1$:
    \begin{equation}
      \begin{aligned}
        A_{nn} &= \sum_{\substack{k=0 \\ k \text{ odd}}}^{t} l_{2t+1,k}^2 \\
        &= 2^{2t+1} (t+1)!^2 \sum_{k=0}^{t-1} \frac{1}{k+1}\binom tk^2\\
        &= 2^n t! (t+1)! \sum_{k=0}^{t-1} \binom{t+1}{k+1}\binom tk \\
        &= 2^n t! (t+1)! \left( \binom{2t+1}t - 1  \right) \\
        &= 2^n n! \left( 1- \binom n{\floor*{\frac n2}}^{-1}  \right).
    \end{aligned}
    \end{equation}
\end{itemize}

\noindent In both cases, $A$ is indeed constructed as:
\begin{equation}
    (A_{\bm{\mu}})_{i,j} = \begin{cases} 
      \sum_{k=0}^l \mu_k \binom lk l! & \text{when } i+j=2l, \\
      0 & \text{otherwise},
\end{cases}
\end{equation}
for $\bm{\mu} = (F_n^n, F_n^n, \dots, F_n^n, 1 - F_n^n)$ with $F_n^n = \frac {1}{2^n} \binom n{\floor*{\frac n2}}$, and this provides a feasible solution of \refprog{lowerDSDPn} for $m=n$, with value $F_n^n$.
This shows the optimality of $\bm F^n$ for (SDP$_n^{n,\le}$) (and the fact that strong duality holds between the programs (SDP$_n^{n,\le}$) and (D-SDP$_n^{n,\le}$), which we already knew from Theorem~\ref{th:sdlower}).

\end{proof}

\noindent For \refprog{lowerDSDPn}, we see numerically that the optimal solution is the same as for (D-SDP$_n^{n,\le}$), for a few values of $m$ greater than $n$. However, this is no longer the case for higher values, for example when $n=3$ and $m=10$.


\newpage
\hypersetup{linkcolor=black}
\section{Strong duality between  \texorpdfstring{\ref{prog:LP_recall}}{(LP)} and  \texorpdfstring{\ref{prog:DLP_recall}}{(D-LP)}}
\label{sec:app_sdLP}
\hypersetup{linkcolor=cyan}

Theorem~\ref{th:sdLP} shows that the optimal values of \refprog{LP_recall} and \refprog{DLP_recall} are equal, i.e., that we have strong duality between those programs both when they are expressed in the search space $L^2(\R_+)$ and $\mathcal S(\R_+)$.

In this section we give another proof of strong duality between these linear programs as usually done in the literature of infinite-dimensional optimisation \cite{henrion_14,lasserre_10}, when they are expressed in the search space $L^2(\R_+)$. 

Note that this result itself is not enough to provide the convergence of our hierarchy of semidefinite programs, hence the need for another proof technique: strong duality directly results from the fact the feasible set of \refprog{LP_recall} is closed when it is expressed over the search space $L^2(\R_+)$.

We use notations from Appendix~\ref{sec:app_dualSDP}. From \cite[IV--(7.2)]{barvinok_02}, there is no duality gap between \refprog{LPstdform} and \refprog{DLPstdform} if there is a primal feasible plan and the cone:
\begin{equation}
    \begin{aligned}
        \mathcal{K} &= \Big\{ \Big( A(e_1), \langle e_1,c \rangle_1 \Big) : e_1 \in K_1   \Big\} \\
        &= \Big\{ \Big( \sum_k u_k, x \in \R_ + \mapsto f(x) - \sum_k u_k \mathcal L_k(x), F_n \Big) : ((u_k),f) \in K_1 \Big\}
    \end{aligned}
\end{equation}
is closed in $E_2 \oplus \R$ (for the weak topology).

\begin{proof}
The null sequence provides a feasible plan for the primal problem.

Next, we consider a sequence $(e_{1j})_j = (((u_k^j)_k)_j\footnote{Because we are dealing with a sequence of sequences, we use the upper index to refer to the embracing sequence.},(f_j)_j) \in K_1^{\N} = \ell^2(\N)\times L^2_+(\R_+)^{\N} $ and we want to show that the accumulation point $(b,g,a) = \lim_{j \rightarrow \infty} (A(e_{1j}), \langle e_{1j},c \rangle_1) $ belongs to $\mathcal{K}$ where $a,b \in \R$ and $g \in L^2(\R_+)$. 

For all $j \in \N$, $(u_k^j)_k \in \ell^2$ and for all $k \in \N$, $u_k^j$ is bounded. 
Thus, for all $k \in \N$, the sequence $(u_k^j)_j$ is bounded and via diagonal extraction there exists $\phi: \N \rightarrow \N$ strictly increasing such that $(u_k^j)_{\phi(j)}$ converges. 
We denote $\tilde{u}_k$ its limit. Since $\ell^2$ is closed, the sequence $(\tilde{u}_k)_k$ belongs to $\ell^2$ and we have $b = \sum_k \tilde{u}_k$ and $a=\tilde{u}_n$. 

Now $f_j - \sum_k u_k^j \mathcal L_k \longrightarrow g$ so that $f_j \longrightarrow g + \sum_k \tilde{u}_k \mathcal L_k \in L^2_+(\R_+)$ since $L^2_+(\R_+)$ is closed.
Thus, for $\tilde{e}_1 = ((\tilde{u}_k)_k, g + \sum_k \tilde{u}_k \mathcal L_k) \in K_1$, $(b,g,a) = (A(\tilde{e}_1), \langle \tilde{e}_1,c \rangle_1)$ and $(b,g,a) \in \mathcal{K}$.
\end{proof}


\newpage
\section{Multimode case}
\label{sec:app_multi}

In this section, we provide the technical background for proving the convergence of the hierarchies of semidefinite programs in the multimode setting. The results obtained are summarised in Fig.~\ref{fig:structuremulti}.

\subsection{Multimode notations and definitions}
\label{sec:app_multidef}

We use bold math for multi-index notations. The main advantage of these notations is that the proofs of most technical results are easily extended to the multimode setting by replacing standard notations by multi-index notations.

Let $M$ denote the number of modes. We consider $M$ copies $\mathcal H^{\otimes M}$ of a separable Hilbert space $\mathcal H$. We denote the corresponding multimode orthonormal Fock basis by $\{\ket{\bm n}\}_{\bm n\in\mathbb N^M}$.
For all $\bm\alpha=(\alpha_1,\dots,\alpha_M)\in\mathbb C^M$, $\bm k=(k_1,\dots,k_M)\in\mathbb N^M$, $\bm n=(n_1,\dots,n_M)\in\mathbb N^M$ and $m\in\mathbb N$, we introduce the notations:
\begin{equation}
    \begin{aligned}
    \bm0&=(0,\dots,0)\in\mathbb N^M\\
    \bm1&=(1,\dots,1)\in\mathbb N^M\\
    m\bm 1&=(m,\dots,m)\\
    m\bm k&=(mk_1,\dots,mk_M)\\
    \pi_{\bm k}&=\prod_{i=1}^M(k_i+1)\\
    \hat D(\bm\alpha)&=\hat D(\alpha_1)\otimes\cdots\otimes\hat D(\alpha_M)\\
    \ket{\bm k}&=\ket{k_1}\otimes\cdots\otimes\ket{k_M}\\
    \bra{\bm k}&=\bra{k_1}\otimes\cdots\otimes\bra{k_M}\\
    \bm k\le\bm n\,&\Leftrightarrow\,k_i\le n_i\quad\forall i=1,\dots,M\\
    L_{\bm k}(\bm\alpha)&=L_{k_1}(\alpha_1)\cdots L_{k_M}(\alpha_M)\\
    \mathcal L_{\bm k}(\bm\alpha)&=\mathcal L_{k_1}(\alpha_1)\cdots\mathcal L_{k_M}(\alpha_M)\\
    |\bm k|&=k_1+\cdots+k_M\\
    \bm\alpha^{\bm k}&=\alpha_1^{k_1}\cdots\alpha_M^{k_M}\\
    \bm k!&=k_1!\cdots k_m!\\
    \binom{\bm n}{\bm k}&=\binom{n_1}{k_1}\cdots\binom{n_M}{k_M}\\
    \bm k+\bm n&=(k_1+n_1,\dots,k_M+n_M)\\
    e^{\bm\alpha}&=e^{\alpha_1}\cdots e^{\alpha_M}.
    \end{aligned}
\end{equation}
A multivariate polynomial of degree $p\in\mathbb N$ may then be written in the compact form $P(\bm x)=\sum_{|\bm l|\le p}p_{\bm l}\bm x^{\bm l}$, where the sum is over all the tuples $\bm l\in\N^M$ such that $|\bm l|\le m$, also known as the weak compositions of the integer $m$. There are $\binom{M+m}m$ such tuples. In what follows, we will also consider sums over all the tuples $\bm l\in\N^M$ such that $\bm l\le\bm k$, for $\bm k=(k_1,\dots,k_M)\in\N^M$. There are $\pi_{\bm k}$ such tuples.
In particular, for all $\bm x\in\R_+^M$ and all $\bm k\in\N^M$,
\begin{equation}
    L_ {\bm k}(\bm x)=\sum_{\bm l\le\bm k}\frac{(-1)^{|\bm l|}}{\bm l!}\binom{\bm k}{\bm l}\bm x^{\bm l}\quad\text{and}\quad \bm x^{\bm k}=\sum_{\bm l\le\bm k}(-1)^{|\bm l|}\binom{\bm k}{\bm l}\bm k!L_{\bm l}(\bm x).
\end{equation}
We extend a few definitions from the single-mode case.

For $\bm s=(s_{\bm k})_{\bm k\in\N^M}\in\R^{\N^M}$, we define the associated formal series of multivariate Laguerre functions:
\begin{equation}
    f_{\bm s}:=\sum_{\bm k}s_{\bm k}\mathcal L_{\bm k},
\end{equation}
where $\bm s$ is the so-called sequence of Laguerre moments of $f_{\bm s}$. 
For $m\in\mathbb N$, we also define the associated $\binom{M+m}m\times\binom{M+m}m$ matrix $A_{\bm s}$ by
\begin{equation}
    (A_{\bm s})_{\bm i,\bm j} = \begin{cases} 
      \sum_{\bm k\le\bm l}s_{\bm k}\binom{\bm l}{\bm k}\bm l! &\text{when }\bm i+\bm j=2\bm l, \\
      0 &\text{otherwise,}
   \end{cases}
\end{equation}
where $\bm i,\bm j\in\N^M$ with $|\bm i|\le m$ and $|\bm j|\le m$.

The multimode Laguerre functions $(\mathcal L_{\bm k})_{\bm k\in\N^M}$ form an orthonormal basis of the space $L^2(\R_+^M)$ of real square-integrable functions over $\R_+^M$ equipped with the usual scalar product:
\begin{equation}\label{eq:braketmulti}
\braket{f,g}=\int_{\R_+^M}{f(\bm x)g(\bm x)d\bm x},
\end{equation}
for $f,g\in L^2(\R_+^M)$. We denote by $L^2_+(\R_+^M)$ its subset of non-negative elements. 
The space $L^2(\R_+^M)$ is isomorphic to its dual space ${L^2}'(\R_+^M)$: elements of ${L^2}'(\R_+^M)$ can be identified by the Lebesgue measure on $\R_+^M$ times the corresponding function in $L^2(\R_+^M)$.

Moreover, the elements of the space $\mathcal S(\R_+^M)$ of Schwartz functions over $\R_+^M$, i.e., the space of $C^{\infty}$ functions that go to $0$ at infinity faster than any inverse polynomial, as do their derivatives, can be written as series of Laguerre functions with a sequence indexed by $\N^M$ of rapidly decreasing coefficients (which go to $0$ at infinity faster than any inverse $M$-variate polynomial). Its dual space $\mathcal S'(\R_+^M)$ of tempered distributions over $\R_+^M$ is characterised as the space of formal series of Laguerre functions with a slowly increasing sequence indexed by $\N^M$ of coefficients (sequences that are upper bounded by an $M$-variate polynomial)~\cite{guillemot1971developpements}.
We also extend the definition of the duality $\braket{\dummy,\dummy}$ in~\refeq{braketmulti} to these spaces.

For all $m\in\mathbb N$, the set of series of Laguerre functions over $\R_+^M$ truncated at $m$ is denoted $\mathcal R_m(\R_+^M)$. This is the set of $M$-variate polynomials $P(\bm x)=\sum_{|\bm k|\le m}p_{\bm k}\bm x^{\bm k}$ of degree at most $m$ multiplied by the function $\bm x\mapsto e^{-\frac{\bm x}2}$. 
Let $\mathcal R_{m,+}(\R_+^M)$ denotes its subset of non-negative elements where the polynomial $P$ is such that $\bm x\mapsto P(\bm x^2)$ has a sum-of-squares decomposition. 

Similarly, for $\bm m\in\N^M$, the set of truncated series of Laguerre functions over $\R_+^M$ with monomials smaller than $\bm m$ is denoted $\mathcal R_{\bm m}(\R_+^M)$. This is the set of $M$-variate polynomials $P(\bm x)=\sum_{\bm k\le\bm m}p_{\bm k}\bm x^{\bm k}$, multiplied by the function $\bm x\mapsto e^{-\frac{\bm x}2}$. Let $\mathcal R_{\bm m,+}(\R_+^M)$ denotes its subset of non-negative elements where the polynomial $P$ is such that $\bm x\mapsto P(\bm x^2)$ has a sum-of-squares decomposition. 

We recall here the expressions of the linear program~\refprog{LP_multi} and its dual~\refprog{DLP_multi}:
\leqnomode
\begin{flalign*}
   \label{prog:LP_multiapp}
    \tag*{(LP$_{\bm n}^{L^2}$)}
    \hspace{2.8cm}\left\{
        \begin{aligned}
            & & & \Sup{(F_{\bm k})_{\bm k\in\N^M}\in\ell^2(\N^M)} F_{\bm n}  \\
            & \text{subject to} & & \sum_{\bm k} F_{\bm k} =  1  \\
            & \text{and} & & \forall \bm k \in \N^M, \quad F_{\bm k} \geq  0  \\
            & \text{and} & & \forall \bm x \in \R_+^M, \quad\sum_{\bm k} F_{\bm k}\mathcal L_{\bm k}(\bm x) \geq  0,
        \end{aligned}
    \right. &&
\end{flalign*}
\reqnomode
where the optimisation is over real sequences indexed by elements of $\mathbb N^M$. Its dual linear program reads
\leqnomode
\begin{flalign*}
   \label{prog:DLP_multiapp}
    \tag*{(D-LP$_{\bm n}^{L^2}$)}
    \hspace{2.8cm}\left\{
        \begin{aligned}
            & & & \Inf{\substack{y\in\R\\\mu\in {L^2}'(\R_+^M)}}y  \\
            & \text{subject to} & & \forall \bm k \neq \bm n \in \N^M,\; y \geq \int_{\R_+^M}{\mathcal L_{\bm k}}{d\mu} \\
            & \text{and} & & y \geq 1 + \int_{\R_+^M}{\mathcal L_{\bm n}}{d\mu}\\
            & \text{and} & & \forall f \in L^2_+(\R_+^M), \; \langle \mu,f \rangle \geq  0.
        \end{aligned}
    \right. &&
\end{flalign*}
\reqnomode
\noindent We also recall the expression or \refprog{LP_multiS} and \refprog{DLP_multiS}:
\leqnomode
\begin{flalign*}
   \label{prog:LP_multiappS}
    \tag*{(LP$_{\bm n}^{\mathcal S}$)}
    \hspace{2.8cm}\left\{
        \begin{aligned}
            & & & \Sup{(F_{\bm k})_{\bm k\in\N^M}\in\ \mathcal S(\N^M)} F_{\bm n}  \\
            & \text{subject to} & & \sum_{\bm k} F_{\bm k} =  1  \\
            & \text{and} & & \forall \bm k \in \N^M, \quad F_{\bm k} \geq  0  \\
            & \text{and} & & \forall \bm x \in \R_+^M, \quad\sum_{\bm k} F_{\bm k}\mathcal L_{\bm k}(\bm x) \geq  0,
        \end{aligned}
    \right. &&
\end{flalign*}
\reqnomode
\leqnomode
\begin{flalign*}
   \label{prog:DLP_multiappS}
    \tag*{(D-LP$_{\bm n}^{\mathcal S}$)}
    \hspace{2.8cm}\left\{
        \begin{aligned}
            & & & \Inf{\substack{y\in\R\\\mu\in \mathcal S'(\R_+^M)}}y  \\
            & \text{subject to} & & \forall \bm k \neq \bm n \in \N^M,\; y \geq \int_{\R_+^M}{\mathcal L_{\bm k}}{d\mu} \\
            & \text{and} & & y \geq 1 + \int_{\R_+^M}{\mathcal L_{\bm n}}{d\mu}\\
            & \text{and} & & \forall f \in L^2_+(\R_+^M), \; \langle \mu,f \rangle \geq  0.
        \end{aligned}
    \right. &&
\end{flalign*}
\reqnomode

\subsection{Multimode semidefinite programs}
\label{sec:app_multiSDP}

\begin{figure}[ht]
	\begin{center}
		\centering

\begin{tikzpicture}[scale=0.70]

\node[inner sep=0pt] (A) at (-0.8,-1) {};
\node[inner sep=0pt] (B) at (-0.8,8) {};
\node[inner sep=0pt] (C1) at (-0.8,0) {};
\node[inner sep=0pt] (C2') at (-0.8,2.75) {};
\node[inner sep=0pt] (C2) at (-0.8,4.25) {};
\node[inner sep=0pt] (C3) at (-0.8,7) {};

\draw[|-|] (-0.8,-1) -- (-0.8,-1);
\draw[-|] (-0.8,-1) -- (-0.8,0);
\draw[-|] (-0.8,0) -- (-0.8,2.75);
\draw[dashed] (-0.8,2.85) -- (-0.8,4.2);
\draw[|-|] (-0.8,4.25) -- (-0.8,7);
\draw[-|] (-0.8,7) -- (-0.8,8);

\node[left] (zero) at (A) {0};
\node[left] (one) at (B) {1};
\node[left] (omegainf) at (C1) {$\omega_{\bm n}^{m\bm1,\le}$};
\node[left] (omega) at (C2) {$\omega_{\bm n}^{L^2}$};
\node[left] (omega) at (C2') {$\omega_{\bm n}^{\mathcal S}$};
\node[left] (omegasup) at (C3) {$\omega_{\bm n}^{m\bm1,\ge}$};

\hypersetup{linkcolor=black}
\node[inner sep=0pt] (SDPsup) at (2,7) {\Large \ref{prog:upperSDPnalt1_multiapp}};
\node[inner sep=0pt] (D-SDPsup) at (9,7) {\Large \ref{prog:upperDSDPnalt1_multiapp}};

\node[inner sep=0pt] (LP) at (2,4.25) {\Large \ref{prog:LP_multiapp}};
\node[inner sep=0pt] (D-LP) at (9,4.25) {\Large \ref{prog:DLP_multiapp}};

\node[inner sep=0pt] (LPS) at (2,2.75) {\Large \ref{prog:LP_multiappS}};
\node[inner sep=0pt] (D-LPS) at (9,2.75) {\Large \ref{prog:DLP_multiappS}};

\node[inner sep=0pt] (SDPinf) at (2,0) {\Large \ref{prog:lowerSDPnalt1_multiapp}};
\node[inner sep=0pt] (D-SDPinf) at (9,0) {\Large \ref{prog:lowerDSDPnalt1_multiapp}};
\hypersetup{linkcolor=cyan}

\node[inner sep=0pt] (eqsup) at ($.54*(SDPsup)+.46*(D-SDPsup)$) {$\substack{\text{Theorem}~\ref{th:sduppermulti} \\ =\joinrel=}$};

\node[inner sep=0pt] (eq) at ($.54*(LP)+.46*(D-LP)$) {$\substack{\text{Theorem}~\ref{th:sdLPmulti} \\ =\joinrel=}$};

\node[inner sep=0pt] (eq) at ($.54*(LPS)+.46*(D-LPS)$) {$\substack{\text{Theorem}~\ref{th:sdLPmulti} \\ =\joinrel=}$};

\node[inner sep=0pt] (eqinf) at ($.54*(SDPinf)+.46*(D-SDPinf)$) {$\substack{\text{Theorem}~\ref{th:sdlowermulti} \\ =\joinrel=}$};

\node[inner sep=0pt] (CVlu) at ($.5*(SDPsup)+.5*(LP)$) {$\substack{\text{Theorem~\ref{th:upperCVmulti}} \\\hspace{20pt} \big\downarrow m\rightarrow \infty}$};

\node[inner sep=0pt] (CVld) at ($.5*(SDPinf)+.5*(LPS)$) {$\substack{\hspace{20pt} \big\uparrow m\rightarrow \infty \\ \text{Theorem~\ref{th:lowerCVmulti}}}$};

\node[inner sep=0pt] (CVru) at ($.5*(D-SDPsup)+.5*(D-LP)$) {$\substack{\text{Theorem~\ref{th:upperCVmulti}} \\\hspace{20pt} \big\downarrow m\rightarrow \infty}$};

\node[inner sep=0pt] (CVrd) at ($.5*(D-SDPinf)+.5*(D-LPS)$) {$\substack{\hspace{20pt} \big\uparrow m\rightarrow \infty \\ \text{Theorem~\ref{th:lowerCVmulti}}}$};

\end{tikzpicture}
		\hypersetup{linkcolor=black}
		\caption{Multimode hierarchies of semidefinite relaxations and restrictions converging to the linear program \ref{prog:LP_multiapp}, together with their dual programs. The upper index $m$ denotes the level of the relaxation or restriction. On the left are the associated optimal values. The equal sign denotes strong duality, i.e., equality of optimal values, and the arrows denote convergence of the corresponding sequences of optimal values. The hierarchies \ref{prog:upperSDPn_multi} and \ref{prog:lowerSDPn_multi} in the main text are different from the ones appearing in the figure, but equivalent by Lemma~\ref{lem:equiv}. The question of whether $\omega_{\bm n}^{L^2} = \omega_{\bm n}^{\mathcal S}$ is left open.}
		\hypersetup{linkcolor=cyan}
		\label{fig:structuremulti}
	\end{center}
\end{figure}
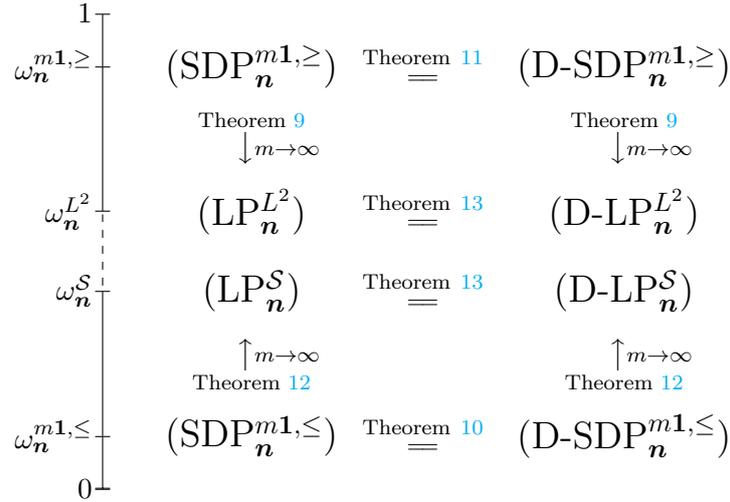

Hereafter, we state without proofs the technical results used to derive the semidefinite programs in section~\ref{sec:multi} and their dual programs. It is a straightforward exercise to obtain the proofs from their single-mode version by using multi-index notations.

There are two natural ways to obtain relaxations and restrictions by replacing constraints on non-negative functions by constraints on non-negative polynomials: either by considering polynomials $P(\bm x)=\sum_{|\bm k|\le m}p_{\bm k}\bm x^{\bm k}$ of degree at most $m$ for $m\in\N$, or by considering polynomials $P(\bm x)=\sum_{\bm k\le \bm m}p_{\bm k}\bm x^{\bm k}$ with monomials smaller than $\bm m$ for $m\in\N^M$. Note that when $M=1$ these two are equivalent. 

All the results below containing conditions of the form $|\bm k|\le m$ for $m\in\N$ are also valid when replaced by conditions of the form $\bm k\le\bm m$ for $\bm m=(m_1,\dots,m_M)\in\N^M$, with the same proofs, by replacing $\binom{M+m}m$ by $\pi_{\bm m}=\prod_{i=1}^M(m_i+1)$.

\begin{lemma}[Equivalent of Lemma~\ref{lem:pospolyR}]\label{lem:pospolyRmulti}
Let $p\in\mathbb N$ and let $P$ be a multivariate polynomial of degree $2p$. Let $\bm X=(\bm x^{\bm k})_{|\bm k|\le p}$ be the vector of monomials. Then, $P$ has a sum-of-squares decomposition if and only if there exists a real ${\binom{M+p}p}\times{\binom{M+p}p}$ positive semidefinite matrix $Q$ such that for all $x\in\R^M$,
\begin{equation}
    P(\bm x)=\bm X^TQ\bm X.
\end{equation}
\end{lemma}

\begin{lemma}[Generalisation of Lemma~\ref{lem:pospolyR+}]
Let $P$ be a non-negative polynomial over $\R_+^M$ such that $\bm x\mapsto P(\bm x^2)$ has a sum-of-squares decomposition. Then, $P$ can be written as a sum of polynomials of the form $\sum_{|\bm l|\le p}\bm x^{\bm l}\sum_{\bm i+\bm j=2\bm l}y_{\bm i}y_{\bm j}$ for $p\in\N$ and $y_{\bm i}\in\mathbb R$ for all $\bm i\in\mathbb N^M$ such that $|\bm i|\le p$.
\end{lemma}

\begin{lemma}[Generalisation of Lemma~\ref{lem:momentmatrix}]\label{lem:momentmatrixmulti}
Let $m\in\N$ and let $\bm s=(s_{\bm k})_{\bm k\in\N^M}\in\R^{\N^M}$. The following propositions are equivalent:
\begin{enumerate}[label=(\roman*)]
\item $\forall g\in\mathcal R_{m,+}(\R_+^M),\;\braket{f_{\bm s},g}\ge0$,
\item $A_{\bm s}\succeq0$.
\end{enumerate}
\end{lemma}

Using Lemma~\ref{lem:momentmatrixmulti} we obtain the semidefinite \textit{relaxations}:
\leqnomode
\begin{flalign*}
   \label{prog:upperSDPn_multiapp}
    \tag*{$(\text{SDP}^{m,\geq}_{\bm n})$}
    \hspace{2.8cm}\left\{
        \begin{aligned}
            & & & \Sup{\substack{A\in\text{Sym}_{\binom{M+m}m}\\ \bm F\in\R^{\binom{M+m}m}}} F_{\bm n}  \\
            & \text{subject to} & & \sum_{|\bm k|\le m} F_{\bm k} =  1  \\
            & \text{and} & & \forall |\bm k|\le m, \quad F_{\bm k} \geq  0 \\
            & \text{and} & & \forall|\bm l|\le m,\forall\bm i+\bm j=2\bm l,\quad A_{\bm i\bm j}=\sum\limits_{\bm k\le\bm l}F_{\bm k}\binom{\bm l}{\bm k}\bm l!\\
            & \text{and} & & \forall|\bm r|\le2m,\bm r\!\neq\!2\bm l,\forall|\bm l|\le m,\forall\bm i+\bm j=\bm r,\;A_{\bm i\bm j}=0\\
            & \text{and} & & A \succeq 0.
        \end{aligned}
    \right. &&
\end{flalign*}
\reqnomode
for all $m\ge|\bm n|$. We denote its optimal value by $\omega_{\bm n}^{m,\ge}$.
The corresponding dual programs are given by:
\leqnomode
\begin{flalign*}
    \label{prog:upperDSDPn_multiapp}
    \tag*{(D-SDP$_{\bm n}^{m,\ge}$)}\hspace{2.8cm}\left\{
        \begin{aligned}
            & & & \Inf{\substack{Q\in\text{Sym}_{\binom{M+m}m}\\y \in \R,\bm\mu\in\R^{\binom{M+m}m}}} y\\
            & \text{subject to} & & y\ge1+\mu_{\bm n}\\
            &\text{and} & & \forall |\bm k|\le m,\bm k\neq\bm n,\quad y\ge\mu_{\bm k}\\
            &\text{and} & & \forall|\bm l|\le m,\quad\sum_{\bm i+\bm j=2\bm l}Q_{\bm i\bm j}=\sum_{\bm k\ge\bm l}\frac{(-1)^{|\bm k|+|\bm l|}}{\bm l!}\binom{\bm k}{\bm l}\mu_{\bm k}\hspace{-2cm}\\
            &\text{and} & & Q\succeq0,
        \end{aligned}
    \right. &&
\end{flalign*}
\reqnomode
for all $m\ge|\bm n|$. Similarly, using Lemma~\ref{lem:pospolyRmulti} we obtain the semidefinite \textit{restrictions}:
\leqnomode
\begin{flalign*}
    \label{prog:lowerSDPn_multiapp}
    \tag*{(SDP$_{\bm n}^{m,\le}$)}\hspace{2.8cm}\left\{
        \begin{aligned}
            & & & \Sup{\substack{Q\in\text{Sym}_{\binom{M+m}m}\\\bm F\in\R^{\binom{M+m}m}}} F_{\bm n}\\
            & \text{subject to} & & \sum_{|\bm k|\le m} F_{\bm k} =  1\\
            &\text{and} & & \forall |\bm k|\le m, \quad F_{\bm k} \geq  0\\
            &\text{and} & & \forall |\bm l|\le m,\quad\sum_{\bm i+\bm j=2\bm l}Q_{\bm i\bm j}=\sum_{\bm k\ge\bm l} \frac{(-1)^{|\bm k|+|\bm l|}}{\bm l!}\binom{\bm k}{\bm l}F_{\bm k}\hspace{-4cm}\\
            &\text{and} & & \forall |\bm r|\le 2m,\bm r\neq2\bm l,\forall |\bm l|\le m,\quad\sum_{\bm i+\bm j=\bm r}Q_{\bm i\bm j}=0\hspace{-4cm}\\
            &\text{and} & & Q\succeq0,
        \end{aligned}
    \right. &&
\end{flalign*}
\reqnomode
for all $m\ge|\bm n|$. We denote its optimal value by $\omega_{\bm n}^{m,\le}$.
The corresponding dual programs are given by:
\leqnomode
\begin{flalign*}
   \label{prog:lowerDSDPn_multiapp}
    \tag*{$(\text{D-SDP}^{m,\leq}_{\bm n})$}
    \hspace{2.8cm}\left\{
        \begin{aligned}
            & & & \Inf{\substack{A\in\text{Sym}_{\binom{M+m}m}\\y,\bm\mu\in\R\times\R^{\binom{M+m}m}}} y\\
            & \text{subject to} & & y\ge1+\mu_{\bm n}\\
            & \text{and} & & \forall |\bm k|\le m,\bm k\neq\bm n,\quad y\ge\mu_{\bm k}\\
            & \text{and} & & \forall|\bm l|\le m,\forall\bm i+\bm j=2\bm l,\quad A_{\bm i\bm j}=\sum\limits_{\bm k\le\bm l}\mu_{\bm k}\binom{\bm l}{\bm k}\bm l!\\
            & \text{and} & & A \succeq 0,
        \end{aligned}
    \right. &&
\end{flalign*}
\reqnomode
for all $m\ge|\bm n|$. Like in the single-mode case, note that without loss of generality the condition $y\le1$, and thus $\mu_{\bm k}\le1$ for all $\bm k$, can be added to the optimisation, since setting $A=0$, $y=1$ and $\bm\mu=0$ gives a feasible solution with objective value 1.

These are the relaxations and restrictions of~\refprog{LP_multiapp} obtained by considering polynomials of degree less or equal to $m$, where the optimisation is over matrices and vectors indexed by elements of $\mathbb N^m$ with sum of coefficients lower that $m$. Alternatively, we may also consider the relaxations and restrictions obtained by considering polynomials with monomials smaller than $\bm m\in\N^M$, where the optimisation is over matrices and vectors indexed by elements of $\mathbb N^m$ lower that $\bm m$. Recalling the notation $\pi_{\bm m}=\prod_{i=1}^M(m_i+1)$, the corresponding semidefinite \textit{relaxations} are given by
\leqnomode
\begin{flalign*}
   \label{prog:upperSDPnalt_multiapp}
    \tag*{$(\text{SDP}^{\bm m,\geq}_{\bm n})$}
    \hspace{2.8cm}\left\{
        \begin{aligned}
            & & & \Sup{\substack{A\in\text{Sym}_{\pi_{\bm m}}\\ \bm F\in\R^{\pi_{\bm m}}}} F_{\bm n}  \\
            & \text{subject to} & & \sum_{\bm k\le\bm m} F_{\bm k} =  1  \\
            & \text{and} & & \forall\bm k\le\bm m, \quad F_{\bm k} \geq  0 \\
            & \text{and} & & \forall\bm l\le\bm m,\forall\bm i+\bm j=2\bm l,\quad A_{\bm i\bm j}=\sum\limits_{\bm k\le\bm l}F_{\bm k}\binom{\bm l}{\bm k}\bm l!\\
            & \text{and} & & \forall\bm r\le2\bm m,\bm r\!\neq\!2\bm l,\forall\bm l\le\bm m,\forall\bm i+\bm j=\bm r,\;A_{\bm i\bm j}=0\\
            & \text{and} & & A \succeq 0.
        \end{aligned}
    \right. &&
\end{flalign*}
\reqnomode
for all $\bm m\ge\bm n$. We denote its optimal value by $\omega_{\bm n}^{\bm m,\ge}$.
The corresponding dual programs are given by:
\leqnomode
\begin{flalign*}
    \label{prog:upperDSDPnalt_multiapp}
    \tag*{(D-SDP$_{\bm n}^{\bm m,\ge}$)}\hspace{2.8cm}\left\{
        \begin{aligned}
            & & & \Inf{\substack{Q\in\text{Sym}_{\pi_{\bm m}}\\y,\bm\mu\in\R\times\R^{\pi_{\bm m}}}} y\\
            & \text{subject to} & & y\ge1+\mu_{\bm n}\\
            &\text{and} & & \forall\bm k\le\bm m,\bm k\neq\bm n,\quad y\ge\mu_{\bm k}\\
            &\text{and} & & \forall\bm l\le\bm m,\quad\sum_{\bm i+\bm j=2\bm l}Q_{\bm i\bm j}=\sum_{\bm k\ge\bm l}\frac{(-1)^{|\bm k|+|\bm l|}}{\bm l!}\binom{\bm k}{\bm l}\mu_{\bm k}\hspace{-2cm}\\
            &\text{and} & & Q\succeq0,
        \end{aligned}
    \right. &&
\end{flalign*}
\reqnomode
for all $\bm m\ge\bm n$.
Similarly, the semidefinite \textit{restrictions} are given by:
\leqnomode
\begin{flalign*}
    \label{prog:lowerSDPnalt_multiapp}
    \tag*{(SDP$_{\bm n}^{\bm m,\le}$)}\hspace{2.8cm}\left\{
        \begin{aligned}
            & & & \Sup{\substack{Q\in\text{Sym}_{\pi_{\bm m}}\\\bm F\in\R^{\pi_{\bm m}}}} F_{\bm n}\\
            & \text{subject to} & & \sum_{\bm k\le\bm m} F_{\bm k} =  1\\
            &\text{and} & & \forall\bm k\le\bm m, \quad F_{\bm k} \geq  0\\
            &\text{and} & & \forall\bm l\le\bm m,\quad\sum_{\bm i+\bm j=2\bm l}Q_{\bm i\bm j}=\sum_{\bm k\ge\bm l} \frac{(-1)^{|\bm k|+|\bm l|}}{\bm l!}\binom{\bm k}{\bm l}F_{\bm k}\hspace{-4cm}\\
            &\text{and} & & \forall \bm r\le 2\bm m,\bm r\neq2\bm l,\forall\bm l\le\bm m,\quad\sum_{\bm i+\bm j=\bm r}Q_{\bm i\bm j}=0\hspace{-4cm}\\
            &\text{and} & & Q\succeq0,
        \end{aligned}
    \right. &&
\end{flalign*}
\reqnomode
for all $\bm m\ge\bm n$. We denote its optimal value by $\omega_{\bm n}^{\bm m,\le}$.
The corresponding dual programs are given by:
\leqnomode
\begin{flalign*}
   \label{prog:lowerDSDPnalt_multiapp}
    \tag*{$(\text{D-SDP}^{\bm m,\leq}_{\bm n})$}
    \hspace{2.8cm}\left\{
        \begin{aligned}
            & & & \Inf{\substack{A\in\text{Sym}_{\pi_{\bm m}}\\y,\bm\mu\in\R\times\R^{\pi_{\bm m}}}} y\\
            & \text{subject to} & & y\ge1+\mu_{\bm n}\\
            & \text{and} & & \forall\bm k\le\bm m,\bm k\neq\bm n,\quad y\ge\mu_{\bm k}\\
            & \text{and} & & \forall\bm l\le\bm m,\forall\bm i+\bm j=2\bm l,\quad A_{\bm i\bm j}=\sum\limits_{\bm k\le\bm l}\mu_{\bm k}\binom{\bm l}{\bm k}\bm l!\\
            & \text{and} & & A \succeq 0,
        \end{aligned}
    \right. &&
\end{flalign*}
\reqnomode
for all $\bm m\ge\bm n$. 

The programs~\refprog{upperSDPn_multiapp} and~\refprog{lowerSDPn_multiapp} respectively provide hierarchies of relaxations and restrictions of~\refprog{LP_multiapp}, since the set of $M$-variate polynomials of degree $m$ is included in the set of $M$-variate polynomials of degree $m+1$.
On the other hand, there is no natural ordering in $\mathbb N^M$ of the relaxations~\refprog{upperSDPnalt_multiapp} or the restrictions~\refprog{lowerSDPnalt_multiapp} (consider for instance $\bm m=(2,1)$ and $\bm m'=(1,2)$). In order to obtain proper hierarchies of semidefinite programs, we thus consider the subset of these programs where the tuple $\bm m$ is of the form $m\bm1=(m,\dots,m)\in\N^M$, for $m\in\N$. We have $\pi_{m\bm1}=(m+1)^M$, and the \textit{relaxations} are then given by
\leqnomode
\begin{flalign*}
   \label{prog:upperSDPnalt1_multiapp}
    \tag*{$(\text{SDP}^{m\bm1,\geq}_{\bm n})$}
    \hspace{2.8cm}\left\{
        \begin{aligned}
            & & & \Sup{\substack{A\in\text{Sym}_{(m+1)^M}\\ \bm F\in\R^{(m+1)^M}}} F_{\bm n}  \\
            & \text{subject to} & & \sum_{\bm k\le m\bm1} F_{\bm k} =  1  \\
            & \text{and} & & \forall\bm k\le m\bm1, \quad F_{\bm k} \geq  0 \\
            & \text{and} & & \forall\bm l\le m\bm1,\forall\bm i+\bm j=2\bm l,\quad A_{\bm i\bm j}=\sum\limits_{\bm k\le\bm l}F_{\bm k}\binom{\bm l}{\bm k}\bm l!\\
            & \text{and} & & \forall\bm r\le2m\bm 1,\bm r\!\neq\!2\bm l,\forall\bm l\le m\bm1,\forall\bm i+\bm j=\bm r,\;A_{\bm i\bm j}=0\\
            & \text{and} & & A \succeq 0.
        \end{aligned}
    \right. &&
\end{flalign*}
\reqnomode
for $m\ge\max_in_i$. We denote its optimal value by $\omega_{\bm n}^{m\bm1,\ge}$.
The corresponding dual programs are given by:
\leqnomode
\begin{flalign*}
    \label{prog:upperDSDPnalt1_multiapp}
    \tag*{(D-SDP$_{\bm n}^{m\bm1,\ge}$)}\hspace{2.8cm}\left\{
        \begin{aligned}
            & & & \Inf{\substack{Q\in\text{Sym}_{(m+1)^M}\\y,\bm\mu\in\R\times\R^{(m+1)^M}}} y\\
            & \text{subject to} & & y\ge1+\mu_{\bm n}\\
            &\text{and} & & \forall\bm k\le m\bm1,\bm k\neq\bm n,\quad y\ge\mu_{\bm k}\\
            &\text{and} & & \forall\bm l\le m\bm1,\quad\sum_{\bm i+\bm j=2\bm l}Q_{\bm i\bm j}=\sum_{\bm k\ge\bm l}\frac{(-1)^{|\bm k|+|\bm l|}}{\bm l!}\binom{\bm k}{\bm l}\mu_{\bm k}\hspace{-2cm}\\
            &\text{and} & & Q\succeq0,
        \end{aligned}
    \right. &&
\end{flalign*}
\reqnomode
for $m\ge\max_in_i$.
Similarly, the \textit{restrictions} are given by:
\leqnomode
\begin{flalign*}
    \label{prog:lowerSDPnalt1_multiapp}
    \tag*{(SDP$_{\bm n}^{m\bm1,\le}$)}\hspace{2.8cm}\left\{
        \begin{aligned}
            & & & \Sup{\substack{Q\in\text{Sym}_{(m+1)^M}\\\bm F\in\R^{(m+1)^M}}} F_{\bm n}\\
            & \text{subject to} & & \sum_{\bm k\le m\bm1} F_{\bm k} =  1\\
            &\text{and} & & \forall\bm k\le m\bm1, \quad F_{\bm k} \geq  0\\
            &\text{and} & & \forall\bm l\le m\bm1,\quad\sum_{\bm i+\bm j=2\bm l}Q_{\bm i\bm j}=\sum_{\bm k\ge\bm l} \frac{(-1)^{|\bm k|+|\bm l|}}{\bm l!}\binom{\bm k}{\bm l}F_{\bm k}\hspace{-4cm}\\
            &\text{and} & & \forall \bm r\le2m\bm1,\bm r\neq2\bm l,\forall\bm l\le m\bm1,\quad\sum_{\bm i+\bm j=\bm r}Q_{\bm i\bm j}=0\hspace{-4cm}\\
            &\text{and} & & Q\succeq0,
        \end{aligned}
    \right. &&
\end{flalign*}
\reqnomode
for $m\ge\max_in_i$. We denote its optimal value by $\omega_{\bm n}^{m\bm1,\le}$.
The corresponding dual programs are given by:
\leqnomode
\begin{flalign*}
   \label{prog:lowerDSDPnalt1_multiapp}
    \tag*{$(\text{D-SDP}^{m\bm1,\leq}_{\bm n})$}
    \hspace{2.8cm}\left\{
        \begin{aligned}
            & & & \Inf{\substack{A\in\text{Sym}_{(m+1)^M}\\y,\bm\mu\in\R\times\R^{(m+1)^M}}} y\\
            & \text{subject to} & & y\ge1+\mu_{\bm n}\\
            & \text{and} & & \forall\bm k\le m\bm1,\bm k\neq\bm n,\quad y\ge\mu_{\bm k}\\
            & \text{and} & & \forall\bm l\le m\bm1,\forall\bm i+\bm j=2\bm l,\quad A_{\bm i\bm j}=\sum\limits_{\bm k\le\bm l}\mu_{\bm k}\binom{\bm l}{\bm k}\bm l!\\
            & \text{and} & & A \succeq 0,
        \end{aligned}
    \right. &&
\end{flalign*}
\reqnomode
for $m\ge\max_in_i$. 

The programs~\refprog{upperSDPnalt1_multiapp} and~\refprog{lowerSDPnalt1_multiapp}  are the relaxations and restrictions of~\refprog{LP_multiapp} obtained by considering polynomials of individual degree in each variable less or equal to $m$. These programs respectively provide hierarchies of relaxations and restrictions of~\refprog{LP_multiapp}, since the set of $M$-variate polynomials with monomials lower than $m\bm1$ is included in the set of $M$-variate polynomials with monomials lower than $(m+1)\bm1$. 

Note that these hierarchies of programs obtained by setting $\bm m$ of the form $m\bm1$ capture the behaviour of all bounds that can be obtained from the more general family of programs indexed by $\bm m=(m_1,\dots,m_M)$, since $\bm m\le(\max_im_i)\bm1$, i.e., the bound obtained by considering the program indexed by $(\max_im_i)\bm1$ supersedes the bound obtained by considering the program indexed by $\bm m$. Formally, for all $\bm m=(m_1,\dots,m_M)\in\N^M$,
\begin{equation}
    \omega_{\bm n}^{\bm m,\ge}\ge\omega_{\bm n}^{(\max_im_i)\bm1,\ge}\quad\text{and}\quad\omega_{\bm n}^{\bm m,\le}\le\omega_{\bm n}^{(\max_im_i)\bm1,\le}.
\end{equation}
Finally, we show that both ways of defining the hierarchies are equivalent:

\begin{lemma}\label{lem:equiv}
For all $m\in\N$,
\begin{equation}\label{eq:interupper}
    \omega_{\bm n}^{m,\ge}\ge\omega_{\bm n}^{m\bm1,\ge}\ge\omega_{\bm n}^{Mm,\ge},
\end{equation}
and
\begin{equation}\label{eq:interlower}
    \omega_{\bm n}^{m,\le}\le\omega_{\bm n}^{m\bm1,\le}\le\omega_{\bm n}^{Mm,\le}.
\end{equation}
\end{lemma}

\begin{proof} For all $\bm k\in\N^M$ and all $m\in\N$ we have
\begin{equation}
    |\bm k|\le m\quad\Rightarrow\quad\bm k\le m\bm1=(m,\dots,m)\quad\Rightarrow\quad|\bm k|\le Mm.
\end{equation}
We thus obtain the corresponding inclusions between sets of $M$-variate polynomials:
\begin{enumerate*}[label=(\roman*)] \item$M$-variate polynomials of degree less or equal to $m$ have all their monomials lower than $m\bm1$, and \item all $M$-variate polynomials with monomials lower than $m\bm1$ have degree less or equal to $Mm$. \end{enumerate*}
Hence,
\begin{equation}\label{eq:inclusion}
    \mathcal R_{m,+}(\R_+^M)\overset{(\text{i})}{\subset}\mathcal R_{m\bm1,+}(\R_+^M)\overset{(\text{ii})}{\subset}\mathcal R_{Mm,+}(\R_+^M).
\end{equation}
In particular, \refprog{upperSDPn_multiapp} is a relaxation of~\refprog{upperSDPnalt1_multiapp} which is itself a relaxation of $(\text{SDP}_{\bm n}^{Mm,\ge})$, and \refprog{lowerSDPn_multiapp} is a restriction of~\refprog{lowerSDPnalt1_multiapp} which is itself a restriction of $(\text{SDP}_{\bm n}^{Mm,\le})$.
\end{proof}

\noindent This result implies that the two versions of the hierarchies of relaxations are interleaved (Eq.~\eqref{eq:interupper}), and that the two versions of the hierarchies of restrictions are also interleaved (Eq.~\eqref{eq:interlower}). As such, for any bound obtained with one version of the hierarchy at some fixed level, a better bound can be obtained with the other version at some other level. While this means that the hierarchies are equivalent, note that in practice it may be simpler to solve numerically the version where the parameter space is smaller.

\subsection{Convergence of the multimode hierarchies}
\label{sec:app_multiCV}
\setcounter{theorem}{7}

For $\bm n=(n_1,\dots,n_M)\in\N^M$, the sequences $(\omega_{\bm n}^{m,\ge})_{m\ge|\bm n|}$ and $(\omega_{\bm n}^{m\bm1,\ge})_{m\ge\max_in_i}$ (resp.\ $(\omega_{\bm n}^{m,\le})_{m\ge|\bm n|}$ and $(\omega_{\bm n}^{m\bm1,\le})_{m\ge\max_in_i}$) are decreasing (resp.\ increasing) sequences, lower bounded (resp.\ upper bounded) by $\omega_{\bm n}^{L^2}$ (resp. $\omega_{\bm n}^{\mathcal S})$). Hence, these sequences are converging. 

We show in what follows that $(\omega_{\bm n}^{m\bm1,\ge})_{m\ge\max_in_i}$ (resp.\ $(\omega_{\bm n}^{m\bm1,\le})_{m\ge\max_in_i}$) converges to $\omega_{\bm n}^{L^2}$ (resp. $\omega_{\bm n}^{\mathcal S}$). By Lemma~\ref{lem:equiv}, this implies that the sequence $(\omega_{\bm n}^{Mm,\ge})_{m\ge\max_in_i}$ (resp.\ $(\omega_{\bm n}^{Mm,\le})_{m\ge\max_in_i}$) also converges to $\omega_{\bm n}$. Since this is a subsequence of the converging sequence $(\omega_{\bm n}^{m,\ge})_{m\ge|\bm n|}$ (resp.\ $(\omega_{\bm n}^{m,\le})_{m\ge|\bm n|}$), it implies that the sequence $(\omega_{\bm n}^{m,\ge})_{m\ge|\bm n|}$ (resp.\ $(\omega_{\bm n}^{m,\le})_{m\ge|\bm n|}$) also converges to $\omega_{\bm n}^{L^2}$ (resp. $\omega_{\bm n}^{\mathcal S}$).

With similar proofs to the single-mode case using multi-index notations, we obtain the following result:

\begin{theorem}[Generalisation of Theorem~\ref{th:RHLaguerre}]
Let $\bm\mu=(\mu_{\bm k})_{\bm k\in\N^M}\in\R^{\N^M}$. Then, $\bm\mu$ is the sequence of Laguerre moments $\int_{\R_+^M}\mathcal L_{\bm k}(\bm x)d\mu(\bm x)$ of a non-negative distribution $\mu$ supported on $\R_+^M$ if and only if $\forall m\in\N,\forall g\in\mathcal R_{m,+}(\R_+^M),\;\braket{f_{\bm\mu},g}\ge0$.
\end{theorem}

\noindent The proof of this theorem is identical to the univariate case, with the use of Riesz--Haviland theorem over $\R_+^M$~\cite{haviland1936momentum} rather than $\R_+$. 

With Eq.~\eqref{eq:inclusion}, the proof of convergence of the multimode hierarchy of upper bounds is then obtained directly from its single-mode counterpart using multi-index notations:

\begin{theorem}[Generalisation of Theorem~\ref{th:upperCV}]\label{th:upperCVmulti}
The decreasing sequence of optimal values $\omega^{m\bm1,\ge}_{\bm n}$ of~\refprog{upperSDPnalt1_multiapp} converges to the optimal value $\omega_{\bm n}^{L^2}$ of \refprog{LP_multiapp}:
\begin{equation}
    \lim_{m\rightarrow+\infty}\omega^{m\bm1,\ge}_{\bm n}=\omega_{\bm n}^{L^2}.
\end{equation}
\end{theorem}

\noindent With Lemma~\ref{lem:equiv}, we also obtain
\begin{equation}
    \lim_{m\rightarrow+\infty}\omega^{m,\ge}_{\bm n}=\omega_{\bm n}^{L^2}.
\end{equation}

On the other hand, the proof of convergence of the single-mode hierarchy of lower bounds crucially exploits analytical feasible solutions of the programs~\refprog{lowerSDPn} in order to obtain two results:

\begin{itemize}
    \item Strong duality between programs~\refprog{lowerSDPn} and~\refprog{lowerDSDPn} (Theorem~\ref{th:sdlower}).
    \item The fact that the feasible set of~\refprog{lowerDSDPn} is compact with coefficients bounded independently of $m$ (Eq.~\eqref{eq:mucompact}). 
\end{itemize}

\noindent In what follows, we generalise these two results to the multimode setting by obtaining multimode analytical feasible solutions from products of single-mode ones. 

\begin{lemma}\label{lem:productsol}
For $m,n\in\mathbb N$ with $m\ge n$, let $Q(m,n)\in\text{Sym}_{m+1}$ and $\bm F(m,n)=(F_k(m,n))_k\in\R^{m+1}$ be feasible solutions of~\refprog{lowerSDPn}. Let $\bm m=(m_1,\dots,m_M)\in\N^M$ and $\bm n=(n_1,\dots,n_M)\in\N^M$ with $\bm m\ge\bm n$. Let $Q:=Q(m_1,n_1)\otimes\dots\otimes Q(m_M,n_M)\in\text{Sym}_{\pi_{\bm m}}$ and $\bm F=(F_{\bm k})_{\bm k\le\bm m}\in\R^{\pi_{\bm m}}$, where for all $\bm k=(k_1,\dots,k_M)\le\bm m$, $F_{\bm k}:=\prod_{i=1}^MF_{k_i}(m_i,n_i)$. Then, $(Q,\bm F)$ is a feasible solution of \refprog{lowerSDPnalt_multiapp}. Moreover, if $(Q(m_i,n_i),\bm F(m_i,n_i))$ is strictly feasible for all $i=1,\dots,M$ then $(Q,\bm F)$ is a strictly feasible solution of \refprog{lowerSDPnalt_multiapp}.
\end{lemma}

\begin{proof}
With the notations of the Lemma, we show the feasibility of $(Q,\bm F)$ (resp.\ strict feasibility).
We have $Q\succeq0$, $F_{\bm k}\ge0$ (resp.\ $Q\succ0$, $F_{\bm k}>0$) for all $\bm k\le\bm m$, and $Q_{\bm i\bm j}=\prod_{p=1}^MQ_{i_pj_p}(m_p,n_p)$ for all $\bm i=(i_1,\dots,i_M)\le\bm m$ and $\bm j=(j_1,\dots,j_M)\le\bm m$. Hence, for all $\bm r=(r_1,\dots,r_M)\le2\bm m$,
\begin{equation}
    \begin{aligned}
        \sum_{\bm i+\bm j=\bm r}Q_{\bm i\bm j}&=\sum_{i_1+j_1=r_1,\dots,i_M+j_M=r_M}\prod_{p=1}^MQ_{i_pj_p}(m_p,n_p)\\
        &=\prod_{p=1}^M\sum_{i_p+j_p=r_p}Q_{i_pj_p}(m_p,n_p).
    \end{aligned}
\end{equation}
In particular, if $\bm r\neq2\bm l$ for all $\bm l\le\bm m$, then at least one coefficient $r_p$ is odd, and the corresponding sum gives $0$ since $(Q(m_p,n_p),\bm F(m_p,n_p))$ is feasible for $(\text{SDP}_{n_p}^{m_p,\le})$. In that case, $\sum_{\bm i+\bm j=\bm r}Q_{\bm i\bm j}=0$. Otherwise, for all $\bm l=(l_1,\dots,l_M)\le\bm m$,
\begin{equation}
    \begin{aligned}
        \sum_{\bm i+\bm j=2\bm l}Q_{\bm i\bm j}&=\prod_{p=1}^M\sum_{i_p+j_p=2l_p}Q_{i_pj_p}(m_p,n_p)\\
        &=\prod_{p=1}^M\sum_{k_p\ge l_p} \frac{(-1)^{k_p+l_p}}{l_p!}\binom{k_p}{l_p}F_{k_p}(m_p,n_p)\\
        &=\sum_{l_1\le k_1\le m_1,\dots,l_M\le k_M\le m_M}\prod_{p=1}^M \frac{(-1)^{k_p+l_p}}{l_p!}\binom{k_p}{l_p}F_{k_p}(m_p,n_p)\\
        &=\sum_{\bm k\ge\bm l} \frac{(-1)^{|\bm k|+|\bm l|}}{\bm l!}\binom{\bm k}{\bm l}F_{\bm k},
    \end{aligned}
\end{equation}
where we used the feasibility of $(Q(m_p,n_p),\bm F(m_p,n_p))$ in the second line. Finally, 
\begin{equation}
    \begin{aligned}
        \sum_{\bm k\le\bm m}F_{\bm k}&=\sum_{k_1\le m_1,\dots,k_M\le m_M}\prod_{i=1}^MF_{k_i}(m_i,n_i)\\
        &=\prod_{i=1}^M\sum_{k_i=0}^{m_i}F_{k_i}(m_i,n_i)\\
        &=1,
    \end{aligned}
\end{equation}
since $\sum_{k=0}^mF_{k}(m,n)=1$ for all $m,n\in\N$ with $m\ge n$. This shows that $(Q,\bm F)$ is a feasible solution of~\refprog{lowerSDPnalt_multiapp} (resp.\ strictly feasible).
\end{proof}

\noindent A direct consequence of this construction is the following result:

\begin{theorem}[Generalisation of Theorem~\ref{th:sdlower}]\label{th:sdlowermulti}
Strong duality holds between the programs \refprog{lowerSDPnalt_multiapp} and \refprog{lowerDSDPnalt_multiapp}.
\end{theorem}

\begin{proof}
The proof of Theorem~\ref{th:sdlower} gives a strictly feasible solution $(Q(m,n),\bm F(m,n))$ of \refprog{lowerSDPn} for all $m\ge n$. By Lemma~\ref{lem:productsol}, the program \refprog{lowerSDPnalt_multiapp} thus has a strictly feasible solution, for all $\bm m\ge\bm n$. By Slater condition, this implies that strong duality holds between the programs \refprog{lowerSDPnalt_multiapp} and \refprog{lowerDSDPnalt_multiapp}.
\end{proof}

\noindent In particular, strong duality holds between the programs \refprog{lowerSDPnalt1_multiapp} and \refprog{lowerDSDPnalt1_multiapp}. Note that the multimode generalisation of Theorem~\ref{th:sdupper} is a direct consequence of Theorem~\ref{th:sdlowermulti}:

\begin{theorem}[Generalisation of Theorem~\ref{th:sdupper}]\label{th:sduppermulti}
Strong duality holds between the programs \refprog{upperSDPnalt_multiapp} and \refprog{upperDSDPnalt_multiapp}.
\end{theorem}

\begin{proof}
the strictly feasible solution of \refprog{lowerSDPnalt_multiapp} derived in the proof of Theorem~\ref{th:sdlowermulti} yields a strictly feasible solution for \refprog{upperSDPnalt_multiapp}. 
With Slater condition, this shows again that strong duality holds between the programs \refprog{upperSDPnalt_multiapp} and \refprog{upperDSDPnalt_multiapp}.
\end{proof}

\noindent In particular, strong duality holds between the programs \refprog{upperSDPnalt1_multiapp} and \refprog{upperDSDPnalt1_multiapp}.

We recall the following definition from the main text: for all $n\in\mathbb N$, $\bm F^n=(F_k^n)_{k\in\N}\in\R^{\N}$ where\\
$\bullet$ if $n$ is even:
\begin{equation}
        F_k^n:=\begin{cases}\frac1{2^n}\binom k{\frac k2}\binom{n-k}{\frac{n-k}2}&\text{when }k\le n, k\text{ even},\\0&\text{otherwise},\end{cases}
\end{equation}
$\bullet$ if $n$ is odd:
\begin{equation}
        F_k^n:=\begin{cases} \frac1{2^n}\frac{\binom n{\floor{\frac n2}}\binom{\floor{\frac n2}}{\floor{\frac k2}}^2 }{\binom nk},&\text{when }k\le n,\\0&\text{otherwise}. \end{cases}
\end{equation}
Let us define, for all $\bm n=(n_1,\dots,n_M)\in\mathbb N^m$, $\bm F^{\bm n}=(F_{\bm k}^{\bm n})_{k\in\N^M}\in\R^{\N^M}$ where
\begin{equation}\label{eq:defFbm}
    F_{\bm k}^{\bm n}:=\begin{cases}
      \prod_{i=1}^MF_{k_i}^{n_i} &\text{when }\bm k\le\bm n, \\
      0 &\text{otherwise.}
      \end{cases}
\end{equation}
By~\eqref{eq:boundFnn}, for all $n\in\mathbb N$, $F_n^n\ge\frac1{n+1}$, so for all $\bm n=(n_1,\dots,n_M)\in\mathbb N^M$,
\begin{equation}\label{eq:boundFnnmulti}
    F_{\bm n}^{\bm n}\ge\frac1{\pi_{\bm n}}.
\end{equation}
Like in the single-mode case, the program~\refprog{lowerSDPnalt_multiapp} is equivalent to
\leqnomode
\begin{flalign*}
    \label{prog:lowerSDPnalt_multiapp2}
    \tag*{(SDP$_{\bm n}^{\bm m,\le}$)}\hspace{2.8cm}\left\{
        \begin{aligned}
            & & & \Sup{\bm F\in\R^{\pi_{\bm m}}} F_{\bm n}\\
            & \text{subject to} & & \sum_{\bm k\le\bm m} F_{\bm k} =  1\\
            &\text{and} & & \forall\bm k\le\bm m, \quad F_{\bm k} \geq  0\\
            &\text{and} & & \sum_{\bm k\le\bm m}F_{\bm k}\mathcal L_{\bm k}\in\mathcal R_{\bm m,+}(\R_+^M),
        \end{aligned}
    \right. &&
\end{flalign*}
\reqnomode
with the dual program given by
\leqnomode
\begin{flalign*}
   \label{prog:lowerDSDPnalt_multiapp2}
    \tag*{$(\text{D-SDP}^{\bm m,\leq}_{\bm n})$}
    \hspace{2.8cm}\left\{
        \begin{aligned}
            & & & \Inf{y,\bm\mu\in\R\times\R^{\pi_{\bm m}}} y\\
            & \text{subject to} & & y\ge1+\mu_{\bm n}\\
            & \text{and} & & \forall\bm k\le\bm m,\bm k\neq\bm n,\quad y\ge\mu_{\bm k}\\
            & \text{and} & & \forall g\in\mathcal R_{\bm m,+}(\R_+^M),\left\langle \sum_{\bm k\le\bm m}\mu_{\bm k}\mathcal L_{\bm k},g\right\rangle\ge0,
        \end{aligned}
    \right. &&
\end{flalign*}
\reqnomode
for all $\bm m\ge\bm n$. Moreover, adding the condition $\mu_{\bm k}\le1$ for all $\bm k\le\bm m$ does not change the optimal value of the program. We enforce this condition in what follows.
With Lemma~\ref{lem:feasible} and Lemma~\ref{lem:productsol}, we thus obtain the following result:

\begin{lemma}[Generalisation of Lemma~\ref{lem:feasible}]\label{lem:feasiblemulti}
For all $\bm m\ge\bm n$, $\bm F^{\bm m}$ (defined in Eq.~\eqref{eq:defFbm}) is a feasible solution of~\refprog{lowerSDPnalt_multiapp2}.
\end{lemma}

\noindent In particular, for all $\bm m\in\N^M$, $\sum_{\bm k\le\bm m}F_{\bm k}^{\bm m}\mathcal L_{\bm k}\in\mathcal R_{\bm m,+}(\R_+^M)$. For $\bm m,\bm n\in\N^M$ with $\bm m\ge\bm n$, let $\bm\mu\in\R^{\pi_{\bm m}}$ be a feasible solution of~\refprog{lowerDSDPnalt_multiapp2}. Then, for all $\bm l\le\bm m$
\begin{equation}
    \left\langle\sum_{\bm k\le\bm m}\mu_{\bm k}\mathcal L_{\bm k},\sum_{\bm k\le\bm l}F_{\bm k}^{\bm l}\mathcal L_{\bm k}\right\rangle\ge0,
\end{equation}
so that
\begin{equation}
    \sum_{\bm k\le\bm l}\mu_{\bm k}F_{\bm k}^{\bm l}\ge0.
\end{equation}
Hence, for all $\bm l\ge\bm m$,
\begin{equation}
    \begin{aligned}
        \mu_{\bm l}&\ge-\frac1{F_{\bm l}^{\bm l}}\sum_{\substack{\bm k\le\bm l\\\bm k\neq\bm l}}\mu_{\bm k}F_{\bm k}^{\bm l}\\
        &\ge-\frac1{F_{\bm l}^{\bm l}}\sum_{\substack{\bm k\le\bm l\\\bm k\neq\bm l}}F_{\bm k}^{\bm l}\\
        &=1-\frac1{F_{\bm l}^{\bm l}}\\
        &\ge1-\pi_{\bm l},
    \end{aligned}
\end{equation}
where we used $F_{\bm l}^{\bm l}>0$ in the first line, $\mu_{\bm k}\le1$ and $F_{\bm k}^{\bm l}\ge0$ in the second line, $\sum_{\bm k\le\bm l}F_{\bm k}^{\bm l}=1$ in the third line, and Eq.~\eqref{eq:boundFnnmulti} in the last line. 

With these additional results, the proof of convergence of the multimode hierarchy of lower bounds \refprog{lowerSDPnalt1_multiapp} is then obtained directly from its single-mode counterpart using multi-index notations:

\begin{theorem}[Generalisation of Theorem~\ref{th:lowerCV}]\label{th:lowerCVmulti}
The increasing sequence of optimal values $\omega^{m\bm1,\le}_{\bm n}$ of \refprog{lowerSDPnalt1_multiapp} converges to the optimal value $\omega_{\bm n}^\mathcal S$ of \refprog{LPS}:
\begin{equation}
    \lim_{m\rightarrow+\infty}\omega^{m\bm1,\le}_{\bm n}=\omega_{\bm n}^\mathcal S.
\end{equation}
\end{theorem}

\noindent With Lemma~\ref{lem:equiv}, we also obtain
\begin{equation}
    \lim_{m\rightarrow+\infty}\omega^{m,\le}_{\bm n}=\omega_{\bm n}^\mathcal S.
\end{equation}
Like in the single-mode case, Theorem~\ref{th:upperCVmulti} and Theorem~\ref{th:lowerCVmulti} imply strong duality between the linear programs:

\begin{theorem}[Generalisation of Theorem~\ref{th:sdLP}]\label{th:sdLPmulti}
Strong duality holds between the programs \refprog{LP_multiapp} and \refprog{DLP_multiapp} and between programs \refprog{LP_multiappS} and \refprog{DLP_multiappS}.
\end{theorem}


\newpage 
\section{Bounds on threshold values of several witness for \texorpdfstring{$n=3$}{n=3} }
\label{sec:weightedwitness_n3}
Below, we provide tables of numerical upper bounds and lower bounds obtained on the threshold values for witnesses of the form:
\begin{equation}
    \hat \Omega_{(a_1,a_2,a_3)} = a_1 \ket 1\! \bra 1 + a_2 \ket 2\! \bra 2 + a_3 \ket 3\! \bra 3
\end{equation}
where $\forall i \in \{1,2,3\},\, 0 \leq a_i \leq 1$ and $\max_i a_i =1$. We focused on these particular witnesses for experimental considerations as it is challenging to obtain fidelities with higher Fock states. We vary each $a_i$ from $0$ to $1$ with a step of $0.1$.

\begin{table}[ht]
\small
\centering
\begin{tabular}{||ccc|cc||ccc|cc||ccc|cc||}
$a_1$ & $a_2$ & $a_3$ & $\omega_{\bm a}^{\leq}$ & $\omega_{\bm a}^{\geq}$ & $a_1$ & $a_2$ & $a_3$ & $\omega_{\bm a}^{\leq}$ & $\omega_{\bm a}^{\geq}$ & $a_1$ & $a_2$ & $a_3$ & $\omega_{\bm a}^{\leq}$ & $\omega_{\bm a}^{\geq}$ \\ \midrule
1.0 & 0.0 & 0.0 & 0.500 & 0.529 & 1.0 & 0.3 & 0.8 & 0.589 & 0.590 & 1.0 & 0.7 & 0.5 & 0.705 & 0.715 \\
1.0 & 0.0 & 0.1 & 0.500 & 0.529 & 1.0 & 0.3 & 0.9 & 0.606 & 0.610 & 1.0 & 0.7 & 0.6 & 0.718 & 0.720 \\
1.0 & 0.0 & 0.2 & 0.500 & 0.529 & 1.0 & 0.3 & 1.0 & 0.626 & 0.633 & 1.0 & 0.7 & 0.7 & 0.735 & 0.738 \\
1.0 & 0.0 & 0.3 & 0.500 & 0.529 & 1.0 & 0.4 & 0.0 & 0.610 & 0.615 & 1.0 & 0.7 & 0.8 & 0.754 & 0.758 \\
1.0 & 0.0 & 0.4 & 0.500 & 0.529 & 1.0 & 0.4 & 0.1 & 0.610 & 0.615 & 1.0 & 0.7 & 0.9 & 0.774 & 0.781 \\
1.0 & 0.0 & 0.5 & 0.500 & 0.528 & 1.0 & 0.4 & 0.2 & 0.610 & 0.615 & 1.0 & 0.7 & 1.0 & 0.795 & 0.805 \\
1.0 & 0.0 & 0.6 & 0.500 & 0.529 & 1.0 & 0.4 & 0.3 & 0.610 & 0.615 & 1.0 & 0.8 & 0.0 & 0.739 & 0.751 \\
1.0 & 0.0 & 0.7 & 0.500 & 0.529 & 1.0 & 0.4 & 0.4 & 0.610 & 0.615 & 1.0 & 0.8 & 0.1 & 0.739 & 0.751 \\
1.0 & 0.0 & 0.8 & 0.500 & 0.529 & 1.0 & 0.4 & 0.5 & 0.610 & 0.615 & 1.0 & 0.8 & 0.2 & 0.739 & 0.751 \\
1.0 & 0.0 & 0.9 & 0.500 & 0.530 & 1.0 & 0.4 & 0.6 & 0.610 & 0.615 & 1.0 & 0.8 & 0.3 & 0.739 & 0.751 \\
1.0 & 0.0 & 1.0 & 0.500 & 0.563 & 1.0 & 0.4 & 0.7 & 0.614 & 0.615 & 1.0 & 0.8 & 0.4 & 0.739 & 0.751 \\
1.0 & 0.1 & 0.0 & 0.526 & 0.529 & 1.0 & 0.4 & 0.8 & 0.629 & 0.631 & 1.0 & 0.8 & 0.5 & 0.742 & 0.751 \\
1.0 & 0.1 & 0.1 & 0.526 & 0.529 & 1.0 & 0.4 & 0.9 & 0.648 & 0.653 & 1.0 & 0.8 & 0.6 & 0.759 & 0.761 \\
1.0 & 0.1 & 0.2 & 0.526 & 0.529 & 1.0 & 0.4 & 1.0 & 0.669 & 0.677 & 1.0 & 0.8 & 0.7 & 0.777 & 0.780 \\
1.0 & 0.1 & 0.3 & 0.526 & 0.529 & 1.0 & 0.5 & 0.0 & 0.640 & 0.649 & 1.0 & 0.8 & 0.8 & 0.796 & 0.801 \\
1.0 & 0.1 & 0.4 & 0.526 & 0.529 & 1.0 & 0.5 & 0.1 & 0.640 & 0.649 & 1.0 & 0.8 & 0.9 & 0.816 & 0.824 \\
1.0 & 0.1 & 0.5 & 0.526 & 0.529 & 1.0 & 0.5 & 0.2 & 0.640 & 0.649 & 1.0 & 0.8 & 1.0 & 0.837 & 0.848 \\
1.0 & 0.1 & 0.6 & 0.526 & 0.529 & 1.0 & 0.5 & 0.3 & 0.640 & 0.649 & 1.0 & 0.9 & 0.0 & 0.773 & 0.790 \\
1.0 & 0.1 & 0.7 & 0.526 & 0.529 & 1.0 & 0.5 & 0.4 & 0.640 & 0.649 & 1.0 & 0.9 & 0.1 & 0.773 & 0.790 \\
1.0 & 0.1 & 0.8 & 0.526 & 0.529 & 1.0 & 0.5 & 0.5 & 0.640 & 0.649 & 1.0 & 0.9 & 0.2 & 0.773 & 0.790 \\
1.0 & 0.1 & 0.9 & 0.526 & 0.528 & 1.0 & 0.5 & 0.6 & 0.640 & 0.649 & 1.0 & 0.9 & 0.3 & 0.773 & 0.790 \\
1.0 & 0.1 & 1.0 & 0.542 & 0.563 & 1.0 & 0.5 & 0.7 & 0.654 & 0.655 & 1.0 & 0.9 & 0.4 & 0.773 & 0.790 \\
1.0 & 0.2 & 0.0 & 0.552 & 0.555 & 1.0 & 0.5 & 0.8 & 0.671 & 0.673 & 1.0 & 0.9 & 0.5 & 0.783 & 0.790 \\
1.0 & 0.2 & 0.1 & 0.552 & 0.555 & 1.0 & 0.5 & 0.9 & 0.690 & 0.696 & 1.0 & 0.9 & 0.6 & 0.800 & 0.803 \\
1.0 & 0.2 & 0.2 & 0.552 & 0.555 & 1.0 & 0.5 & 1.0 & 0.711 & 0.719 & 1.0 & 0.9 & 0.7 & 0.818 & 0.822 \\
1.0 & 0.2 & 0.3 & 0.552 & 0.555 & 1.0 & 0.6 & 0.0 & 0.672 & 0.682 & 1.0 & 0.9 & 0.8 & 0.838 & 0.844 \\
1.0 & 0.2 & 0.4 & 0.552 & 0.555 & 1.0 & 0.6 & 0.1 & 0.672 & 0.682 & 1.0 & 0.9 & 0.9 & 0.858 & 0.867 \\
1.0 & 0.2 & 0.5 & 0.552 & 0.555 & 1.0 & 0.6 & 0.2 & 0.672 & 0.682 & 1.0 & 0.9 & 1.0 & 0.879 & 0.891 \\
1.0 & 0.2 & 0.6 & 0.552 & 0.555 & 1.0 & 0.6 & 0.3 & 0.672 & 0.683 & 1.0 & 1.0 & 0.0 & 0.809 & 0.830 \\
1.0 & 0.2 & 0.7 & 0.552 & 0.555 & 1.0 & 0.6 & 0.4 & 0.672 & 0.682 & 1.0 & 1.0 & 0.1 & 0.809 & 0.830 \\
1.0 & 0.2 & 0.8 & 0.552 & 0.555 & 1.0 & 0.6 & 0.5 & 0.672 & 0.682 & 1.0 & 1.0 & 0.2 & 0.809 & 0.830 \\
1.0 & 0.2 & 0.9 & 0.565 & 0.567 & 1.0 & 0.6 & 0.6 & 0.678 & 0.682 & 1.0 & 1.0 & 0.3 & 0.809 & 0.830 \\
1.0 & 0.2 & 1.0 & 0.584 & 0.594 & 1.0 & 0.6 & 0.7 & 0.694 & 0.696 & 1.0 & 1.0 & 0.4 & 0.809 & 0.830 \\
1.0 & 0.3 & 0.0 & 0.581 & 0.585 & 1.0 & 0.6 & 0.8 & 0.712 & 0.716 & 1.0 & 1.0 & 0.5 & 0.824 & 0.830 \\
1.0 & 0.3 & 0.1 & 0.581 & 0.585 & 1.0 & 0.6 & 0.9 & 0.732 & 0.739 & 1.0 & 1.0 & 0.6 & 0.841 & 0.845 \\
1.0 & 0.3 & 0.2 & 0.581 & 0.584 & 1.0 & 0.6 & 1.0 & 0.753 & 0.762 & 1.0 & 1.0 & 0.7 & 0.860 & 0.865 \\
1.0 & 0.3 & 0.3 & 0.581 & 0.584 & 1.0 & 0.7 & 0.0 & 0.705 & 0.716 & 1.0 & 1.0 & 0.8 & 0.879 & 0.887 \\
1.0 & 0.3 & 0.4 & 0.581 & 0.584 & 1.0 & 0.7 & 0.1 & 0.705 & 0.715 & 1.0 & 1.0 & 0.9 & 0.900 & 0.910 \\
1.0 & 0.3 & 0.5 & 0.581 & 0.584 & 1.0 & 0.7 & 0.2 & 0.705 & 0.715 & 1.0 & 1.0 & 1.0 & 0.922 & 0.934 \\
1.0 & 0.3 & 0.6 & 0.581 & 0.584 & 1.0 & 0.7 & 0.3 & 0.705 & 0.715 &     &     &     &       &       \\
1.0 & 0.3 & 0.7 & 0.581 & 0.584 & 1.0 & 0.7 & 0.4 & 0.705 & 0.715 &     &     &     &       &      
\end{tabular}
\label{tab:n3first}
\caption{Upper and lower bounds on the threshold values of witnesses of the form $\hat \Omega_{(1,a_2,a_3)}$.}
\end{table}

\begin{table}[ht]
\small
\centering
\begin{tabular}{||ccc|cc||ccc|cc||ccc|cc||}
$a_1$ & $a_2$ & $a_3$ & $\omega_{\bm a}^{\leq}$ & $\omega_{\bm a}^{\geq}$ & $a_1$ & $a_2$ & $a_3$ & $\omega_{\bm a}^{\leq}$ & $\omega_{\bm a}^{\geq}$ & $a_1$ & $a_2$ & $a_3$ & $\omega_{\bm a}^{\leq}$ & $\omega_{\bm a}^{\geq}$ \\ \midrule
0.0 & 1.0 & 0.0 & 0.500 & 0.546 & 0.3 & 1.0 & 0.4 & 0.594 & 0.603 & 0.6 & 1.0 & 0.8 & 0.767 & 0.779 \\
0.0 & 1.0 & 0.1 & 0.500 & 0.551 & 0.3 & 1.0 & 0.5 & 0.617 & 0.629 & 0.6 & 1.0 & 0.9 & 0.792 & 0.806 \\
0.0 & 1.0 & 0.2 & 0.500 & 0.546 & 0.3 & 1.0 & 0.6 & 0.642 & 0.655 & 0.6 & 1.0 & 1.0 & 0.817 & 0.849 \\
0.0 & 1.0 & 0.3 & 0.500 & 0.546 & 0.3 & 1.0 & 0.7 & 0.667 & 0.683 & 0.7 & 1.0 & 0.0 & 0.680 & 0.701 \\
0.0 & 1.0 & 0.4 & 0.517 & 0.546 & 0.3 & 1.0 & 0.8 & 0.693 & 0.712 & 0.7 & 1.0 & 0.1 & 0.680 & 0.702 \\
0.0 & 1.0 & 0.5 & 0.543 & 0.559 & 0.3 & 1.0 & 0.9 & 0.719 & 0.741 & 0.7 & 1.0 & 0.2 & 0.680 & 0.702 \\
0.0 & 1.0 & 0.6 & 0.569 & 0.588 & 0.3 & 1.0 & 1.0 & 0.747 & 0.789 & 0.7 & 1.0 & 0.3 & 0.691 & 0.702 \\
0.0 & 1.0 & 0.7 & 0.597 & 0.618 & 0.4 & 1.0 & 0.0 & 0.570 & 0.603 & 0.7 & 1.0 & 0.4 & 0.710 & 0.714 \\
0.0 & 1.0 & 0.8 & 0.625 & 0.650 & 0.4 & 1.0 & 0.1 & 0.570 & 0.600 & 0.7 & 1.0 & 0.5 & 0.729 & 0.735 \\
0.0 & 1.0 & 0.9 & 0.654 & 0.682 & 0.4 & 1.0 & 0.2 & 0.579 & 0.603 & 0.7 & 1.0 & 0.6 & 0.750 & 0.759 \\
0.0 & 1.0 & 1.0 & 0.683 & 0.740 & 0.4 & 1.0 & 0.3 & 0.600 & 0.607 & 0.7 & 1.0 & 0.7 & 0.772 & 0.782 \\
0.1 & 1.0 & 0.0 & 0.505 & 0.551 & 0.4 & 1.0 & 0.4 & 0.622 & 0.631 & 0.7 & 1.0 & 0.8 & 0.794 & 0.806 \\
0.1 & 1.0 & 0.1 & 0.505 & 0.551 & 0.4 & 1.0 & 0.5 & 0.644 & 0.653 & 0.7 & 1.0 & 0.9 & 0.817 & 0.830 \\
0.1 & 1.0 & 0.2 & 0.505 & 0.551 & 0.4 & 1.0 & 0.6 & 0.667 & 0.679 & 0.7 & 1.0 & 1.0 & 0.842 & 0.870 \\
0.1 & 1.0 & 0.3 & 0.517 & 0.551 & 0.4 & 1.0 & 0.7 & 0.692 & 0.705 & 0.8 & 1.0 & 0.0 & 0.722 & 0.744 \\
0.1 & 1.0 & 0.4 & 0.542 & 0.556 & 0.4 & 1.0 & 0.8 & 0.717 & 0.733 & 0.8 & 1.0 & 0.1 & 0.722 & 0.744 \\
0.1 & 1.0 & 0.5 & 0.567 & 0.583 & 0.4 & 1.0 & 0.9 & 0.743 & 0.763 & 0.8 & 1.0 & 0.2 & 0.722 & 0.744 \\
0.1 & 1.0 & 0.6 & 0.593 & 0.610 & 0.4 & 1.0 & 1.0 & 0.769 & 0.813 & 0.8 & 1.0 & 0.3 & 0.724 & 0.744 \\
0.1 & 1.0 & 0.7 & 0.619 & 0.641 & 0.5 & 1.0 & 0.0 & 0.604 & 0.632 & 0.8 & 1.0 & 0.4 & 0.741 & 0.745 \\
0.1 & 1.0 & 0.8 & 0.647 & 0.671 & 0.5 & 1.0 & 0.1 & 0.604 & 0.638 & 0.8 & 1.0 & 0.5 & 0.760 & 0.764 \\
0.1 & 1.0 & 0.9 & 0.675 & 0.701 & 0.5 & 1.0 & 0.2 & 0.610 & 0.632 & 0.8 & 1.0 & 0.6 & 0.779 & 0.786 \\
0.1 & 1.0 & 1.0 & 0.704 & 0.750 & 0.5 & 1.0 & 0.3 & 0.629 & 0.635 & 0.8 & 1.0 & 0.7 & 0.800 & 0.809 \\
0.2 & 1.0 & 0.0 & 0.519 & 0.556 & 0.5 & 1.0 & 0.4 & 0.650 & 0.658 & 0.8 & 1.0 & 0.8 & 0.822 & 0.832 \\
0.2 & 1.0 & 0.1 & 0.519 & 0.556 & 0.5 & 1.0 & 0.5 & 0.672 & 0.681 & 0.8 & 1.0 & 0.9 & 0.844 & 0.857 \\
0.2 & 1.0 & 0.2 & 0.522 & 0.556 & 0.5 & 1.0 & 0.6 & 0.694 & 0.705 & 0.8 & 1.0 & 1.0 & 0.867 & 0.891 \\
0.2 & 1.0 & 0.3 & 0.544 & 0.556 & 0.5 & 1.0 & 0.7 & 0.717 & 0.729 & 0.9 & 1.0 & 0.0 & 0.765 & 0.787 \\
0.2 & 1.0 & 0.4 & 0.567 & 0.578 & 0.5 & 1.0 & 0.8 & 0.742 & 0.756 & 0.9 & 1.0 & 0.1 & 0.765 & 0.787 \\
0.2 & 1.0 & 0.5 & 0.592 & 0.605 & 0.5 & 1.0 & 0.9 & 0.767 & 0.784 & 0.9 & 1.0 & 0.2 & 0.765 & 0.787 \\
0.2 & 1.0 & 0.6 & 0.617 & 0.633 & 0.5 & 1.0 & 1.0 & 0.793 & 0.831 & 0.9 & 1.0 & 0.3 & 0.765 & 0.787 \\
0.2 & 1.0 & 0.7 & 0.643 & 0.660 & 0.6 & 1.0 & 0.0 & 0.641 & 0.665 & 0.9 & 1.0 & 0.4 & 0.774 & 0.787 \\
0.2 & 1.0 & 0.8 & 0.669 & 0.690 & 0.6 & 1.0 & 0.1 & 0.641 & 0.665 & 0.9 & 1.0 & 0.5 & 0.791 & 0.795 \\
0.2 & 1.0 & 0.9 & 0.697 & 0.722 & 0.6 & 1.0 & 0.2 & 0.641 & 0.665 & 0.9 & 1.0 & 0.6 & 0.810 & 0.815 \\
0.2 & 1.0 & 1.0 & 0.725 & 0.770 & 0.6 & 1.0 & 0.3 & 0.660 & 0.665 & 0.9 & 1.0 & 0.7 & 0.829 & 0.836 \\
0.3 & 1.0 & 0.0 & 0.542 & 0.575 & 0.6 & 1.0 & 0.4 & 0.679 & 0.685 & 0.9 & 1.0 & 0.8 & 0.850 & 0.860 \\
0.3 & 1.0 & 0.1 & 0.542 & 0.575 & 0.6 & 1.0 & 0.5 & 0.700 & 0.708 & 0.9 & 1.0 & 0.9 & 0.872 & 0.883 \\
0.3 & 1.0 & 0.2 & 0.550 & 0.575 & 0.6 & 1.0 & 0.6 & 0.722 & 0.732 & 0.9 & 1.0 & 1.0 & 0.894 & 0.913 \\
0.3 & 1.0 & 0.3 & 0.572 & 0.580 & 0.6 & 1.0 & 0.7 & 0.744 & 0.756 &     &     &     &       &      
\end{tabular}
\label{tab:n3second}
\caption{Upper and lower bounds on the threshold values of witnesses of the form $\hat \Omega_{(a_1,1,a_3)}$.}
\end{table}

\begin{table}[ht]
\small
\centering
\begin{tabular}{||ccc|cc||ccc|cc||ccc|cc||}
$a_1$ & $a_2$ & $a_3$ & $\omega_{\bm a}^{\leq}$ & $\omega_{\bm a}^{\geq}$ & $a_1$ & $a_2$ & $a_3$ & $\omega_{\bm a}^{\leq}$ & $\omega_{\bm a}^{\geq}$ & $a_1$ & $a_2$ & $a_3$ & $\omega_{\bm a}^{\leq}$ & $\omega_{\bm a}^{\geq}$ \\ \midrule
0.0 & 0.0 & 1.0 & 0.377 & 0.428 & 0.3 & 0.4 & 1.0 & 0.529 & 0.547 & 0.6 & 0.8 & 1.0 & 0.734 & 0.752 \\
0.0 & 0.1 & 1.0 & 0.401 & 0.437 & 0.3 & 0.5 & 1.0 & 0.561 & 0.583 & 0.6 & 0.9 & 1.0 & 0.775 & 0.793 \\
0.0 & 0.2 & 1.0 & 0.429 & 0.452 & 0.3 & 0.6 & 1.0 & 0.594 & 0.621 & 0.7 & 0.0 & 1.0 & 0.463 & 0.521 \\
0.0 & 0.3 & 1.0 & 0.458 & 0.474 & 0.3 & 0.7 & 1.0 & 0.628 & 0.657 & 0.7 & 0.1 & 1.0 & 0.490 & 0.522 \\
0.0 & 0.4 & 1.0 & 0.488 & 0.510 & 0.3 & 0.8 & 1.0 & 0.667 & 0.694 & 0.7 & 0.2 & 1.0 & 0.521 & 0.542 \\
0.0 & 0.5 & 1.0 & 0.519 & 0.533 & 0.3 & 0.9 & 1.0 & 0.707 & 0.733 & 0.7 & 0.3 & 1.0 & 0.555 & 0.580 \\
0.0 & 0.6 & 1.0 & 0.550 & 0.572 & 0.4 & 0.0 & 1.0 & 0.426 & 0.480 & 0.7 & 0.4 & 1.0 & 0.594 & 0.615 \\
0.0 & 0.7 & 1.0 & 0.583 & 0.605 & 0.4 & 0.1 & 1.0 & 0.451 & 0.485 & 0.7 & 0.5 & 1.0 & 0.635 & 0.653 \\
0.0 & 0.8 & 1.0 & 0.616 & 0.642 & 0.4 & 0.2 & 1.0 & 0.480 & 0.501 & 0.7 & 0.6 & 1.0 & 0.676 & 0.692 \\
0.0 & 0.9 & 1.0 & 0.649 & 0.678 & 0.4 & 0.3 & 1.0 & 0.511 & 0.528 & 0.7 & 0.7 & 1.0 & 0.717 & 0.733 \\
0.1 & 0.0 & 1.0 & 0.390 & 0.440 & 0.4 & 0.4 & 1.0 & 0.543 & 0.567 & 0.7 & 0.8 & 1.0 & 0.758 & 0.774 \\
0.1 & 0.1 & 1.0 & 0.414 & 0.448 & 0.4 & 0.5 & 1.0 & 0.576 & 0.602 & 0.7 & 0.9 & 1.0 & 0.800 & 0.815 \\
0.1 & 0.2 & 1.0 & 0.441 & 0.464 & 0.4 & 0.6 & 1.0 & 0.610 & 0.637 & 0.8 & 0.0 & 1.0 & 0.475 & 0.535 \\
0.1 & 0.3 & 1.0 & 0.471 & 0.486 & 0.4 & 0.7 & 1.0 & 0.649 & 0.674 & 0.8 & 0.1 & 1.0 & 0.504 & 0.535 \\
0.1 & 0.4 & 1.0 & 0.501 & 0.519 & 0.4 & 0.8 & 1.0 & 0.689 & 0.713 & 0.8 & 0.2 & 1.0 & 0.537 & 0.562 \\
0.1 & 0.5 & 1.0 & 0.532 & 0.549 & 0.4 & 0.9 & 1.0 & 0.729 & 0.752 & 0.8 & 0.3 & 1.0 & 0.576 & 0.597 \\
0.1 & 0.6 & 1.0 & 0.565 & 0.588 & 0.5 & 0.0 & 1.0 & 0.439 & 0.494 & 0.8 & 0.4 & 1.0 & 0.617 & 0.633 \\
0.1 & 0.7 & 1.0 & 0.597 & 0.623 & 0.5 & 0.1 & 1.0 & 0.464 & 0.497 & 0.8 & 0.5 & 1.0 & 0.658 & 0.672 \\
0.1 & 0.8 & 1.0 & 0.631 & 0.659 & 0.5 & 0.2 & 1.0 & 0.494 & 0.513 & 0.8 & 0.6 & 1.0 & 0.700 & 0.713 \\
0.1 & 0.9 & 1.0 & 0.665 & 0.696 & 0.5 & 0.3 & 1.0 & 0.525 & 0.544 & 0.8 & 0.7 & 1.0 & 0.742 & 0.754 \\
0.2 & 0.0 & 1.0 & 0.402 & 0.453 & 0.5 & 0.4 & 1.0 & 0.558 & 0.583 & 0.8 & 0.8 & 1.0 & 0.784 & 0.808 \\
0.2 & 0.1 & 1.0 & 0.426 & 0.460 & 0.5 & 0.5 & 1.0 & 0.592 & 0.617 & 0.8 & 0.9 & 1.0 & 0.825 & 0.849 \\
0.2 & 0.2 & 1.0 & 0.454 & 0.476 & 0.5 & 0.6 & 1.0 & 0.631 & 0.655 & 0.9 & 0.0 & 1.0 & 0.488 & 0.549 \\
0.2 & 0.3 & 1.0 & 0.484 & 0.499 & 0.5 & 0.7 & 1.0 & 0.671 & 0.693 & 0.9 & 0.1 & 1.0 & 0.518 & 0.549 \\
0.2 & 0.4 & 1.0 & 0.515 & 0.530 & 0.5 & 0.8 & 1.0 & 0.711 & 0.732 & 0.9 & 0.2 & 1.0 & 0.559 & 0.578 \\
0.2 & 0.5 & 1.0 & 0.547 & 0.570 & 0.5 & 0.9 & 1.0 & 0.752 & 0.773 & 0.9 & 0.3 & 1.0 & 0.600 & 0.613 \\
0.2 & 0.6 & 1.0 & 0.579 & 0.604 & 0.6 & 0.0 & 1.0 & 0.451 & 0.506 & 0.9 & 0.4 & 1.0 & 0.642 & 0.652 \\
0.2 & 0.7 & 1.0 & 0.613 & 0.640 & 0.6 & 0.1 & 1.0 & 0.477 & 0.509 & 0.9 & 0.5 & 1.0 & 0.684 & 0.693 \\
0.2 & 0.8 & 1.0 & 0.646 & 0.676 & 0.6 & 0.2 & 1.0 & 0.507 & 0.527 & 0.9 & 0.6 & 1.0 & 0.726 & 0.736 \\
0.2 & 0.9 & 1.0 & 0.685 & 0.714 & 0.6 & 0.3 & 1.0 & 0.540 & 0.564 & 0.9 & 0.7 & 1.0 & 0.768 & 0.779 \\
0.3 & 0.0 & 1.0 & 0.414 & 0.466 & 0.6 & 0.4 & 1.0 & 0.573 & 0.598 & 0.9 & 0.8 & 1.0 & 0.810 & 0.821 \\
0.3 & 0.1 & 1.0 & 0.439 & 0.505 & 0.6 & 0.5 & 1.0 & 0.613 & 0.635 & 0.9 & 0.9 & 1.0 & 0.852 & 0.866 \\
0.3 & 0.2 & 1.0 & 0.467 & 0.488 & 0.6 & 0.6 & 1.0 & 0.653 & 0.673 &     &     &     &       &       \\
0.3 & 0.3 & 1.0 & 0.497 & 0.513 & 0.6 & 0.7 & 1.0 & 0.693 & 0.712 &     &     &     &       &      
\end{tabular}
\label{tab:n3third}
\caption{Upper and lower bounds on the threshold values of witnesses of the form $\hat \Omega_{(a_1,a_2,1)}$.}
\end{table}

\end{document}